\documentclass[11pt]{article}

\usepackage{amssymb}
\usepackage{graphicx}
\usepackage{algorithm,algorithmic}
\usepackage{amsmath}
\usepackage{amsthm}
\usepackage{amsopn}
\usepackage{bbm}
\usepackage{url}
\usepackage{multicol}
\usepackage{lscape}
\usepackage[figuresright]{rotating}
\usepackage{natbib}
\usepackage{authblk}

\newtheorem{theorem}{Theorem}[section]
\newtheorem{lemma}[theorem]{Lemma}
\newtheorem{proposition}[theorem]{Proposition}
\newtheorem{corollary}[theorem]{Corollary}
\newtheorem{definition}[theorem]{Definition}
\newtheorem{property}[theorem]{Property}
\newcommand{\minus}[1]{{-#1}}
\newcommand{\indicator}[1]{\mathbbm{1}{\left[ {#1} \right] }}

\def \BR {{\mathcal{BR}}}
\def \NE {{\mathcal{NE}}}
\def \G {{\mathcal{G}}}
\def \KL {{\text{KL}}}
\def \P {{\mathbf{P}}}
\def \argmax {{\arg\max}}

\providecommand{\abs}[1]{\lvert#1\rvert} 
\providecommand{\norm}[1]{\lVert#1\rVert} 
\providecommand{\myvec}[1]{\mathbf{#1}}
\providecommand{\R}{\mathbb{R}} 
\providecommand{\x}{\mathbf{x}} 
\providecommand{\y}{\mathbf{y}}

\title{On Influence, Stable Behavior, and the Most Influential Individuals in Networks: A Game-Theoretic Approach}

\author[1]{Mohammad T. Irfan\thanks{This work was done while the author was a PhD student at Stony Brook University. Email: \texttt{mirfan@bowdoin.edu}.}}
\affil[1]{Department of Computer Science\\
      Bowdoin College\\
      Brunswick, ME 04011}

\author[2]{Luis E. Ortiz\thanks{Corresponding author. Email: \texttt{leortiz@cs.stonybrook.edu}, Phone: 631-632-1805, Fax: 631-632-8334.}}

\affil[2]{Department of Computer Science\\
      Stony Brook University\\
      Stony Brook, NY 11794}

\begin{document}

\maketitle

\begin{abstract}
We introduce a new approach to the study of influence in strategic
settings where the action of an individual depends on that of others
in a network-structured way. We propose \emph{influence games (IGs)}
as a \emph{game-theoretic (GT)} model of the behavior of a large but finite networked
population. IGs allow \emph{both} positive and negative
\emph{influence factors}, permitting reversals in behavioral choices.
We embrace \emph{pure-strategy Nash equilibrium (PSNE)}, an
important solution concept in non-cooperative game theory, to
formally define the \emph{stable outcomes} of an IG and to predict
potential outcomes without explicitly considering intricate dynamics. 
We address an important problem in network influence, the
identification of the \emph{most influential individuals}, and approach it
algorithmically using PSNE
computation. \emph{Computationally}, we provide (a) 
complexity characterizations of various problems on IGs; (b) efficient algorithms for
several special cases and heuristics for hard cases; and (c) approximation algorithms, with provable guarantees, for the
problem of identifying the most influential individuals. \emph{Experimentally}, we evaluate our
approach using both synthetic IGs and real-world settings of general
interest, each corresponding to a separate branch of the
U.S. Government.
\emph{Mathematically,} we connect IGs to important GT models:
\emph{potential and polymatrix games}.

\flushleft
{\small
\textit{Keywords:} Computational Game Theory, Social Network Analysis, Influence in Social Networks, Nash Equilibrium, Computational Complexity
}
\end{abstract}


\section{Introduction}
\label{sec:intro}

The influence of an entity on its peers is a commonly noted phenomenon in both online and real-life social networks. 
%
In fact, there is growing scientific evidence that suggests that influence can induce behavioral changes among the entities in a network. For example, recent work in medical social sciences posits the intriguing hypothesis that many of our behavioral aspects, such as
smoking~\citep{christakisandfowler08}, obesity~\citep{christakisandfowler07}, and even happiness~\citep{fowlerandchristakis08} are contagious within a social network.

Regardless of the specific problem addressed, 
the underlying system under study in that research exhibits several core features. First, it is often very \emph{large and complex}, with many entities exhibiting different behaviors and interactions. Second, the \emph{network structure of complex interactions} is central. Third, the
\emph{directions and strengths of local influences} are highlighted as 
very relevant to the global behavior of
the system as a whole. Fourth, the view that the behavioral choice of
an individual in the network is potentially \emph{strategic}, given that one's choice depends on the choices made by one's peers.  


The prevalence of systems and problems like the ones just described, combined with the obvious issue of often limited control over individuals, raises immediate, broad, difficult, and longstanding policy questions: e.g., {\em Can we achieve a
desired goal, such as reducing the level of smoking or controlling obesity via targeted, minimal
interventions in a system?  How do we optimally allocate our often
limited resources to achieve the largest impact in such systems?}

Clearly, these issues are not exclusive to obesity, smoking or happiness; similar issues arise in a
large variety of settings: drug use, vaccination, crime networks, security, marketing, markets, the
economy, and public policy-making and regulations, and even
congressional voting!~\footnote{For the last example, consider (one
  of) the so-called
  ``debt-ceiling crisis'' in the
U.S. in $2011$, which received a lot of national and international
media attention. One may argue that the last-minute deal between 
democratic and republican senators, which helped avoid an
``undesirable situation,'' bears a clear signature of influence among
the senators in a strategic setting. Moreover, it is hard not to view the bipartisan
``gang-of-six'' senators, specifically chosen to work out a solution
to that crisis, 
as a type of \emph{intervention} in such large, complex systems; that
is, interventions via the formation of groups that would not naturally arise otherwise.} 
%
The work reported in this paper is in large part motivated by such
questions/settings and their broader implication. 

We begin by providing a brief and informal description of our approach
to influence in networks. In the next section, we place our approach
within the context of existing literature.

\subsection{Overview of Our Model of Influence}
Consider a social network where each \emph{individual} has a \emph{binary} choice of
\emph{action or behavior}, denoted by $-1$ and $1$. Let us represent this
network as a \emph{directed} graph, where each node represents an individual. Each node of this graph has a
\emph{threshold level}, which can be positive, negative, or zero; and
the threshold levels of all the nodes are not required to be the
same. Each arc of this graph is weighted by an \emph{influence
  factor}, which signifies the level of \emph{direct} influence the tail node of that arc has on the head node. Again, the influence factors can be positive, negative, or zero and are not required to be the same (i.e., symmetric) between two nodes.

Given such a network, our model specifies the \emph{best response} of
a node (i.e., what action it should choose) with respect to the
actions chosen by the other nodes. The best response of a node is to
adopt the action $1$ if the \emph{total influence} on it exceeds its
threshold and $-1$ if the opposite happens. (In the case of a tie, the
node is \emph{indifferent} between choosing $1$ and $-1$, i.e., either
would be its best response.) Here, we calculate the total influence on
a node as follows. First, sum up the incoming influence factors on the node from the ones who have adopted the action $1$. Second, sum up those influence factors that are coming in from the ones who have adopted $-1$. Finally, subtract the second sum from the first to get the total influence on that node.

Clearly, in a network with $n$ nodes, there are $2^n$ possible
\emph{joint actions}, resulting from the action choice of each
individual node. Among all these joint actions, we call the ones where
every node has chosen its best response to everyone else a
\emph{pure-strategy Nash equilibria (PSNE)}. We use PSNE to
mathematically model the \emph{stable outcomes} that such a networked
system could support. 

\subsection{Overview of the Most-Influential-Nodes Problem}

We formulate the most-influential-nodes problem \emph{with respect to
  a goal of interest}. The ``goal of interest'' indirectly determines what we call the \emph{desired stable
  outcome(s)}. Unlike the mainstream literature on the most-influential-nodes problem \citep{kleinberg07}, maximizing the \emph{spread} of a
particular behavior is \emph{not} our objective. Rather, \emph{the desired
stable outcome(s) resulting from the goal of interest is what determines
our computational objective}. In addition, while for \emph{some
instances} of our general formulation of the most-influential-nodes
problem in our context our goal may be some \emph{stable}
outcome with the largest number of nodes adopting a particular behavior, our
solution concept abstracts away and does not rely on the dynamics or
the so-called
``diffusion'' process by which such a ``spread of behavior'' happens.

Roughly speaking, in our approach, we consider a set of individuals
$S$ in a network to be \emph{most influential}, with respect to a set
of desired stable outcomes that satisfy a particular goal of interest,
if $S$ is the \emph{most preferred} subset among all those that 
satisfy the following condition: were the individuals in $S$ to choose
the behavior $\x_S$ prescribed to them by a desired stable outcome $\x
\equiv (\x_S, \x_{-S})$ which achieves the goal of interest,
then the \emph{only} stable outcome of the system that remains consistent with their choices $\x_S$ is $\x$ itself. 

Said more intuitively, once the nodes in the most influential set $S$
follow the behavior $\x_S$ prescribed to them by a desired stable
outcome $\x$ achieving the goal of interest, they become collectively ``so influential'' that their behavior ``forces'' every other individual to a unique choice of behavior! 
Our proposed concept of the most influential individuals is illustrated in Figure~\ref{fig:infl_example} with a very simple example. 

Now, there could be many different sets $S$ that satisfy the above
condition. For example, $S$ could consist of all the individuals,
which might not be what we want. To account for this, we also specify
a preference over all subsets of individuals. While this preference
function could in principle be arbitrary, a natural example would be
one that prefers a set $S$ the with the minimum cardinality.

\begin{figure}[h]
\centering
\includegraphics[width=4in]{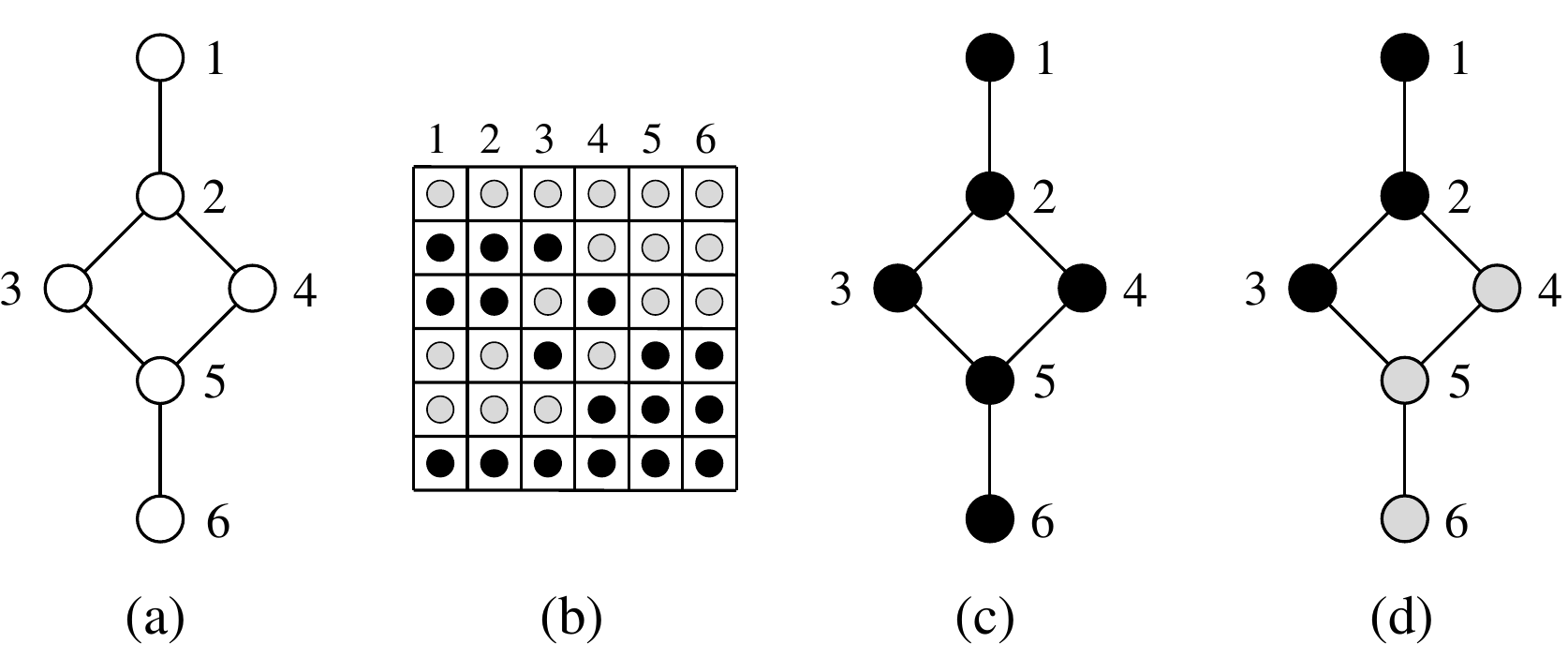}
\caption{{\bf Illustration of our approach to influence in networks.} 
Each node has a binary choice of behavior, $\{-1,+1\}$, and, \emph{in this instance},
 wants to behave like the majority of its neighbors (and is indifferent if there is a tie). 
 We adopt {\em pure-strategy Nash equilibrium} (formally defined later), abbreviated as PSNE, as the notion of stable outcome.
 The network is shown in (a) and the enumeration of PSNE (a row for each PSNE, where black denotes node's behavior $1$, gray $-1$) in (b). 
We want to achieve the objective of every node choosing 1 (desired outcome).
Selecting the set of nodes $\{1, 2, 3\}$ and assigning these nodes behavior prescribed by the desired outcome (i.e., 1 for each) lead to two consistent stable outcomes of the system, shown in (c) and (d).
Thus, $\{1, 2, 3\}$ cannot be a most-influential set of nodes in our setting. On the other hand, selecting $\{1, 6\}$ and assigning these nodes behavior 1 lead to the desired outcome as the unique stable outcome remaining. Therefore, $\{1, 6\}$ is a most-influential set, even though these two nodes are at the fringes of the network. Furthermore, note that $\{1, 6\}$ is not most influential in the diffusion setting, since it does not maximize the spread of behavior 1. Also note that $\{3, 4\}$ is another most-influential set in our setting.
(Of course, we study a much richer class of games in this paper than
this particular instance.)
}
\label{fig:infl_example}
\end{figure}

\subsection{Our Contributions}
Our major contributions include 
\begin{enumerate}
\item a new approach, grounded in non-cooperative game theory, to the study of
influence in networks where individuals exhibit
\emph{strategic behavior}; 
\item \textit{influence games} as a new class of graphical games to
  model the behavior of individuals in networks, including establishing connections to potential games and polymatrix games;
\item a theoretical and empirical study of various computational aspects
of influence games, including an algorithm for identifying the most
influential individuals; and
\item the application of our approach to two real-world settings:
  the \emph{U.S. Supreme Court} and the \emph{U.S. Senate}.
\end{enumerate}

\section{Background}

``Influence'' in social networks, however defined, has been a subject of both theoretical
and empirical studies for decades (see, e.g.,
\citet{wassermanandfaust94} and the references therein). Although our
focus is primarily on computation,  the roots of our model go
back to early literature in sociology on collective behavior, and more
recent literature on collective action. In this section, we will place
our model in the context of the relevant literature from sociology,
economics, and computer science.~\footnote{To avoid a long aside, and in the interest of keeping
  the flow of our presentation, we refer the interested reader to
  Appendix~\ref{app:rev} for a more detailed exposition on the connection of
our model to the century-old study of collective behavior and collective action
in sociology.}


\subsection{Connection to Collective Action in Sociology}
Although our approach may seem close to the rational calculus models
of collective action (see Appendix~\ref{app:conn} for a detailed
description), particularly to~\cite{granovetter78}'s threshold models, our
objective is very much different from that of collective action
theory. The focus of collective action theory in sociology is to
explain \textit{how} individual behavior in a group leads to
collective outcomes. For example, 
\cite{Schelling_dynamic_71,Schelling78}'s models explain how different
distributions of the level of tolerances of individuals lead to
residential segregations of different properties. 
\cite{berk_1974} explains how a compromise (such as placing a barricade) evolves within a mixture of rational individuals of different predispositions (militants vs. moderates). \cite{granovetter78} shows how a little perturbation in the distribution of thresholds can possibly lead to a completely different collective outcome. In short, \emph{explaining} collective social phenomena is at the heart of all these studies. 

While such an explanation is a scientific pursuit of utmost importance, our focus is rather on an engineering approach to \textit{predicting} stable behavior in a networked population setting. Our approach is not to go through the fine-grained details of a process, such as forward recursion, which is often plagued with problems when the sociomatrix contains negative elements. Instead, we adopt the notion of PSNE to define stable outcomes. Said differently, the \textit{path} to an equilibrium is not what we focus on; rather, it is the prediction of the equilibrium itself that we focus on. 

We next justify our approach in the context of rational calculus models.

\subsubsection{PSNE as the solution concept} 
Nash equilibrium is the most
central solution concept in non-cooperative game theory. The
fundamental aspect of a Nash equilibrium is stability. That is, at a
Nash equilibrium, no ``player'' has any incentive to deviate from
it, assuming the others do not deviate either. As a result, Nash equilibrium is very often a natural choice  to
mathematically model stable outcomes of a complex system. 

In this
work, we adopt PSNE, one particular form of Nash equilibrium,  to model the stable behavioral
outcomes of a networked system of influence.~\footnote{As a side note,
  interested readers can find more on the interpretation of Nash
  equilibrium, including its underlying concept of stability, in
  \citep[Ch. 7]{luce1957} and \citep[Ch. 2--4]{osborne2003}. Although
  we use the notion of PSNE in its original form, we note
  here that Nash equilibrium has been refined in various ways with
  particular applications in mind
  \citep[Ch. 2]{laffont1995}. Furthermore, as mentioned above, we
  purposefully avoid the question of \emph{how} a PSNE is reached and
  focus on the prediction of PSNE instead. For the interested reader,
  there is a large body of literature on how a Nash equilibrium may be reached, although there is no general consensus on this topic \citep{fudenberg1998}.}

\subsubsection{Abstraction of fine-grained details}
Sociologists have recorded minute details of various collective action scenarios in order to substantiate their theories with empirical accounts. 
However, in the application scenarios that we are interested in, such
as strategic interaction in the U.S. Congress and the U.S. Supreme
Court, very little details can be obtained about \textit{how} a
collective outcome emerges. 

For example, the Budget Control Act of 2011 was passed by 74--26 votes
in the U.S. Senate on August 2, 2011, ending a much debated debt-ceiling crisis. Despite intense media coverage,  it would be difficult, if not impossible, to give an accurate account of how this agreement on debt-ceiling was reached. Even if there were an exact account of every conversation and every negotiation that had taken place, it would be extremely challenging to translate such a subjective account into a mathematically defined process, let alone learning the parameters and computing stable outcomes of such a complex model.

\subsubsection{Influence games as a less-restrictive model}
Typical models of dynamics used in the existing
literature almost always impose restrictions or assumptions to keep the
model simple enough to permit analytical solutions, or
to facilitate algorithmic analysis. For example, as mentioned above,
the forward recursion process implicitly assumes that the sociomatrix
does not have negative elements. The hope is that the essence of the
general phenomena that 
one wants to capture with 
the model remains even after imposing such
restrictions. 

In our case, by using the concept of PSNE to abstract dynamical
processes, we can deal with rich models without having to impose some
of the same restrictions, and at the same time, we can capture
equilibria beyond the ones captured by a simple model of dynamics. In
particular, our model captures any equilibrium that the process of forward recursion converges to (with any initial configuration); but in addition, our model can capture equilibria that the forward recursion process cannot.

\subsubsection{Focus on practical applications}
Although our model of influence games is grounded in non-cooperative
game theory, the way we apply it to real-world settings such as the
U.S. Congress is deeply rooted in modern AI ~\citep{russellandnorvig09}. One of the distinctive features
of the field of AI is that it is able to build useful tools, often without
gaining the full scientific knowledge of \emph{how} a system
works. 
For example, even though we do not have a full understanding or model
of the exact process by which humans' perform inferences and reason
when performing tasks as simple as playing different types of parlor
games, modern AI has been able to devise 
systems that perform better than humans on the same tasks, often by a
considerable margin~\citep{russellandnorvig09}.

The scientific question of
\emph{how} humans perform 
inferences or reason
is of course very important,
but, in our view, which may be the prevailing view, the modern focus of AI in general is to \emph{engineer solutions} that would
serve our purpose without necessarily having to explain the specific
and intricate details of complex physical
phenomena often found in the real world~\citep{russellandnorvig09}. Of course, under the right conditions, AI does sometimes help experts
understand physical phenomena too, although not necessarily purposefully,
by suggesting effective insights and potentially useful directions, as
obtained from modern AI-based models or systems.~\footnote{One relatively recent example is the research by psychologist Alison
Gopnik, which suggests that young children, even 2-year-olds, perform
Bayesian (belief) inference while learning from the environment~\citep{gopnik_bayes}.} 

In short, we propose an AI-based approach, including AI-inspired models and algorithms, to build a
computational
tool for predicting the behavior of large,
networked populations. Our approach does not model the complex
behavioral dynamics in the network, but abstracts it via PSNE. This
allows us to deal with a rich set of models and concentrate our efforts on the prediction of stable outcomes. 

\subsection{Connection to Literature on Most-Influential Nodes}
\label{sec:diffusion}
To date, the study of influence in a network, by both economists \citep{morris00,chwe_ca} and  computer scientists \citep{kleinberg07,even_dar07}, has been rooted in rational calculus models of behavior. Their approach to connecting individual behavior to collective outcome is mostly by adopting the process of forward recursion \citep[p. 1426]{granovetter78}, which is often employed in studying diffusion of innovations \citep[p. 168]{granovetter_diffusion}. As a result, the term ``contagion'' in these settings has a rational connotation contrary to the early sociology literature on collective behavior, where ``contagion'' or ``social contagion'' alludes to irrational and often hysteric nature of the individuals in a crowd \citep{park_burgess_1921,blumer_1939}. The computational question of identifying the most influential nodes in a network \citep{kleinberg07}, originally posed by Domingos and Richardson~\citep{domingos01}, has also been studied using forward recursion within the context of rational-calculus models. 

In the traditional setting as described by~\cite{kleinberg07}, ``cascade'' or ``diffusion
models,'' each node behaves in one of these two ways---it either
adopts a new behavior or does not, and initially, none of the nodes
adopts the new behavior. Given a number $k$, their formulation of the
most-influential-nodes problem in the cascade model within the
diffusion setting asks us to select a set of $k$ nodes
such that the \emph{spread} of the new behavior is maximized by the
selected nodes being the initial adopters of the new behavior. (Note
that in their setting, the set of initial adopters, some of whom may
have thresholds greater than $0$, are \textit{externally} selected in
order to set off the forward recursion process, whereas in
Granovetter's setting, the initial adopters must have a threshold of
$0$.) The most-influential-nodes question in the cascade or diffusion
settings typically concerns \emph{infinite} graphs
\cite[p. 615]{kleinberg07}, such as Morris' local interaction games
\cite[p. 59]{morris00}. In contrast, we concern ourselves with large
but \emph{finite} graphs here.

The notion of ``most influential nodes'' used in this paper is just 
different than the one traditionally used within the diffusion setting
because it seeks to address problems in different settings (i.e., fully/strictly
strategic) and
achieve generally different objectives (i.e., desired stable
outcomes relative to whatever the goal of interest happens to be). If anything, our approach, by
taking a strictly game-theoretic perspective, may \emph{complement} the traditional line
of work based on diffusion, although even that is not really our main
intent. In our view, these are clearly disparate, non-competing
approaches. Yet, despite these fundamental differences in objectives, settings,
models, problem formulations, solution concepts and simply general approach, in the following
paragraphs, because of the high level of interest in diffusion models
within the most recent
(theoretical) computer science literature,
we still attempt to 
briefly mention some possible points of contrast between the
typical approach to identifying the most influential nodes in the
diffusion setting, as
described by~\cite{kleinberg07}, and ours.~\footnote{The reader should bear in mind
that our intent is not to argue
for or against one approach over the other, as each has their own pros
and cons.}

\subsubsection{Stability of outcomes}
A subtle aspect of diffusion models is that each node in the network behaves as an independent agent. Any observed influence that a node's neighbors impose on the node is the result of the same node's ``rational'' or ``natural'' response to the neighbors' behavior.
Thus, in many cases, it would be desirable that the solution to the most-influential-nodes problem lead us to a stable outcome of the system, in which each node's behavior is a best response to the neighbors' behavior. 
However, if we select a set of nodes with the goal of maximizing the
spread of the new behavior, then some of the selected nodes may end up
``unhappy'' being the initial adopters of the new behavior, with
respect to their neighbors' final behavior at the end of forward recursion. 
For example, a selected node's best behavioral response could be \textit{not} adopting the new behavior after all.  

We believe that in some cases it is more natural to require that the desired
final state of the system, such as, e.g., the state in which the
maximum \emph{possible} number of individuals
adopting a particular behavior, 
be stable (i.e., everyone is ``happy'' with their behavioral response).


\subsubsection{On arbitrary influence factors: positive and negative}
In general, to address the question of finding the most influential nodes, the forward recursion process has been modeled as ``monotonic.'' 
(Here, a monotonic process refers to the setting where once an agent adopts the new behavior, it cannot go back.) 
Even in more recent work on influence maximization and minimization \citep{budak2011,chen2011,he2012}, the influence factors (or weights on the edges) are defined to be non-negative.~\footnote{From the description of \cite{he2012}'s model, it may at first seem that they are allowing both positive and negative influence weights. However, that is not the case. Their terminology of ``positive'' and ``negative'' weights on the edges refers to the positive and negative cascades that are defined in their context. The values of the weights are non-negative.}
If we think of an application such as reducing the incidence of smoking or obesity, then a model 
that allows a ``change of mind'' based on the response of the immediate neighborhood may make more sense. 
Thus, a notable contrast between the traditional treatment of the
most-influential-nodes problem and ours is that we do not restrict the
influence among the nodes of the network to non-negative numbers. 

In fact, in many applications, both positive and negative influence factors may exist in the same problem instance. Take the U.S. Congress as an example: senators belonging to the same party may have non-negative influence factors on each other (as usually perceived from voting instances on legislation issues), but one senator may (and often does) have a negative influence on another belonging to a different party. 
While generalized versions of threshold models that allow ``reversals'' have been derived in the social science literature, 
to the best of our knowledge, there is no substantive work on the most-influential-nodes problem in that context.~\footnote{Note that this is
different from recent work in diffusion settings on the notion of positive and negative \emph{opinions}~\citep{chen2011,budak2011},
which in our case would correspond to differing choices of
behavior, and the notion of simultaneously occurring
positive and negative ``cascades''~\citep{he2012}, used to
model, e.g., the spread of two competing ideas.}

\subsubsection{Abstraction of intricate dynamics}
Finally, the traditional approach to the most-influential-nodes
problem emphasizes modeling the complex dynamics of interactions among
nodes as a way to a final answer: a set of the most influential
nodes. In fact, our model is inspired by the same threshold models
used for the traditional approach. However, as mentioned earlier,
\emph{our emphasis is not on the dynamics of interactions, but on the
  stable outcomes in a game-theoretic setting.} By doing this, we seek to capture significant, basic, and core \emph{strategic} aspects
of complex interaction in networks that naturally appear in many
real-world problems (e.g., identifying the most influential senators
in the U.S. Congress). Of course, we recognize the importance of the
dynamics of interactions when modeling and studying problems of influence at a fine level of detail. Yet, we believe that our approach can still capture significant aspects of the problem even at the coarser level of ``steady-state'' or stable outcomes.

%

\subsubsection{A brief note on submodularity}
\label{sec:submod}

Submodularity plays a central role in the (algorithmic) analysis of traditional
diffusion models in computer science. A deeper implication of allowing negative influence factors in the
traditional diffusion models is that we
cannot restrict them to \emph{submodular} influence ``spread''
functions.~\footnote{Given a set of initial adopters, the influence
  spread function gives the number of nodes that would ultimately
  adopt the new behavior in a diffusion setting.} This lack of
submodularity of the ``influence spread function'' voids the highly heralded theoretical guarantees of
simple, greedy-selection
approximation algorithms commonly used for the problem of selecting
most influential nodes in the cascade model.

In general, \emph{the role that submodularity plays in our
  approach, if any, is not evident}. In
fact, it is unclear what the equivalent or analogous concept of the ``influence spread
function'' is in our setting, given that we do not explicitly consider the
dynamics by which stable outcomes may arise. We realize that this may
be a point of contention, but we delay a more substantive formal
discussion on this matter until after we formally define the problem of identifying most-influential
nodes in our setting in Section~\ref{sec:probform}.

\subsection{Related Work in Game Theory}

Other researchers have used similar game-theoretic notions of ``influential
individuals'' in specific contexts. Particularly close to ours is the work of~\citet{healandkunreuther03,healandkunreuther05,healandkunreuther06,healandkunreuther07}, \citet{kunreutherandheal03}, \citet{kunreutherandmichel-kerjan09}, \citet{ballesteretal04,ballesteretal06}, and~\citet{kearnsandortiz04}. 

Also, our interest is on identifying an ``optimal'' set of influential nodes for a variety of optimality
criteria, depending on the particular context of interest. For instance, we may prefer the set of influential individuals of
minimal size. Such a preference is similar to the concept of ``minimal
critical coalitions'' in the work of~\citet{healandkunreuther03,healandkunreuther05,healandkunreuther06,healandkunreuther07} and~\citet{kunreutherandmichel-kerjan09}. 

\emph{Mechanism design} is a core area in game theory whose main focus
is to ``engineer'' games, by changing the existing
underlaying game or by creating a new one, whose stable outcomes (i.e.,
equilibria) achieve a desired objective~\citep{nisan07_ch}. Although
our notion of the most influential individuals is also defined with
respect to a desired objective, our approach is conceptually very
different. We are not interested in changing, defining, or engineering
a new system---\emph{the system is what it is}. Rather, our
interest is in altering the \emph{behavior} within the \emph{same system} so as to ``lead'' or
``tip'' the system to achieve a desired stable outcome. 

In the next section, we will formally define our model of influence
game and our notion of ``most influential nodes'' in a network. We
will also establish connections to well-studied classes of game models in
game theory: polymatrix and potential games.

\section{Influence Games} 
Inspired by threshold models~\citep{granovetter78},
we introduce {\em network influence games\/}---in brief,
\emph{influence games}---as a model of influence in large networked
populations.~\footnote{As pointed out by a reviewer,
  the term \emph{influence games} has been used in various
  contexts. For example, there exist \emph{cooperative influence
    games} \citep{cooperativeinfluencegames}, \emph{political
    influence games} \citep{congleton2000,becker1983}, \emph{judicial
    influence games} \citep{bardsley2005}, and \emph{dynamic influence
    games} \citep{levy2011}, to name a few. Our \emph{network influence
    games} appear to
  be different from the above games and does not seek to generalize any
  of them.}
Even though 
the model falls within the general class of graphical
games~\citep{kearns01}, a distinctive feature of the instance of
influence games we concentrate on most in this paper, the \emph{linear
influence game}, formally defined in Section~\ref{sec:lig}, is a very compact, parametric representation. 
Our emphasis is on the problem of computing stable
outcomes of systems of influence and identifying influential agents within the network relative to a particular objective. 

\subsection{Our Game-Theoretic Model of Behavior}
We will first formalize {\em influence games} as a model of behavior.
Let $n$ be the number of individuals in a population. 
For simplicity, we restrict our attention to the case of
binary behavior, a common assumption in most of the work in
this area. Thus, 
 $x_i \in \{-1,1\}$ denotes the {\em behavior}
of individual $i$, where $x_i = 1$ indicates that $i$
``adopts'' a particular behavior and $x_i = -1$ indicates $i$
``does not adopt'' the behavior.
Some examples of behavior of this kind are supporting a particular political measure,
candidate or party; holding a particular view or belief; vaccinating against a particular disease; installing
virus protection software (and keeping it up-to-date);  acquiring
fire/home insurance; becoming overweight; taking up
smoking; becoming a criminal or participating in criminal activity; among many others.

First, we denote by $f_i : \{-1,1\}^{n-1} \to \R$ the function that
quantifies the {\em influence\/} of other individuals on $i$.

\begin{definition}[Payoff Function]
\label{def:payoff}
In influence games, 
we define the {\em payoff\/} function $u_i : \{-1,1\}^n \to \R$
quantifying the preferences of each player $i$ as $u_i(x_i,\x_{-i}) \equiv x_i  f_i(\x_{-i})$, where $\x_{-i}$ denotes the vector of all joint-actions excluding that of $i$.
\end{definition}

Using the above definition and notations, we define an influence game as follows.

\begin{definition}[Influence Game]
An influence game ${\G}$ is defined by a set of $n$ players and for each player $i$, a set of actions $\{-1,1\}$ and a payoff  function $u_i : \{-1,1\}^n \to \R$.  
\end{definition}

Next, we characterize the stable outcomes of an influence games. We start with the definition of the best-response correspondence.

\begin{definition}[Best-Response Correspondence]
Given $\x_{-i} \in \{-1,1\}^{n-1}$, the \emph{best-response correspondence} $\BR_i^{\G} : \{-1,1\}^{n-1} \to 2^{\{-1,1\}}$ of a player $i$ of an influence game $\G$ is defined as follows.
\begin{align*}
&\BR_i^{\G}(\x_{-i}) \equiv \argmax_{x_i \in \{-1,1\}} u_i(x_i,\x_{-i}). 
\end{align*}
\end{definition}

Therefore, for all individuals $i$ and any possible
behavior $\x_{-i} \in \{-1,1\}^{n-1}$ of the other individuals in the
population, the {\em best-response behavior\/} $x_i^*$ of individual $i$ to the
behavior $\x^*_{-i}$ of others satisfies
\begin{align*}
f_i(x^*_{-i}) > 0 &\implies x_i^* = 1,\\
f_i(x^*_{-i}) < 0 &\implies x_i^* = -1, \text{ and }\\
f_i(x^*_{-i}) = 0 &\implies x_i^* \in \{-1,1\}.
\end{align*}
Informally, ``positive influences'' lead an individual to adopt the behavior,
while ``negative influences'' lead the individual to ``reject'' the
behavior; the individual is indifferent if there is ``no influence.'' We formally characterize the {\em stable outcomes\/} of the system by the following notion of {\em pure-strategy Nash equilibria (PSNE)\/} of the corresponding influence game.

 \begin{definition}[Pure-Strategy Nash Equilibrium]
   A pure-strategy Nash equilibrium (PSNE) of an influence game $\G$ is a behavior assignment $\x^* \in \{-1,1\}^n$ that satisfies the following condition. Each player $i$'s behavior $x_i^*$ is a (simultaneous) best-response to the behavior $\x_{-i}^*$ of the rest. 
 \end{definition}
   
   We denote the set of all PSNE of the game $\G$ by 
\begin{align*}
\NE(\G) \equiv \{ \x^* \in \{-1,1\}^n \mid x_i^* \in \BR_i^{\G}(\x_{-i}^*) \text{ for all } i \}.
\end{align*}


\subsection{Most Influential Nodes: Problem Formulation}
\label{sec:probform}

In formulating the most-influential-nodes problem in a network, we
depart from the traditional model of diffusion and adopt influence
games as the model of strategic behavior among the nodes in the
network. We introduce two functions in our definition which we discuss in
more detail immediately after the problem formulation. One is what we
call the ``goal'' or ``objective function,'' denoted by $g$, which we
introduce as a way to formally
express that, in our approach, the notion of ``most influential'' is
\emph{relative to a specific goal or objective of interest}. The other, which
we call the ``set-preference function'' and denote by $h$, is a way to
choose among all sets of nodes that achieve our goal of interest.

\begin{definition}
\label{def:probform}
Let $\G$ be an influence game, $g : \{-1,1\}^n \times 2^{[n]} \to \R$ be the \emph{goal or objective function} mapping a joint-action and a subset of the players in $\G$ to a real number quantifying the general preferences over the space of joint-actions and players' subsets, and $h : 2^{[n]} \to \R$ be the \emph{set-preference function} mapping a subset of the players to a real number quantifying the \textit{a priori} preference over the space of players' subsets. Denote by $\mathcal{X}^*_g(S) \equiv \argmax_{\x \in \NE(\G)} g(\x,S)$ the \emph{optimal set of PSNE} of $\G$, with respect to $g$ and a fixed subset of players $S \subset [n]$. We say that a set of nodes/players $S^* \subset [n]$ in $\G$ is \emph{most influential} with respect to $g$ and $h$, if
\begin{align*}
S^* \in \argmax_{S \subset [n]} {h(S), \text{s.t.}, |\{ \x \in \NE(\G) \mid \x_S = \x^*_{S}, \x^* \in \mathcal{X}^*_g(S) \}| = 1 }.
\end{align*}
\end{definition}
As mentioned earlier, we can interpret the players in $S^*$ to be collectively so influential that they are able to restrict every other player's choice of action to a unique one: the action prescribed by some desired stable outcome $\mathbf{x^*}$.

An example of a goal function $g$ that captures the objective of
achieving a specific stable outcome $\x^* \in \NE(\G)$ is $g(\x,S)
\equiv \indicator{\x = \x^*}$. Another example that captures the
objective of achieving a stable outcome with the largest number of
individuals adopting the behavior is $g(\x,S) \equiv \sum_{i=1}^n
\frac{x_i + 1}{2}$, or equivalently, $g(\x,S) \equiv \sum_{i=1}^n
x_i$. Note that both of the functions just presented ignore the set
$S$. One alternative that does not is $g(\x,S) \equiv \sum_{i \in S}
t_i x_i - \sum_{i \notin S}
t_i x_i$, where the $t_i$'s reflect a weighted preference over
individual nodes, e.g., thus capturing 
interest in some ``weighted maximum set of adopters.''

A common example of the set-preference function $h$ that captures the
preference for sets of small cardinality is to simply define $h$ such
that $h(S) > h(S')$ iff $|S| < |S'|$. Similar to the last 
objective function described in the previous paragraph, an alternative
is $h(S) \equiv \sum_{i \in S} v_i - \sum_{i \notin S} v_i$, where the
$v_i$'s reflect a weighted preference over individual nodes in any
set $S$ that achieves the objective of interest.

\subsubsection{On the topic of submodularity}
\label{subsec:submod}

As briefly discussed in Section~\ref{sec:submod}, the computer-science community
is enthusiastic about the
submodularity of the so called ``influence spread function.'' One
reason for the intense interest is the important role submodularity plays in the
algorithmic design and analysis of methods to solve the
``most influential nodes'' problem as formulated within the diffusion
setting using the cascade model~\citep{kleinberg07}. Because our
approach abstracts away the dynamics, \emph{it
is not even clear what such ``influence spread function'' is in our
context, or even whether it could be meaningfully defined}. 

In this paragraph, we
will consider a potential definition and discuss its meaning
or implications, if any.
We first introduce some additional notation to simplify the
presentation. We denote the set involved in the uniqueness condition
in the last equation presented in the definition of the problem
formulation by $C_{g,\G}(S) \equiv \{ \x \in \NE(\G) \mid \x_S =
\x^*_{S}, \x^* \in \mathcal{X}^*_g(S) \}$. Perhaps the most direct attempt at defining a notion analogous to the
``influence spread function'' in the diffusion setting, which in our
context we will denote by $f_{g,h,\G}$, may be to let
$f_{g,h,\G}(S) \equiv h(S) - \lambda |C_{g,\G}(S)|$, for some constant
$\lambda > 0$. We can interpret such an $f_{g,h,\G}$ as trying to minimize the
PSNE consistent with assigning the nodes in $S$ according to some goal PSNE, while maximizing a general
preference over subsets as captured by $h$, modulo a ``penalization
constant'' $\lambda$. Yet, we cannot say anything
meaningful about such an ``influence spread function'' in general, even for the most
common instantiations of $g$ and $h$; except perhaps that it is at best
unclear to us that the question
of whether that function
is submodular makes sense, or whether it could even have a reasonable
answer. To start, one major reason for
our inability to say anything meaningful at this moment is
that we are unaware of any PSNE characterization result that would
apply to our setting. It seems to us that such characterizations would
be key in any study of the potential submodularity properties of that
``influence spread function'' $f_{g,h,\G}$ as we just defined it.

The study of other potential definitions of an ``influence spread
function'' in our setting, as well as their properties, is beyond the scope of this
paper. (In fact, we do not know of any other reasonable alternatives,
beyond simple variations of the function defined above.) More importantly, we leave open for future work the study of the potential \emph{relevance} that
submodularity may have within our approach, beyond the indirect connection
through the characteristics of certain classes of influence games
discussed in the next section, Section~\ref{SecComputation} (i.e., strategic complementarity and substitutability).
 

\subsection{Linear Influence Games (LIG)}
\label{sec:lig} 
A simple instantiation of the general
influence game model just described is the case of linear influences.
\begin{definition}
\label{def:lig}
In a \emph{linear influence game (LIG)}, the influence function of
each individual $i$ is defined
as
\(
\textstyle
f_{i}(\x_{-i}) \equiv \sum_{j \neq i} w_{ji} x_j - b_i 
\)
where for any other individual $j$, $w_{ji} \in \R$ is a {\em
weight\/} parameter quantifying the ``influence factor'' that $j$ has
on $i$, and $b_{i} \in \R$ is a {\em threshold\/} parameter for 
$i$'s level of ``tolerance'' for negative effects.
\end{definition}
It follows from Definition~\ref{def:payoff} that although the influence function of an LIG is linear, its payoff function is quadratic. Furthermore, the following argument shows that an LIG is a special type of graphical game in \emph{parametric} form. 
In general, the influence factors $w_{ji}$ induce a
directed graph, where nodes represent individuals, and therefore, we obtain a
graphical game having a \emph{linear} (in the number of edges) representation size, as opposed to the  \emph{exponential} (in the maximum degree of a node) representation size of general graphical games in normal form  \citep{kearns01}. In particular, there is a directed edge (or arc)
from individual $j$ to $i$ iff $w_{ji} \neq 0 $. 

\textit{Example.} Figure~\ref{fig:infl_example} shows an example of an LIG with binary behavior. Here, for each edge $(i, j)$, $w_{ji} = 1$ and $w_{ij} = 1$. That is, the game is a special type of LIG with symmetric influence factors. Furthermore, for each node $i$, its threshold $b_i$ is defined to be $0$. Therefore, at any PSNE of this game, each node wants to adopt the behavior of the majority of its neighbors and it is indifferent in the case of a tie.

\subsubsection{Connection to Polymatrix Games}
Polymatrix games~\citep{Janovskaja_1968_MR_by_Isbell} are $n$-player
noncooperative games where a player's total payoff is the sum of the
\textit{partial payoffs} received from the other players. Formally, for any joint action $\x$, player $i$'s payoff is given by $M_i ( x_i, \x_{-i}) \equiv \sum_{j \neq i} \alpha_{ji} (x_j, x_i)$, where $\alpha_{ji} (x_j, x_i)$ is the partial payoff that $i$ receives from $j$ when $i$ plays $x_i$ and $j$ plays $x_j$. Note that this partial payoff is local in nature and is not affected by the choice of actions of the other nodes. We will consider polymatrix games with only binary actions $\{1, -1\}$ here.

The following property shows an equivalence between LIGs and $2$-action polymatrix games. Thus, our computational study of LIGs directly carries over to $2$-action polymatrix games.

\begin{proposition}
\label{propos:polymatrix}
LIGs are equivalent to $2$-action polymatrix games, modulo the set of
PSNE.~\footnote{We present this result in the context of PSNE because
  that is the solution concept we use throughout the paper. However, this result easily extends to the more general notion
  of \emph{correlated equilibria} (CE) \citep{aumann74,aumann87}, which in
  turn generalizes the notion of
  \emph{mixed-strategy Nash equilibria} (MSNE) \citep{nash51}. (See
  \cite{fudenbergandtirole91} for textbook definitions of CE and MSNE.) The reason the
  result generalizes is that, as is well-known, two games with the same
  number of players and the same set of actions for each player have the same
  CE set if the set of payoff functions
  $\{u^1_i\}$ and $\{u^2_i\}$ of the games $1$ and $2$, respectively,
  satisfy the following condition: for all players $i$, there exist a
  positive constant $c_i > 0$ and an arbitrary real-valued function
  $d_i$ of
  the actions of all the players except $i$, such that,
  $u^1_i(x_1,\x_{-i}) = c_i\,  u^2_i(x_i,\x_{-i}) + d_i(\x_{-i})$. It is known that the computation of MSNE of 2-action polymatrix games is PPAD-hard. (It follows from the result of \cite{daskalakis2009}, although it is not explicitly mentioned there.) An implication of this is that the computation of MSNE of LIGs is also PPAD-hard.}
\end{proposition}

\begin{proof}
Assume that the number of players $n > 1$; otherwise, the statement holds trivially. 
We first show that given any instance of an LIG, we can design a polymatrix game that has the same set of PSNE. In an LIG instance, player $i$'s payoff is given by

\begin{eqnarray*}
	\begin{split}
			u_i (x_i, \x_{-i}) & = x_i \left(\sum_{j \neq i} w_{ji} x_j - b_i\right) \\
												 & = x_i \sum_{j \neq i} \left(w_{ji} x_j - \frac{b_i}{n-1}\right) \\
												 & = \sum_{j \neq i} \left(x_i w_{ji} x_j - \frac{x_i b_i}{n-1}\right).
	\end{split}
\end{eqnarray*}

Thus, constructing a polymatrix game instance by defining $\alpha_{ji}(x_j, x_i) \equiv x_i w_{ji} x_j - \frac{x_i b_i}{n-1}$, we have the same set of PSNE in both instances.

Next, we show the reverse direction. Player $i$'s payoff in a $2$-action polymatrix game is given by

\begin{eqnarray*}
	\begin{split}
			M_i ( x_i, \x_{-i}) & = \sum_{j \neq i} \alpha_{ji} (x_j, x_i) \\
													& = \sum_{j \neq i} \left(\indicator{x_i = 1} \alpha_{ji}(x_j, 1) + \indicator{x_i = -1} \alpha_{ji}(x_j, -1) \right) \\
													& = \sum_{j \neq i} \left(\frac{1 + x_i}{2} \alpha_{ji}(x_j, 1) + \frac{1 - x_i}{2} \alpha_{ji}(x_j, -1) \right) \\
													& = \frac{x_i}{2} \sum_{j \neq i} \left(\alpha_{ji} (x_j, 1) - \alpha_{ji} (x_j, -1) \right) + \frac{1}{2} \sum_{j \neq i} \left(\alpha_{ji}(x_j, 1) + \right.\\
													& \hspace{0.25in} \left. \alpha_{ji}(x_j, -1) \right).
	\end{split}
\end{eqnarray*}

Note that the second term above does not have any effect on $i$'s choice of action. Thus, we can re-define the payoff of player $i$, without making any change to the set of PSNE of the original polymatrix game, as follows.

\begin{eqnarray*}
	\begin{split}
			M_i^\prime ( x_i, \x_{-i}) & = \frac{x_i}{2} \sum_{j \neq i} \left(\alpha_{ji} (x_j, 1) - \alpha_{ji} (x_j, -1)\right) \\
																 & = \frac{x_i}{2} \left(\sum_{j \neq i} \left(\indicator{x_j = 1} \alpha_{ji} (1, 1) + \indicator{x_j = -1} \alpha_{ji}(-1, 1)\right) - \right.\\ 
																 		& \hspace*{0.25in} \left. \sum_{j \neq i}  \left(\indicator{x_j = 1} \alpha_{ji} (1, -1) + \indicator{x_j = -1} \alpha_{ji} (-1, -1)\right) \right)\\
																 & = \frac{x_i}{2} \left(\sum_{j \neq i} \left(\frac{1+x_j}{2} \alpha_{ji} (1, 1) + \frac{1-x_j}{2} \alpha_{ji}(-1, 1) \right) - \right.\\ & \hspace*{0.25in} \left. \sum_{j \neq i} \left(\frac{1+x_j}{2} \alpha_{ji} (1, -1) + \frac{1-x_j}{2} \alpha_{ji} (-1, -1) \right) \right)\\
																 & = \frac{x_i}{4} \left( \sum_{j \neq i}  {x_j} \left(\alpha_{ji} (1, 1) - \alpha_{ji}(-1, 1) - \alpha_{ji} (1, -1) + \alpha_{ji} (-1, -1) \right)\right.\\  & \hspace*{0.25in} \left. + \sum_{j \neq i} \left(\alpha_{ji} (1, 1) + \alpha_{ji}(-1, 1) - \alpha_{ji} (1, -1) - \alpha_{ji} (-1, -1)\right)\right).
	\end{split}
\end{eqnarray*}

Therefore, we can construct an LIG that has exactly the same set of PSNE as the polymatrix game, in the following way. For any player $i$, define $b_i \equiv -\sum_{j \neq i}  \frac{1}{4} (\alpha_{ji} (1, 1) + \alpha_{ji}(-1, 1) - \alpha_{ji} (1, -1) - \alpha_{ji} (-1, -1))$, and for any player $i$ and any other player $j$, define $w_{ji} \equiv  \frac{1}{4} (\alpha_{ji} (1, 1) - \alpha_{ji}(-1, 1) - \alpha_{ji} (1, -1) + \alpha_{ji} (-1, -1))$.
\end{proof}

\section{Equilibria Computation in Linear Influence Games}
\label{SecComputation}
We first study the problem of computing and counting PSNE in LIGs. We show that several special cases of LIGs present us with attractive computational advantages, while the general problem is intractable unless P = NP. We present heuristics to compute PSNE in general LIGs.
%
\subsection{Nonnegative Influence Factors}
When all the influence factors are non-negative, an LIG is supermodular~\citep{milgromandroberts90,topkis79}. 
In particular, the game exhibits what is called \textit{strategic
  complementarity}~\citep{bulowetal85,milgromandroberts90,topkis79,topkis78}.~\footnote{\label{foot:stratcomp}
A formal, general definition of 
strategic complementarity is beyond the scope of this paper. Instead, we
present a definition in the context of this paper and refer the
reader to standard references for a general
definition. We say that player $i$ in an influence game $\G$ exhibits \emph{strategic
  complementarity} if for every pair of the joint-actions $\x_{-i}$
and $\x'_{-i}$, if  $\x_{-i} \geq \x'_{-i}$, element-wise, implies
$u_i(x_i=1,\x_{-i}) \geq u_i(x_i=-1,\x'_{-i})$. We then say the influence
game $\G$, as a whole, exhibits \emph{strategic
  complementarity} if every player in the game does. Intuitively, it
says that the action/behavior in the best-response correspondence of any player cannot ``decrease'' (i.e.,
move ``down'' from $\{+1\}$ to $\{-1\}$, or to $\{-1,+1\}$ for that matter) if the
actions/behavior of the other players ``increases'' (i.e., at least one
other player moves ``up'' from $-1$ to $+1$); and \emph{vice
  versa}. Thus, roughly speaking, we can say that the players'
actions ``complement'' each other ``strategically;'' or said
differently, in general, each player prefers to ``play along'' by
choosing an  action ``consistent'' with that chosen by the other
players: if the other players ``move up'' or ``move down'' then the player would like
to ``follow along'' with the other players by choosing the respective
action, or ``stay put.''}
Hence, the best-response dynamics converges in at most $n$ rounds. From this, we obtain the following result.
\begin{proposition}
The problem of computing a PSNE is in P for LIGs on general graphs with 
only non-negative influence factors.
\end{proposition}

This property implies certain monotonicity of the best-response correspondences. More specifically, for
each player $i$, if any subset of the other players ``increases his/her strategy'' by adopting the new behavior,
then player $i$'s best-response cannot be to abandon adoption (i.e., move from $1$ to $-1$). In other
words, once a player adopts the new behavior, it has no incentive to go back. This monotonicity property also follows directly from the linear threshold model. Strategic complementarity implies other interesting characterizations of the structure of PSNE in LIGs and the behavior of best-response dynamics. For example, it is not hard to see that such games always have a PSNE: If we start with the complete
assignment in which either everyone is playing $1$, or everyone is playing
$-1$, parallel/synchronous best-response dynamics converges after at most $n$ rounds~\citep{milgromandroberts90}. 
If both best-response processes starting with all $-1$'s and all $1$'s converge to the same PSNE, then the PSNE is unique. Otherwise, any other PSNE of the game must be ``contained'' between the two different PSNE. 
We can also view this from the perspective of constraint propagation with monotonic constraints \citep{russellandnorvig03}.

\subsection{Special Influence Structures and Potential Games}
Several special subclasses of LIGs are potential
games~\citep{monderer96}.~\footnote{A formal, general definition of 
potential games is beyond the scope of this paper. Instead, we
present a definition in the context of this paper and refer the
reader to the standard reference for a more general
definition~\citep{monderer96}. We say an influence game $\G$ is a \emph{ordinal potential
game} if there exists a function $\Phi : \{-1,+1\}^n \rightarrow \R$, called the
\emph{ordinal potential function}, independent of any specific player, such
that for every joint-action $\x$, for every player $i$, and for every
possible action/pure-strategy $x'_i$ that player $i$ can take, we have
that 
$u_i(x'_i,\x_{-i}) - u_i(x_i,\x_{-i}) > 0$ if and only if
$\Phi(x'_i,\x_{-i}) - \Phi(x_i,\x_{-i}) > 0$. Note the abuse of
notation by letting $\Phi(\x) \equiv \Phi(\x_i,\x_{-i})$ for each player $i$; this
is consistent with the same abuse of notation we use for the payoff
functions $u_i$ throughout the paper, standard in the game-theory
literature. When the overall condition above is the
stricter condition $u_i(x'_i,\x_{-i}) - u_i(x_i,\x_{-i}) =
\Phi(x'_i,\x_{-i}) - \Phi(x_i,\x_{-i})$, then we call $\G$ and
$\Phi$ an
\emph{exact potential game and function}, respectively. For
simplicity, we refer to such $\G$ and $\Phi$
simply as a \emph{potential game and function}, when clear from
context. Intuitively, the potential function, a
global quantity independent of any specific player,
defines the best-response correspondence of each player in the
potential game. Also, the PSNE of a potential game are essentially the
local minima
(or stationary points) of the potential function. Hence, one could
view the players as trying to optimize the potential function as a
``global group'' but via ``individual, local best-responses.''
By performing asynchronous/non-simultaneous best-response dynamics,
the players are implicitly performing an
axis-parallel optimization of the potential function. Because the
domain of the potential function (i.e., the space of joint-actions) is
finite, this process will always converge to a local maxima or stable
point of the potential function, or equivalently, to a PSNE of the
potential game. Hence, every potential games always has a PSNE.}
This connection guarantees the existence of PSNE in such games.
\begin{proposition}
If the influence factors of an LIG $\G$ are symmetric (i.e., $w_{ji} = w_{ij}$, for all $i,j$), then $\G$ is a potential game.
\end{proposition}
\begin{proof} We show that the game has an ordinal potential function, 
\begin{equation}
\label{eqn:pot_sym}
\Phi(\x) = \sum_{t = 1}^{n} x_t \left( \sum_{i \neq t} \frac{x_i w_{it}}{2} - b_t \right). 
\end{equation}
Consider any player $j$. The difference in $j$'s payoff for $x_j = 1$ and $x_j = -1$ (assuming all other players play $\x_{-j}$ in both cases) is
\begin{eqnarray}
\label{eqn:diff_payoff}
\begin{split}
&u_j(1, \x_{-j}) - u_j(-1, \x_{-j}) \\ & = 1 \times \left(\sum_{i \neq j} x_i w_{ij} - b_j \right) - (-1) \times \left(\sum_{i \neq j} x_i w_{ij} - b_j \right) \\
														& = 2 \times \left(\sum_{i \neq j} x_i w_{ij} - b_j \right).
\end{split}
\end{eqnarray}
Next, the difference in the potential function when $j$ plays $1$ and $-1$ is
\begin{eqnarray}
\label{eqn:diff_pot}
\begin{split}
& \Phi(1, \x_{-j}) - \Phi(-1, \x_{-j}) \\ & = 1 \times \left( \sum_{i \neq j} \frac{x_i w_{ij}}{2} - b_j \right) 
																					+ \sum_{t \neq j} x_t \left( \sum_{i \neq t} \indicator{i \neq j} \frac{x_i w_{it}}{2} - b_t \right)\\
																			&	+ \sum_{t \neq j} x_t \left( \sum_{i \neq t} \indicator{i = j} \frac{1 \times w_{it}}{2} - b_t \right)  - \\
																			&	(-1) \times \left( \sum_{i \neq j} \frac{x_i w_{ij}}{2} - b_j \right) 
																					- \sum_{t \neq j} x_t \left( \sum_{i \neq t} \indicator{i \neq j} \frac{x_i w_{it}}{2} - b_t \right)\\
																			&		- \sum_{t \neq j} x_t \left( \sum_{i \neq t} \indicator{i = j} \frac{(-1) \times w_{it}}{2} - b_t \right)\\
																			& = 2 \times \left( \sum_{i \neq j} \frac{x_i w_{ij}}{2} - b_j \right) 
																				+ 2 \times \left(\sum_{t \neq j} \frac{x_t w_{jt}}{2} \right) \\
																			& = 2 \times \left(\sum_{i \neq j} x_i w_{ij} - b_j \right).											
\end{split}															
\end{eqnarray}

The last line follows by the symmetry of the weights (i.e., $w_{ij} = w_{ji}$).
\end{proof}

If, in addition, the threshold $b_i = 0$ for all $i$, the game is a \emph{party-affiliation game}, 
and computing a PSNE in such games is PLS-complete~\citep{fabrikantetal04}. 

The following result is on a large class of games that we call
\textit{indiscriminate} LIGs, where for every player $i$, the
influence weight, $w_{ij} \equiv \delta_i \neq 0$, that $i$ imposes on
every other player $j$ is the same. The interesting aspect of this
result is that these LIGs are potential games despite being possibly
\textit{asymmetric} and exhibiting strategic substitutability (due to
negative influence factors).~\footnote{
\emph{Strategic substitutability} is essentially the opposite of
strategic complementarity, as presented in Footnote~\ref{foot:stratcomp},
\emph{except} that one replaces the $\geq$ sign with a $\leq$ sign
in the hypothesis condition in the definition (i.e., $\x_{-i} \leq \x'_{-i}$). Once
again a formal, general definition is beyond the scope of this paper. Intuitively, it
says that the action/behavior in the best-response correspondence of any player cannot ``decrease'' (i.e.,
move ``down'' from $\{+1\}$ to $\{-1\}$, or to $\{-1,+1\}$ for that matter) if the
actions/behavior of the other players also ``decreases'' (i.e., at least one
other player moves ``down'' from $+1$ to $-1$); and \emph{vice
  versa}. Thus, roughly speaking, we can say that the players'
actions are ``substitutes'' of each other ``strategically;'' or said
differently, in general, each player prefers to play an action that is
opposite to/different than
that chosen by the other players: if the other players ``move up'' or ``move down'' then the player would like
to ``go against the crowd'' by choosing an opposite 
action, or ``stay put.''}
\begin{proposition}
Let $\G$ be an indiscriminate LIG in which 
all $\delta_i$ for all $i$, have the same sign, denoted by $\rho \in \{-1,+1\}$.
Then $\G$ is a potential game with the following potential function
\(
\textstyle
\Phi(\mathbf{x}) = \rho \left[ \left(\sum_{i = 1}^n \delta_i x_i \right)^2 - 2 \sum_{i
= 1}^n b_i \delta_i x_i \right]
\).
\end{proposition}
\begin{proof} It is sufficient to show that the sign of the difference in the individual utilities of any player due to changing her action unilaterally, is the same as the sign of the difference in the corresponding potential functions. For any player $j$, the first difference is 
\begin{align}
\label{eq:diff1}
\begin{split}
& 1 \times \left(\sum_{i \neq j} \delta_i x_i - b_j \right) - (-1) \times \left(\sum_{i \neq j} \delta_i x_i - b_j \right) \\
& = 2 \left(\sum_{i \neq j} \delta_i x_i - b_j \right).
 \end{split}
\end{align}
The potential function when $j$ plays $1$,
\begin{align*}
&\Phi(x_j = 1, \mathbf{x}_{-j})\\ & = \rho \left[ \left(\sum_{i \neq j} \delta_i x_i + \delta_j \times 1 \right)^2 
																- 2 \sum_{i \neq j} b_i \delta_i x_i - 2 b_j \delta_j \times 1 \right]\\
%
& = \rho \left[ \left(\sum_{i \neq j} \delta_i x_i \right)^2 + {\delta_j}^2 + 2 \left(\sum_{i \neq j} \delta_i x_i \right) {\delta_j} 
- 2 \sum_{i \neq j} b_i \delta_i x_i - 2 b_j \delta_j \right].
\end{align*}
The potential function when $j$ plays $-1$,
\begin{align*}
&\Phi(x_j = -1, \mathbf{x}_{-j})\\ & = \rho \left[ \left(\sum_{i \neq j} \delta_i x_i + \delta_j \times (-1) \right)^2 
														- 2 \sum_{i \neq j} b_i \delta_i x_i - 2 b_j \delta_j \times (-1) \right]\\		%
& = \rho \left[ \left(\sum_{i \neq j} \delta_i x_i \right)^2 + {\delta_j}^2 - 2 \left(\sum_{i \neq j} \delta_i x_i \right) {\delta_j} 
- 2 \sum_{i \neq j} b_i \delta_i x_i + 2 b_j \delta_j \right].
\end{align*}
%
%
Thus, the difference in the potential functions,
\begin{align}
\label{eq:diff2}
\begin{split}
\Phi(x_j = 1, \mathbf{x}_{-j}) - \Phi(x_j = -1, \mathbf{x}_{-j})
										& = 4 \rho \delta_j \left(\sum_{i \neq j} \delta_i x_i - b_j \right).
\end{split}
\end{align}
Since $\rho \delta_j > 0$, the quantities given in~\eqref{eq:diff1} and~\eqref{eq:diff2} have the same sign.\end{proof}


\subsection{Tree-Structured Influence Graphs}
The following result follows from a careful, non-trivial modification of the {\bf TreeNash} algorithm~\citep{kearns01}. Note that the running time of the {\bf TreeNash} algorithm is exponential in the degree of a node and thus also exponential in the representation size of an LIG! In contrast, our algorithm is linear in the maximum degree and thereby linear in the representation size of an LIG. The complete proof follows a proof sketch.
\begin{theorem}
\label{thm:pne_tree}
There exists an $O(n d)$ time algorithm to find a PSNE, or to decide that there exists none, in LIGs with tree structures, where 
$d$ is the maximum degree of a node.
\end{theorem}

\noindent \textit{Proof Sketch.}
We use similar notations as in~\citep{kearns01}. The modification of the {\bf TreeNash} involves efficiently (in $O(d)$ time, not $O(2^d)$) determining the existence of a witness vector and constructing one, if it exists,  at each node during the downstream pass, in the following way.

Suppose that an internal node $i$ receives tables $T_{ki}(x_k, x_i)$ from its parents $k$, and that $i$ wants to send a table $T_{ij}(x_i, x_j)$ to its unique child $j$. If for some parent $k$ of $i$, $T_{ki}(-1, x_i) = 0$ and $T_{ki}(1, x_i) = 0$, then $i$ sends the following table entries to $j$: $T_{ij}(x_i, -1) = 0$ and $T_{ij}(x_i, 1) = 0$. Otherwise, we first partition $i$'s set of parents into two sets in $O(d)$ time: $Pa_1(i , x_i)$ consisting of the parents $k$ of $i$ that have a unique best response $\hat{x}_k$ to $i$'s playing $x_i$ and $Pa_2(i, x_i)$ consisting of the remaining parents of $i$.
We show that $T_{ij}(x_i, x_j) = 1$ iff 
\begin{align*}
 & x_i( x_j w_{ji} + \sum_{k \in Pa_1(i, x_i)} \hat{x}_k w_{ki} + \\ 
 & \sum_{t \in Pa_2(i, x_i)} \underbrace{(2 \times \indicator{x_i w_{ti} > 0}-1)}_{t\text{'s action in witness vector}} w_{ti}) \ge 0,
\end{align*}
from which we get a witness vector, if it exists. \qed\ \\

Following is the complete proof of Theorem~\ref{thm:pne_tree}.

\begin{proof}
We denote any node $i$'s action by $x_i \in \{-1, 1\}$, its threshold
by $b_i$, and the influence of any node $i$ on another node $j$ by
$w_{ij}$. Furthermore, denote the set of parents of a node $i$ by $Pa(i)$. 
We now describe the two phases of the modified {\bf TreeNash} algorithm. 
%

\begin{enumerate}
	\item \textbf{Downstream phase.} In this phase each node sends a table to its unique child. We denote the table that node $i$ sends to its child $j$ as $T_{ij}(x_i, x_j)$, indexed by the actions of $i$ and $j$, and define the set of conditional best-responses of a node $i$ to a neighboring node $j$'s action $x_j$ as
	$BR_i(j, x_j) \equiv \{x_i\ |\ T_{ij}(x_i, x_j) = 1\}$. If $|BR_i(j, x_j)| = 1$ then we will abuse this notation by letting $BR_i(j, x_j)$ be the unique best-response of $i$ to $j$'s action $x_j$.

	The downstream phase starts at the leaf nodes. Each leaf node $l$ sends a table $T_{lk}(x_l,x_k)$ to its child $k$, where $T_{lk}(x_l, x_k) = 1$ if and only if $x_l$ is a conditional best-response of $l$ to $k$'s choice of action $x_k$. Suppose that an internal node $i$ obtains tables $T_{ki}(x_k, x_i)$ from its parents $k \in Pa(i)$, and that $i$ needs to send a table to its child $j$. Once $i$ receives the tables from its parents, it first computes (in $O(d)$ time) the following three sets that partition the parents of $i$ based on the size of their conditional best-response sets when $i$ plays $x_i$.	
	\begin{align*}
	&{Pa}_r(i, x_i) \equiv \{k \text{\ s.t.\ } k \in Pa(i)  \text{\ and\ } |BR_k(i, x_i)| = r\}, \text{\ for\ } r = 0, 1, 2.
	\end{align*}
	
		This is how $i$ computes the table $T_{ij}(x_i, x_j)$ sent to $j$: $T_{ij}(x_i, x_j) = 1$ if and only if there exists a \textit{witness vector} $(x_k)_{k \in Pa(i)}$ that satisfies the following two conditions:
		
		\indent \emph{Condition 1.} $T_{ki}(x_k, x_i) = 1$ for all $k \in Pa(i)$.
		
		\indent \emph{Condition 2.} The action $x_i$ is a best-response of node $i$ when every node $k \in Pa(i)$ plays $x_k$ and $j$ plays $x_j$. 
		
		There are two cases.
		
		\textbf{Case I: $Pa_0(i, x_i) \neq \emptyset$.} In this case, there exists some parent $k$ of $i$ for which both $T_{ki} (-1, x_i) = 0$ and $T_{ki} (1, x_i) = 0$. Therefore, there exists no witness vector that satisfies Condition 1, and $i$ sends the following table entries to $j$: $T_{ij}(x_i, x_j) = 0$, for $x_j = -1, 1$. 
	
	\textbf{Case II: $Pa_0(i, x_i) = \emptyset$.} In this case, we will show that there exists a witness vector for $T_{ij}(x_i, x_j) = 1$ satisfying Conditions 1 and 2 if and only if the following inequality	holds (which can be verified in $O(d)$ time). Below, we will use the \emph{sign} function $\sigma$: $\sigma(x) = 1$ if $x > 0$, and $\sigma(x) = -1$ otherwise.

\begin{align}
\label{ineq:tree}
x_i \left(w_{ji} x_j + \sum_{k \in {Pa}_1(i, x_i)} {w_{ki} BR_k(i, x_i) } + \sum_{k \in Pa_2(i, x_i)} w_{ki} \sigma(x_i w_{ki}) - b_i\right) \ge 0.
\end{align}

In fact, if Inequality (\ref{ineq:tree}) holds then we can construct a witness vector in the following way: If $k \in {Pa}_1(i, x_i)$, then let $x_k = BR_k(i, x_i)$, otherwise, let $x_k = \sigma(x_i w_{ki})$. Since each parent $k$ of $i$ is playing its conditional best-response $x_k$ to $i$'s choice of action $x_i$, we obtain, 
$T_{ki}(x_k, x_i) = 1$ for all $k \in Pa(i)$. Furthermore, Inequality (\ref{ineq:tree}) says that $i$ is playing its best-response $x_i$ to each of its parent $k$ playing $x_k$ and its child $j$ playing $x_j$.				
				 
				To prove the reverse direction, we
                                start with a witness vector $(x_k)_{k
                                  \in Pa(i)}$ such that Conditions 1
                                and 2 specified above hold. In
                                particular, we can write Condition 2 as
				\begin{align}
				\label{ineq:cond2}
				x_i \left(w_{ji} x_j + \sum_{k \in Pa(i)} w_{ki} x_k - b_i\right) \ge 0. 
				\end{align}
				The following line of arguments shows that Inequality (\ref{ineq:tree}) holds.
				\begin{align*}
				&\ x_i w_{ki} \sigma(x_i w_{ki})  \ge x_i w_{ki} x_k, \text{\ for any } {k \in Pa_2(i, x_i)}\\
				\Rightarrow &\ x_i \sum_{k \in Pa_2(i, x_i)} w_{ki} \sigma(x_i w_{ki})  \ge x_i  \sum_{k \in Pa_2(i, x_i)} w_{ki} x_k \\
				\Rightarrow &\ x_i \left(w_{ji} x_j + \sum_{k \in {Pa}_1(i, x_i)} {w_{ki} BR_k(i, x_i) } + \sum_{k \in Pa_2(i, x_i)} w_{ki} \sigma(x_i w_{ki}) -b_i \right)\\
				&\  \ge \ x_i \left(w_{ji} x_j + \sum_{k \in Pa(i)} w_{ki} x_k - b_i\right) \\
				\Rightarrow &\ x_i \left(w_{ji} x_j + \sum_{k \in {Pa}_1(i, x_i)} {w_{ki} BR_k(i, x_i) } + \sum_{k \in Pa_2(i, x_i)} w_{ki} \sigma(x_i w_{ki}) -b_i \right)\\
				&\ \ge 0, \text{\ using Inequality (\ref{ineq:cond2})}.
				\end{align*}				
	In addition to computing the table $T_{ij}$, node $i$ stores the following witness vector $(x_k)_{k \in Pa(i)}$ for each table entry $T_{ij}(x_i, x_j)$ that is $1$: if $k \in {Pa}_1(i, x_i)$, then $x_k = BR_k(i, x_i)$, otherwise, $x_k = \sigma(x_i w_{ki})$. The downstream phase ends at the root node $z$, and $z$ computes a unary table $T_z(x_z)$ such that $T_z(x_z) = 1$ if and only if there exists a witness vector $(x_k)_{k \in Pa(z)}$ such that $T_{kz}(x_k, x_z) = 1$ for all $k \in Pa(z)$ and $x_z$ is a best-response of $z$ to $(x_k)_{k \in Pa(z)}$.

The computation of the table at each node, which runs in $O(d)$ time,
dominates the time complexity of the downstream phase. We visit every
node exactly once. So, the total running time of the downstream phase is $O(n d)$. Note that if there does not exist any PSNE in the game then all the table entries computed by some node will be $0$.

	\item \textbf{Upstream phase.} In the upstream phase, each
          node sends instructions to its parents about which actions
          to play, along with the action that the node itself is
          playing. The upstream phase begins at the root node $z$. For
          any table entry $T_z(x_z) = 1$, $z$ decides to play $x_z$
          itself and instructs each of its parents to play the action
          in the witness vector associated with $T_z(x_z) = 1$. At an
          intermediate node $i$, suppose that its child $j$ is playing
          $x_j$ and instruct $i$ to play $x_i$. The node $i$ looks up the witness vector $(x_k)_{k \in Pa(i)}$ associated with $T_{ij}(x_i,x_j) = 1$ and instructs its parents to play according to that witness vector. 
	This process propagates upward, and when we reach all the leaf nodes, we obtain a PSNE for the game. Note that we can find a PSNE in this phase if and only if there exists one.

\end{enumerate}
		
	In the upstream phase, each node sends $O(d)$ instructions to its parents. Thus, the upstream phase takes $O(nd)$ time, and the whole algorithm takes $O(n d)$ time.\end{proof}

\subsection{Hardness Results}
Computational problems are often classified into complexity classes according to their hardness. Some of the hardest classes of problems are NP-complete problems, co-NP-complete problems, and \#P-complete problems \citep{dsjohnson}. In this section, we show that many of the computational problems related to LIGs belong to these hard classes.

First, computing PSNE in a general graphical game is known to be computationally hard \citep{GGS03}. However, that result does not imply intractability of our problem, nor do the proofs seem easily adaptable to our case. LIGs are a special type of graphical game with quadratic payoffs, or in other words a graphical, \emph{parametric} polymatrix game~\citep{Janovskaja_1968_MR_by_Isbell}, and thus have a more succinct representation than general graphical games ($O(nd)$ in contrast to $O(n 2^d)$, where $d$ is the maximum degree of a node). Next, we show that various interesting computational questions on LIGs are intractable, unless P = NP. 

We settle the central hardness question on LIGs (and also on 2-action polymatrix games) in 1(a) below. Related to the most-influential-nodes problem formulation, 1(b) states that given a subset of players, it is NP-complete to decide whether there exists a PSNE in which these players adopt the new behavior. We present a similar result in 1(c). 

A prime feature of our formulation of the most-influential-nodes
problem is the uniqueness of the desired stable outcome when the most influential nodes adopt their behavior according to the desired stable outcome. We show in (2) that deciding whether a given set of players fulfills this criterion is co-NP-complete. 

As we will see later, in order to compute a set of the most influential nodes, it suffices to be able to count the number of PSNE of an LIG (to be more specific, it suffices to count the number of PSNE extensions for a given partial assignment to the players' actions). We show in (3) that this problem is \#P-complete.
Note that the \#P-completeness result for LIGs even with star structure is in contrast to the polynomial-time counterpart for general graphical games with tree graphs, for which not only deciding the existence of a PSNE is in P,
but also counting PSNE on general graphical games with tree graphs is in P.  To better appreciate this result, consider the representation sizes of LIGs and tree-structured graphical games, which are linear and exponential in the maximum degree, respectively.

Below, we first summarize the hardness results with an outline of proof, followed by the complete proof of each individual statement.

\begin{theorem}
\begin{enumerate}
	\item It is NP-complete to decide the following questions in LIGs.
			\begin{enumerate}
				\item Does there exist a PSNE?
				\item Given a designated non-empty set of players, does there exist a PSNE consistent with those players playing $1$?
				\item Given a number $k \ge 1$, does there exist a PSNE with at least $k$ players playing $1$?
			\end{enumerate}
	\item Given an LIG and a designated non-empty set of players, it is co-NP-complete to decide if there exists a \textit{unique} PSNE with those players playing $1$.
	\item It is \#P-complete to count the number of PSNE, even for special classes of the underlying graph, such as a bipartite or a star graph.
\end{enumerate}
\end{theorem}
\noindent \emph{Proof Sketch.} The complete proofs appear immediately following
this proof sketch. The proof of 1(a) reduces the 3-SAT problem to an LIG that consists of a player for each clause and each variable of the 3-SAT instance. The influence factors among these players are designed such that the LIG instance possesses a PSNE if and only if the 3-SAT instance has a satisfying assignment. Since the underlying graph of the LIG instance is always bipartite, we obtain as a corollary that the NP-completeness of that existence problem holds even for LIGs on bipartite graphs. 

The proofs of 1(b), 1(c), and 2 use reductions from the \textit{monotone one-in-three SAT} problem. For 1(b), given a monotone one-in-three SAT instance $I$, we construct an LIG instance $J$ having a player for each clause and each variable of $I$. Again, we design the influence factors in such a way that $I$ is satisfiable if and only if $J$ has a PSNE. The reduction for 1(c) builds upon that of 1(b) with specifically designed \textit{extra players} and additional connectivity in the LIG instance. Again, the gadgets used in the proof of 1(c) are extended for the proof of 2.

The proof of 3 uses reductions from the 3-SAT and the \#KNAPSACK problem. The reduction from the 3-SAT problem is the same as that used in 1(a), and proof of the \#P-hardness of the bipartite case is by showing that the number of solutions to the 3-SAT instance is the same as the number of PSNE of the LIG instance. On the other hand, to prove the claim of \#P-completeness of counting PSNEs of LIGs having star 
graphs, we give a reduction from the \#KNAPSACK problem.
Given a \#KNAPSACK instance, we create an LIG instance with a star structure among the players and with specifically designed influence factors such that the number of PSNE of the LIG instance is the same as the number of solutions to the \#KNAPSACK instance. 
\qed

\subsection*{Complete Proofs of Hardness Results}
To enhance the clarity of the proofs we reduced existing NP-complete problems to LIGs with binary actions $\{0,1\}$, instead of $\{-1,1\}$. We next show, via a linear transformation, that one can reduce any LIG with actions $\{0,1\}$ to an LIG with the same underlying graph, but with actions $\{-1,1\}$.\\

\noindent \textbf{Reduction from $\{0,1\}$-action LIG to $\{-1,1\}$-action LIG.} Consider any $\{0,1\}$-action LIG instance $I$, where $w$ and $b$ denote the influence factors and the thresholds, respectively (see Definition~\ref{def:lig}). 
We next construct a $\{-1,1\}$-action LIG instance $J$ with the same players that are in $I$ and with influence factors $w_{ji}^\prime \equiv \frac{w_{ji}}{2}$ (for any $i$ and any $j \neq i$), thresholds $b_i^\prime \equiv b_i - \sum_{j \neq i}\frac{w_{ji}}{2}$ (for any $i$). We show that $\x$ is a PSNE of $I$ if and only if $\x^\prime$ is a PSNE of $J$, where $x_i^\prime = 2 x_i -1$ for any $i$.

By definition, $\x$ is a PSNE of $I$ if and only if for any player $i$, 

\begin{eqnarray*}
\begin{split}
& x_i \left(\sum_{j \neq i} x_j w_{ji} - b_i\right) \ge (1 - x_i) \left(\sum_{j \neq i} x_j w_{ji} - b_i\right) \\
& \Leftrightarrow(2x_i - 1) \left(\sum_{j \neq i} x_j w_{ji} - b_i\right) \ge 0 \\
& \Leftrightarrow x_i^\prime \left(\sum_{j \neq i} \frac{x_j^\prime + 1}{2} w_{ji} - b_i\right) \ge 0 \\
& \Leftrightarrow x_i^\prime \left(\sum_{j \neq i} {x_j^\prime}\frac{w_{ji}}{2} - \left(b_i - \sum_{j \neq i}\frac{w_{ji}}{2}\right)\right) \ge 0 \\
& \Leftrightarrow x_i^\prime \left(\sum_{j \neq i} x_j^\prime w_{ji}^\prime - b_i^\prime\right) \ge 0,
\end{split}
\end{eqnarray*}
which is the equivalent statement of $\x^\prime$ being a PSNE of $J$. \qed

\begin{theorem}
\label{thm:NPComplete}
It is NP-complete to decide if there exists a PSNE in an LIG.
\end{theorem}
\begin{proof} 
Since we can verify whether a joint action is a PSNE or not in polynomial time, the problem is in NP. We use a reduction from the 3-SAT problem to show that the problem is NP-hard.

Let $I$ be an instance of the 3-SAT problem. Suppose that $I$ has $m$ clauses and $n$ variables. For any variable $i$, we define $C_i$ to be the set of clauses in which $i$ appears, and for any clause $k$, we define $V_k$ to be the set of variables appearing in clause $k$. For any clause $k$ and any variable $i \in V_k$, let $l_{k,i}$ be $1$ if $i$ appears in $k$ in non-negated form and $0$ otherwise.
We now build an LIG instance $J$ from $I$. 
In this game, every clause as well as every variable is a player. Each clause $k$ has arcs to variables in $V_k$, and each variable $i$ has arcs to clauses in $C_i$. The structure of the graph is illustrated in Figure~\ref{fig:3SAT}. We next define the thresholds of the players and the influence factors on the arcs.
For any clause $k$, let its threshold be $1 - \epsilon - \sum_{i \in V_k} (1 - l_{k,i})$. Here, $\epsilon$ is a constant, and $0 < \epsilon < 1$.
For any variable $i$ let its threshold be $\sum_{k \in C_i} (1 - 2 l_{k,i})$. The weight on the arc from any clause $k$ to any variable $i \in V_k$ is defined to be $1 - 2 l_{k,i}$, and that from any variable $i$ to any clause $k \in C_i$ is $2 l_{k,i} - 1$. We denote the action of any clause $k$ by $z_k \in \{0, 1\}$ and that of any variable $i$ by $x_i \in \{0, 1\}$. 

\begin{figure}[h]
\centering
\includegraphics{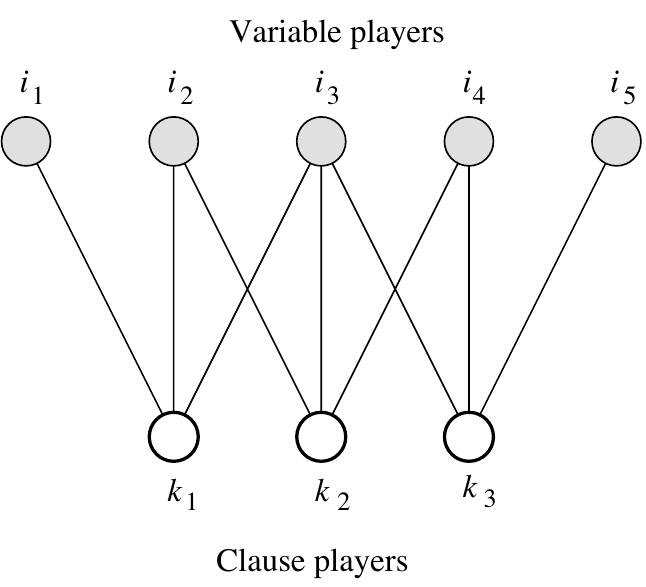}
\caption{Illustration of the structure of an LIG instance from a 3-SAT instance (each undirected edge represents two arcs of opposite directions between the same two nodes). In this example, the 3-SAT instance is $(i_1 \vee i_2 \vee i_3) \wedge (\neg i_2 \vee i_3 \vee i_4) \wedge (\neg i_3 \vee i_4 \vee \neg i_5)$. 
}
\label{fig:3SAT}
\end{figure}

First, we prove that if there exists a satisfying truth assignment in $I$ then there exists a PSNE in $J$. Consider any satisfying truth assignment $S$ in $I$. Let the players in $J$ choose their actions according to their truth values in $S$, that is, $1$ for \textit{true} and $0$ for \textit{false}. Clearly, every clause player is playing $1$. Next, we show that every player in $J$ is playing its best response under this choice of actions.

	We now show that no clause has incentive to play $0$, given that the other players do not change their actions.
	In the solution $S$ to $I$, every clause has a literal that is \textit{true}.
	Therefore, in $J$ every clause $k$ has some variable $i \in V_k$ such that $x_i = l_{k,i}$. We have to show that the total influence on $k$ is at least the threshold of $k$:
\begin{align*}
&	\sum_{i \in V_k} x_i \left(2 l_{k,i} - 1\right) \ge 1 - \epsilon - \sum_{i \in V_k}\left(1 - l_{k,i}\right)\\
	\Leftrightarrow &	\sum_{i \in V_k} \left(x_i \left(2 l_{k,i} - 1\right) + \left(1 - l_{k,i}\right)\right) \ge 1 - \epsilon\\
	\Leftrightarrow &	\sum_{i \in V_k} \left(x_i l_{k,i} + \left(1- x_i\right)\left(1 - l_{k,i}\right)\right) \ge 1 - \epsilon. 
\end{align*}	
	Since for some $i \in V_k$, $x_i = l_{k,i}$, the above inequality holds strictly, that is,
	\begin{align*}
	\sum_{i \in V_k} \left(x_i l_{k,i} + \left(1- x_i\right)\left(1 - l_{k,i}\right)\right) > 1 - \epsilon.
	\end{align*}
	Therefore, every clause $k$ must play $1$.
	
	We need to show that no variable player has incentive to deviate, given that the other players do not change their actions.
	The total influence on any variable player $i$ is $\sum_{k \in C_i} z_k (1 - 2 l_{k,i}) = \sum_{k \in C_i} (1 - 2 l_{k,i})$ (since $z_r = 1$ for every clause $r$). The threshold of $i$ is $\sum_{k \in C_i} (1 - 2 l_{k,i})$. Thus, every variable player $i$ is indifferent between choosing actions 1 and 0 and has no incentive to deviate.

We now consider the reverse direction, that is, given a PSNE in $J$ we show that there exists a satisfying assignment in $I$. We first show that at any PSNE, every clause must play 1. If this is not the case, suppose, for a contradiction, that for some clause $r$, $z_r = 0$. Since $r$'s best response is $0$ (this is a PSNE), we obtain
\begin{align*} 
&	\sum_{i \in V_r} x_i (2 l_{r,i} - 1) \le 1 - \epsilon - \sum_{i \in V_r} (1 - l_{r,i})\\
	\Leftrightarrow &	\sum_{i \in V_r} ( x_i l_{r,i} + (1 - x_i)(1 - l_{r,i})) \le 1 - \epsilon.
\end{align*}

Therefore, for every variable player $j \in V_r$, $x_j \neq l_{r,j}$. Furthermore, for any  $j \in V_r$, $j$ does not have any incentive to deviate. Using these properties of a PSNE we will arrive at a contradiction, and thereby prove that $z_r$ must be $1$.

Consider any variable player $j \in V_r$, and let the difference between $j$'s total incoming influence and its threshold be $U_j$. We get
\begin{align*}
&	U_j = \sum_{k \in C_j} z_k(1 - 2 l_{k,j}) - \sum_{k \in C_j} (1 - 2 l_{k,j}) = \sum_{k \in C_j} ((1-z_k)(2 l_{k,j} - 1))\\
	\Leftrightarrow\ & 	U_j = \sum_{k \in C_j} ((1-z_k)(2 l_{k,j} - 1) \mathbf{1}[l_{k,j} = 1]) + \sum_{k \in C_j} ((1-z_k)(2 l_{k,j} - 1) \mathbf{1}[l_{k,j} = 0])\\
	\Leftrightarrow\ &	 U_j = \sum_{k \in C_j} ((1-z_k) \mathbf{1}[l_{k,j} = 1]) - \sum_{k \in C_j} ((1-z_k) \mathbf{1}[l_{k,j} = 0]).
\end{align*}

At any PSNE, if $x_j = 1$ then $U_j \ge 0$; otherwise, $U_j \le 0$. Thus, the best response condition for variable $j$ gives us
\begin{align*}
&	\sum_{k \in C_j} ((1-z_k) \mathbf{1}[l_{k,j} = x_j]) \ge \sum_{k \in C_j} ((1-z_k) \mathbf{1}[l_{k,j} \neq x_j])\\
\Leftrightarrow &	\sum_{k \in C_j-\{r\}} ((1-z_k) \mathbf{1}[l_{k,j} = x_j]) + (1-z_r) \mathbf{1}[l_{r,j} = x_j] \ge \\
		&	\sum_{k \in C_j-\{r\}} ((1-z_k) \mathbf{1}[l_{k,j} \neq x_j]) + (1-z_r) \mathbf{1}[l_{r,j} \neq x_j]\\
\Leftrightarrow &	\sum_{k \in C_j-\{r\}} ((1-z_k) \mathbf{1}[l_{k,j} = x_j]) \ge \\
		&\sum_{k \in C_j-\{r\}}	((1-z_k)\mathbf{1}[l_{k,j} \neq x_j]) + 1,\ \text{since $l_{r,j} \neq x_j$}.
%
%
\end{align*}

The above inequality cannot be true, because the left hand side is always $0$ (if $l_{k,j} = x_j$ then $z_k$ must be 1 at any PSNE), and the right hand side is $\ge 1$. Thus, we obtained a contradiction, and $z_r$ cannot be 0.

So far, we showed that at any PSNE $z_k = 1$ for any clause player $k$. To complete the proof, we now show that for every clause player $k$, there exists a variable player $i \in V_k$ such that $x_i = l_{k,i}$. If we can show this then we can translate the semantics of the actions in $J$ to the truth values in $I$ and thereby obtain a satisfying truth assignment for $I$.

Suppose, for the sake of a contradiction, that for some clause $k$ and for all variable $i \in V_k$, $x_i \neq l_{k,i}$. Since $z_k = 1$, we find that
\begin{align*}
&	\sum_{i \in V_k} x_i (2 l_{k,i} - 1) \ge 1 - \epsilon - \sum_{i \in V_k} (1 - l_{k,i})\\
\Leftrightarrow\ &	\sum_{i \in V_k} (x_i l_{k,i} + (1 - x_i)(1 - l_{k,i})) \ge 1 - \epsilon\\
\Leftrightarrow\ &	0 \ge 1 - \epsilon \text{, which gives us the desired contradiction.}
\end{align*}
\end{proof}

The proof of Theorem~\ref{thm:NPComplete} reduces the 3-SAT problem to an LIG where the underlying graph is bipartite. Thus, we obtain the following corollary.
\begin{corollary}
\label{cor:bipartite}
It is NP-complete to decide if there exists a PSNE in an LIG on a bipartite graph.
\end{corollary}
The proof of Theorem~\ref{thm:NPComplete} directly leads us to the following result that the counting version of the problem is \#P-complete.
\begin{corollary}
\label{cor:sharpP}
It is \#P-complete to count the number of PSNE of an LIG.
\end{corollary}
\begin{proof} 
The proof follows from the proof of Theorem~\ref{thm:NPComplete}. Membership of this counting problem in \#P is easy to see. Using the same reduction as in the proof of Theorem~\ref{thm:NPComplete}, we find that each satisfying truth assignment (among the $2^n$ possibilities) to the variables of the 3-SAT instance $I$ can be mapped to a distinct PSNE of the LIG instance $J$. 
Furthermore, we saw that at each PSNE in $J$, every clause player must play 1. Thus, for each of the $2^n$ joint strategies of the variable players (while having the clause players play 1), if the joint strategy is a PSNE then we can map it to a distinct satisfying assignment in $I$. Moreover, each of these two mappings are the inverse of the other.
Therefore, the number of satisfying assignments of $I$ is the same as the number of PSNE in $J$.
Since counting the number of satisfying assignments of a 3-SAT instance is \#P-complete, counting the number of PSNE of an LIG, even on a bipartite graph, is also \#P-complete. \end{proof}

While Corollary~\ref{cor:sharpP} shows the hardness of counting the number of PSNE of an LIG on a general graph, we can show the same hardness result even on special classes of graphs, such as star graphs:
\begin{theorem}
\label{thm:starSharpP}
Counting the number of PSNE of an LIG on a star graph is \#P-complete.
\end{theorem}
\begin{proof} 
Since we can verify whether a joint strategy is a PSNE in polynomial time, the problem is in \#P. We will show \#P-hardness using a reduction from \#KNAPSACK, which is the problem of counting the number of feasible solutions in a 0-1 Knapsack problem: Given $n$ items, the weight $a_i \in {\cal Z}^+$ of each item $i$, and the maximum capacity of the sack $W \in {\cal Z}^+$, \#KNAPSACK asks how many ways we can pick the items to satisfy $\sum_{i = 1}^n a_i x_i \le W$, where $x_i = 1$ if the $i$-th item has been picked, and $x_i = 0$ otherwise. Given an instance $I$ of the \#KNAPSACK problem with $n$ items, we construct an LIG instance $J$ on a star graph with $n+1$ nodes. Let us label the nodes $v_0, ..., v_{n}$, where $v_0$ is connected to all other nodes. We define the influence factors among the nodes as follows: the influence of $v_0$ to any other node $v_i$, $w_{v_0 v_i} = 1$, and the influence in the reverse direction,  $w_{v_i v_0} = -a_i$. The threshold of $v_0$ is defined as $b_{v_0} = -W$, and the threshold of every other node $v_i$, $b_{v_i} = 1$. We denote the action of any node $v_i$ by $x_i \in \{0, 1\}$. Note that at any PSNE of $J$, $v_0$ must play $1$. Otherwise, if $v_0$ plays $0$ then all other nodes must also play $0$, and this implies that $v_0$ must play $1$, giving us a contradiction.

We prove that the number of feasible solutions in $I$ is the same as the number of PSNE in $J$. For any $(x_1, ..., x_n) \in \{0, 1\}^n$ in $I$, we map each $x_i$ to the action selected by $v_i$ in $J$, for $1 \le i \le n$. As proved earlier, the action of $v_0$ must be $1$ at any PSNE. Furthermore, when $v_0$ plays $1$, all other nodes become indifferent between playing $0$ and $1$. Thus, the number of PSNE in $J$ is the number of ways of satisfying the inequality $\sum_{i = 1}^n w_{v_i v_0} x_i \ge b_{v_0}$, which is equivalent to $\sum_{i = 1}^n a_i x_i \le W$. Thus the number of PSNE in $J$ is equal to the number of feasible solutions in $I$. \end{proof}

The following three theorems show the hardness of several other variants of the problem of computing a PSNE of an LIG. 

\begin{theorem}
\label{thm:Extension1Hard}
Given an LIG, along with a designated subset of $k$ players in it, it is NP-complete to decide if there exists a PSNE consistent with those $k$ players playing the action 1.
\end{theorem}
\begin{proof} It is easy to see that the problem is in NP, since a succinct yes certificate can be specified by a joint action of the players, where the designated players play 1, and it can be verified in polynomial time whether this is a PSNE or not.

We show a reduction from the monotone one-in-three SAT problem, a known NP-complete problem, to prove that the problem is NP-hard. An instance of the monotone one-in-three SAT problem consists of a set of $m$ clauses and a set of $n$ variables, where each clause has exactly three variables. 
The problem asks whether there exists a truth assignment to the variables such that each clause has exactly one variable with the truth value of \textit{true}. Given an instance of the monotone one-in-three SAT problem, we construct an instance of LIG as follows (please refer to Figure~\ref{fig:m13sat1} for an illustration). For each variable we have a \textit{variable player} in the game, and for each clause we have a \textit{clause player}. Each variable player has a threshold of $0$, and each clause player has a threshold of $\epsilon$, where $0 < \epsilon < 1$. We now define the connectivity among the players of the game. There is an arc with weight (or influence) $-1$ from a variable player $u$ to another variable player $v$ if and only if, in the monotone one-in-three SAT instance, both of the corresponding variables appear together in at least one clause. Also, for each clause $t$ and each variable $w$ appearing in $t$, there is an arc from the variable player (corresponding to $w$) to the clause player (corresponding to $t$) with weight $1$. Furthermore, we assign $k = m$, and assume that the designated set of players is the set of clause players. We also assume that the action $1$ in the LIG corresponds to the truth value of \textit{true} in the monotone one-in-three SAT problem and $0$ to \textit{false}.

\begin{figure}[h]
\centering
\includegraphics{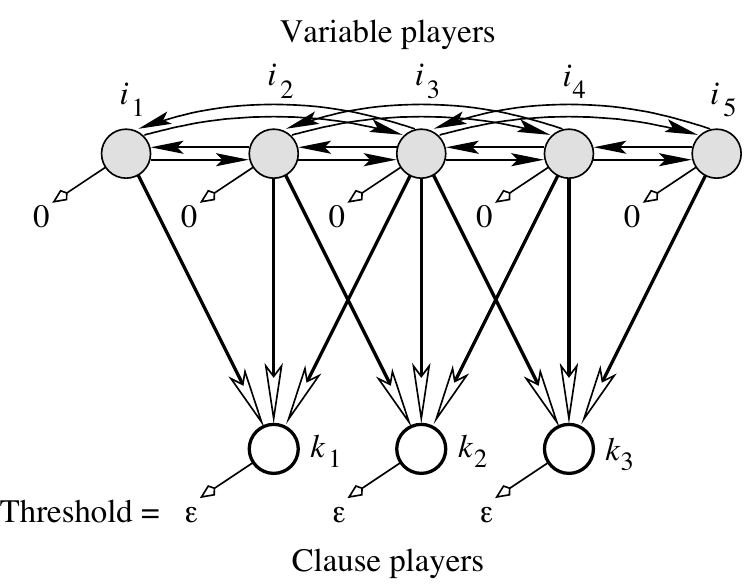}
\caption{Illustration of the NP-hardness reduction of Theorem~\ref{thm:Extension1Hard}. The monotone one-in-three SAT instance is $(i_1 \vee i_2 \vee i_3) \wedge (i_2 \vee i_3 \vee i_4) \wedge (i_3 \vee i_4 \vee i_5)$. 
The threshold of each variable player is 0, and that of each clause player is $\epsilon$.
}
\label{fig:m13sat1}
\end{figure}

Note that the way we constructed the LIG, at most one variable player per clause can play the action 1 at any PSNE. To see this, assume, for contradiction, that at some PSNE two variable players $u$ and $v$, both connected to the same clause $t$, are playing the action 1. Then the influence on either of these two variable players is $\le -1$, which is less than its threshold $0$, and this contradicts the PSNE assumption. Also, note that at any PSNE, each clause player will play the action 1 if and only if at least one of the variable players connected to it plays 1.

First, we show that if there exists a solution to the monotone one-in-three SAT instance then there exists a PSNE in the LIG where the set of clause players play 1. A solution to the monotone one-in-three SAT problem implies that each clause has the truth value of \textit{true} with exactly one of its variables having the truth value of \textit{true}. We claim that in the LIG, every player playing according to its truth assignment, is a PSNE. First, observe that the variable players do not have any incentive to change their actions, since the ones playing 1 are indifferent between playing $0$ and $1$ (because the total influence = 0 = threshold) and the remaining must play 0 (because the total influence is $\le -1 <$ threshold). Since each clause has one of its variables playing 1, each clause player must play 1 (because $1 > \epsilon$). This concludes the first part of the proof.

We next show that if there exists a PSNE with the clause players playing 1 then there exists a solution to the monotone one-in-three SAT instance. Consider any PSNE where the clause players are playing 1. Since each clause player is playing 1, \textit{at least} one of the three variable players connected to the clause player is playing 1.  Furthermore, as we showed earlier, no two variables belonging to the same clause can play $1$ at any PSNE. Thus, for each clause player, \textit{at most} one variable player connected to it is playing 1. Therefore, for every clause player, exactly one variable player connected to it is playing 1. Translating the semantics of the actions to the truth values of the variables and the clauses, we obtain a solution to the monotone one-in-three SAT instance. \end{proof}

\begin{theorem}
\label{thm:Extension1Hard2}
Given an LIG and a number $k \ge 1$, it is NP-complete to decide if there exists a PSNE with at least $k$ players playing the action 1.
\end{theorem}
\begin{proof} Clearly, the problem is in NP, since we can verify a whether a joint action is a PSNE or not in polynomial time.

For the proof of NP-hardness, once again we show a reduction from the monotone one-in-three SAT problem.  Please see Figure~\ref{fig:m13sat2} for an illustration.
Given an instance $I$ of the monotone one-in-three SAT problem, we first build an LIG as shown in the proof of Theorem~\ref{thm:Extension1Hard}. We then add $m(m-1)$ additional players, named \textit{extra players}, to the game, where $m$ is the number of clauses in $I$. Each of these extra players is assigned a threshold of $\epsilon$, where $0 < \epsilon < 1$. The way we connect the extra players to the other players is as follows: From each clause player we introduce $m-1$ arcs, each weighted by $1$, to $m-1$ distinct extra players. That is, no two clause players have arcs to the same extra player. Finally, we set $k = m^2$.
We denote this instance of LIG by $J$.

\begin{figure}[h]
\centering
\includegraphics{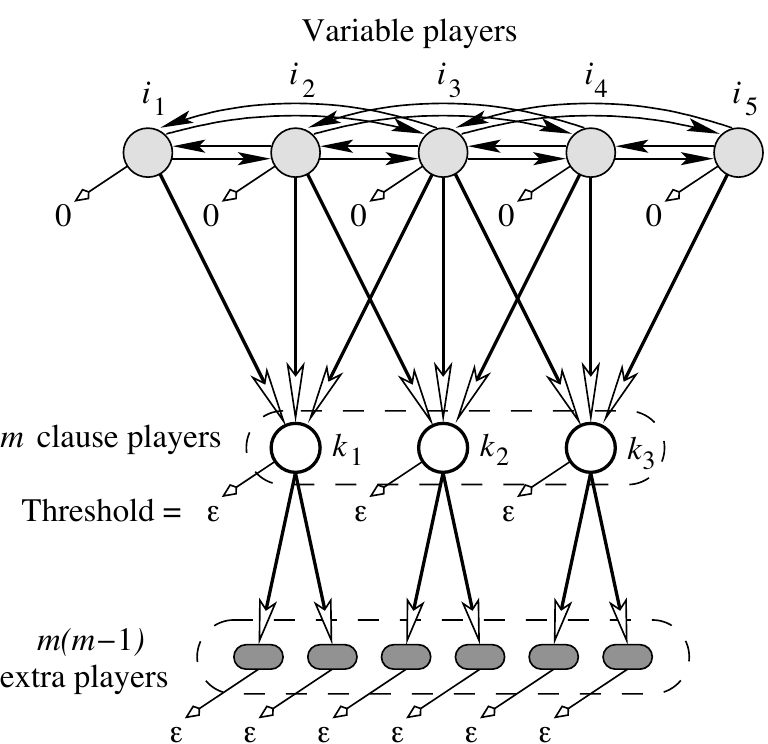}
\caption{Illustration of the NP-hardness reduction of Theorem~\ref{thm:Extension1Hard2}. 
}
\label{fig:m13sat2}
\end{figure}

We prove that for any solution to $I$ there exists a PSNE with $k$ players playing 1 in $J$. Suppose that each of the variable and clause players is playing according to their corresponding truth value in the solution to $I$. None of the variable players has any incentive to change its action, because exactly one variable player connected to each clause player is playing 1. For the same reason, the clause players, each playing 1, also do not have any incentive to deviate. Considering the extra players, each of these players must play 1, because each of the clause players is playing 1. The total number of clause and extra players is $k$. Therefore, we have a Nash equilibrium where at least $k$ players are playing 1.

On the other direction, consider any PSNE in $J$ with at least $k$ players playing 1. 
We claim that all the clause and extra players are playing 1 at this PSNE. If this is not true then at least one of these players is playing 0. This implies that at least one clause player is playing 0, because conditioned on a PSNE, whenever a clause player plays 1, all the extra players connected to it also plays 1. Furthermore, by our construction at most one of the variable players connected to each clause player can play 1. So, the total number of players playing 1 is $\le (m-1)(m+1) < m^2$ (at most $m-1$ clause players are playing $1$, and for each of these clause players, $m-1$ extra players, $1$ variable player, and the clause player itself are playing $1$), which contradicts our assumption that $m^2$ players are playing 1. Thus, at any PSNE with $k$ players playing 1, it must be the case that every clause player is playing 1. This leads us to a solution for $I$.\end{proof}

\begin{theorem}
\label{thm:UniqueHard}
Given an LIG and a designated set of $k \ge 1$ players, it is co-NP-complete to decide if there exists a \textit{unique} PSNE with those players playing the action $1$.
\end{theorem}
\begin{proof} Two distinct joint actions (PSNE), each having the same $k$ players playing 1, can serve as a succinct no certificate, and we can check in polynomial time if these two joint actions are indeed PSNE or not.

Suppose that $I$ is an instance of the monotone one-in-three SAT problem. We reduce $I$ to an instance $J$ of our problem in polynomial time and show that $J$ has a ``no'' answer if and only if $I$ has a ``yes'' answer.

Given $I$, we start constructing an LIG in the same way as in Theorem~\ref{thm:Extension1Hard} (see Figures~\ref{fig:m13sat3} and~\ref{fig:m13sat1}). Assign $k = m^2$. 
Now, add two new players, named the \textit{all-satisfied-verification player} and the \textit{none-satisfied-verification player}, which have threshold values of $m-\epsilon$ and $-\epsilon$, respectively. 
We add arcs from every clause player to these two new players, and the arcs to the all-satisfied-verification player are weighted by 1, and the ones to the none-satisfied-verification player are weighted by $-1$. 

\begin{figure}[h]
\centering
\includegraphics{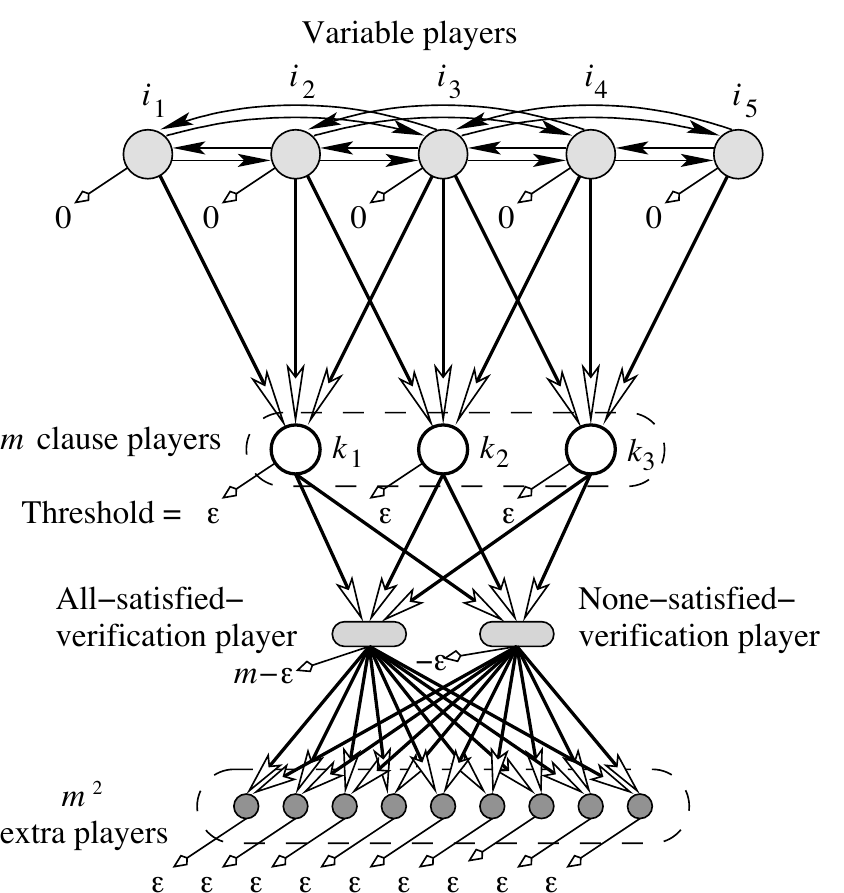}
\caption{Illustration of the NP-hardness reduction (Theorem~\ref{thm:UniqueHard}). For the monotone one-in-three SAT instance of Figure~\ref{fig:m13sat1}, we first obtain the same construction as in Theorem~\ref{thm:Extension1Hard}. We add two extra players, the all-satisfied-verification player and the none-satisfied-verification player, whose tasks are to verify if all clauses are satisfied and if no clause is satisfied, respectively. These two players are connected to $m^2$ extra players.}
\label{fig:m13sat3}
\end{figure}

In addition, add $k = m^2$ new players, named \textit{extra players}, and let these players constitute the set of designated players. Assign a threshold value of $\epsilon$ to each of these extra players, and introduce new arcs, each with weight 1, from the all-satisfied-verification player and the none-satisfied-verification player to every extra player. The resulting LIG is the instance $J$ of the problem in question. 

Note that at any PSNE the all-satisfied-verification player plays 1 if and only if every clause player plays 1, and the none-satisfied-verification player plays 1 if and only if no clause player plays 1. 
Furthermore, at any PSNE, each extra player plays 1 if and only if either every clause player plays 1 or no clause player plays 1.
Therefore, we find that every extra player playing $1$, the none-satisfied-verification player playing $1$, and every other player playing $0$ is a PSNE, and we denote this equilibrium by $E_0$. We claim that there exists a different PSNE where every extra player plays $1$ if and only if $I$ has a solution.

Suppose that there exists a solution $S_I$ to $I$. It can be verified that making the all-satisfied-verification player play 1, none-satisfied-verification player play 0, every extra player play 1, and choosing the actions of the clause and the variable players according to the corresponding truth values in $S_I$ gives us a PSNE that we call $E_1$. Thus $J$ has two PSNE $E_0$ and $E_1$, where the $k$ extra players play 1 in both cases.

Considering the reverse direction, suppose that there exists no solution to $I$. This implies that at any PSNE in $J$ all clause players can never play 1, otherwise we could have translated the PSNE to a satisfying truth assignment for $I$. This further implies that the all-satisfied-verification player always plays 0. The none-satisfied-verification player plays 1 if and only if none of the clause players plays 1. Thus, every extra player plays 1 if and only if no clause player plays 1, if and only if no variable player plays 1. 
Therefore, $E_0$ is the only PSNE in $J$ with the $k$ extra players playing 1. \end{proof}

\subsection{Heuristics for Computing and Counting Equilibria}

The fundamental computational problem at hand is that of computing PSNE in LIGs. We just saw that various computational questions pertaining to LIGs on general graphs, sometimes even on bipartite graphs, are NP-hard. We now present a heuristic to compute PSNE of an LIG on a general graph.

%
%

A natural approach to finding all the PSNE in an LIG would be to perform a backtracking search. However, a standard instantiation of the backtracking search method \citep[Ch. 5]{russellandnorvig03} that ignores the structure of the graph would be destined to failure in practice. Thus, we need to order the node selections in a way that would facilitate pruning the search space.

An outline of the backtracking search procedure that we used is given below. Here, the two main additions to the standard backtracking search method are exploiting the graph, including the influence factors, for node selection and implementing constraint propagation by adapting the NashProp algorithm \citep{ortizandkearns03} to run in polynomial time.
 
The first node selected by the procedure is  a node with the maximum outdegree. 
Afterwards, we do not select nodes by their degrees. 
We rather select a node $i$ that will most likely show that the current partial joint action cannot lead to a PSNE and explore the two actions of $i$, $x_i \in \{-1, 1\}$ in a suitable order. A good node selection heuristic that has worked well in our experiments is to select the one that has the maximum influence on any of the already selected nodes. 

%
Suppose that the nodes are selected in the order $1, 2, ..., n$ (wlog).
After selecting node ${i+1}$ and assigning it an action $x_{i+1}$, we determine if the partial joint action $\x_{1:(i+1)} \equiv (x_1, \ldots, x_{i+1})$ can possibly lead to a PSNE and prune the corresponding search space if not.
Note that a ``no'' answer to this requires a proof that one of the players $j$, $1 \le j \le i+1$, can never play $x_j$ according to the partial joint action $\x_{1:(i+1)} $. A straightforward way of doing this is to consider each player $j$, $1 \le j \le i+1$, and compute the quantities
$\gamma_j^+ \equiv \sum_{k=1,k \neq j}^{i+1} x_k w_{kj} + \sum_{k=i+2}^n \left| x_k w_{kj} \right|$ 
and $\gamma_j^- \equiv \sum_{k=1,k \neq j}^{i+1} x_k w_{kj} - \sum_{k=i+2}^n \left| x_k w_{kj} \right|$, 
and then test if the logical expression $((\gamma_j^- > b_j) \wedge x_j = -1) \vee ((\gamma_j^+ < b_j) \wedge x_j = 1)$ holds, in which case we can discard the partial joint action $\x_{1:(i+1)}$ and prune the corresponding search space. 
Furthermore, it may happen that due to $\x_{1:(i+1)}$, the choices of some of the not-yet-selected players became restricted. To this end, we apply {\bf NashProp}~\citep{ortizandkearns03} with $\x_{1:(i+1)}$ as the starting configuration, and see if the choices of the other players became restricted because of $\x_{1:(i+1)}$. Although each round of updating the table messages in {\bf NashProp} takes exponential time in the maximum degree in general graphical games, we can show in a way similar to Theorem~\ref{thm:pne_tree} that we can adapt the table updates to the case of LIGs so that it takes polynomial time. 

\subsubsection*{A Divide-and-Conquer Approach} To further exploit the structure of the graph in computing the PSNE, we propose a divide-and-conquer approach that relies on the following separation property of LIGs.

\begin{property}
\label{prop:separation}
Let $G = (V, E)$ be the underlying graph of an LIG and $S$ be a vertex separator of $G$ such that removing $S$ from $G$ results in $k \ge 2$ disconnected components: $G_1 = (V_1, E_1)$, ..., $G_k = (V_k, E_k)$. Let $G_i^{\prime}$ be the subgraph of $G$ induced by $V_i \cup S$, for $1 \le i \le k$. Consider the LIGs on these (smaller) graphs $G_i^{\prime}$'s, where we retain all the weights of the original graph, except that we treat the nodes in $S$ to be indifferent (that is, we remove all the incoming arcs to these nodes and set their thresholds to $0$). Computing the set of PSNE on $G_i^{\prime}$'s 
and then merging the PSNE 
(by performing outer-joins of joint actions and testing for PSNE in the original LIG), we obtain the set of all PSNE of the original game.
\end{property}

\noindent \textit{Proof Sketch.}
First, since the joint actions are tested for PSNE in the original LIG, the output will never contain a joint action that is not a PSNE. Second, since the nodes in $S$ are made indifferent in the LIGs on $G_i^{\prime}$, $1 \le i \le k$, no PSNE of the original LIG can get omitted from the result of the outer-join operation.\qed
\ \\

To obtain a vertex separator, we first find an edge separator (using well-known tools such as METIS~\citep{metis}), and then convert the edge separator to a vertex separator (by computing a maximum matching on the bipartite graph spanned by the edge separator). We then use this vertex separator to compute PSNE of the game in the way outlined in Property~\ref{prop:separation}. The benefits of this approach are two-fold: (1) for graphs that have good separation properties (such as preferential-attachment graphs), we found this approach to be computationally effective in practice; and (2) this approach leads to an \textit{anytime algorithm} for enumerating or counting PSNE: Observe that ignoring some edges from the edge separator may result in a smaller vertex separator, which greatly reduces the computation time of the divide-and-conquer algorithm at the expense of producing only a subset of all PSNE. (The reason we obtain a subset of all PSNE is that the edges that are ignored from the edge separator are not permanently removed from the original graph, and that after merging, every resulting joint action is tested for PSNE in the \emph{original} game, not in the game where some of the edges were temporarily removed. As a result, some the original PSNE may not be included in the final output. At the same time, we can never have a joint action in the final output that is not a PSNE.) We can obtain progressively better result as we ignore less number of edges from the edge separator.

\section{Computing the Most Influential Nodes}
\label{SecInfluentialNodes}
We now focus on the problem of computing the most influential set of nodes with respect to a specified desired PSNE and a preference for sets of minimal size. 
In the discussion below, we also assume, \textit{only} for the purpose of establishing and describing the equivalence to the \emph{minimum hitting set problem}~\citep{karp72}, that we are given the set of all PSNE. (As we will see, a counting routine is all that our algorithm requires, not a complete list of PSNE.)
We give a hypergraph representation of this problem that would lead us to a logarithmic-factor approximation by a natural greedy algorithm. 

\begin{figure}[h]
\centering
\includegraphics[width=2in]{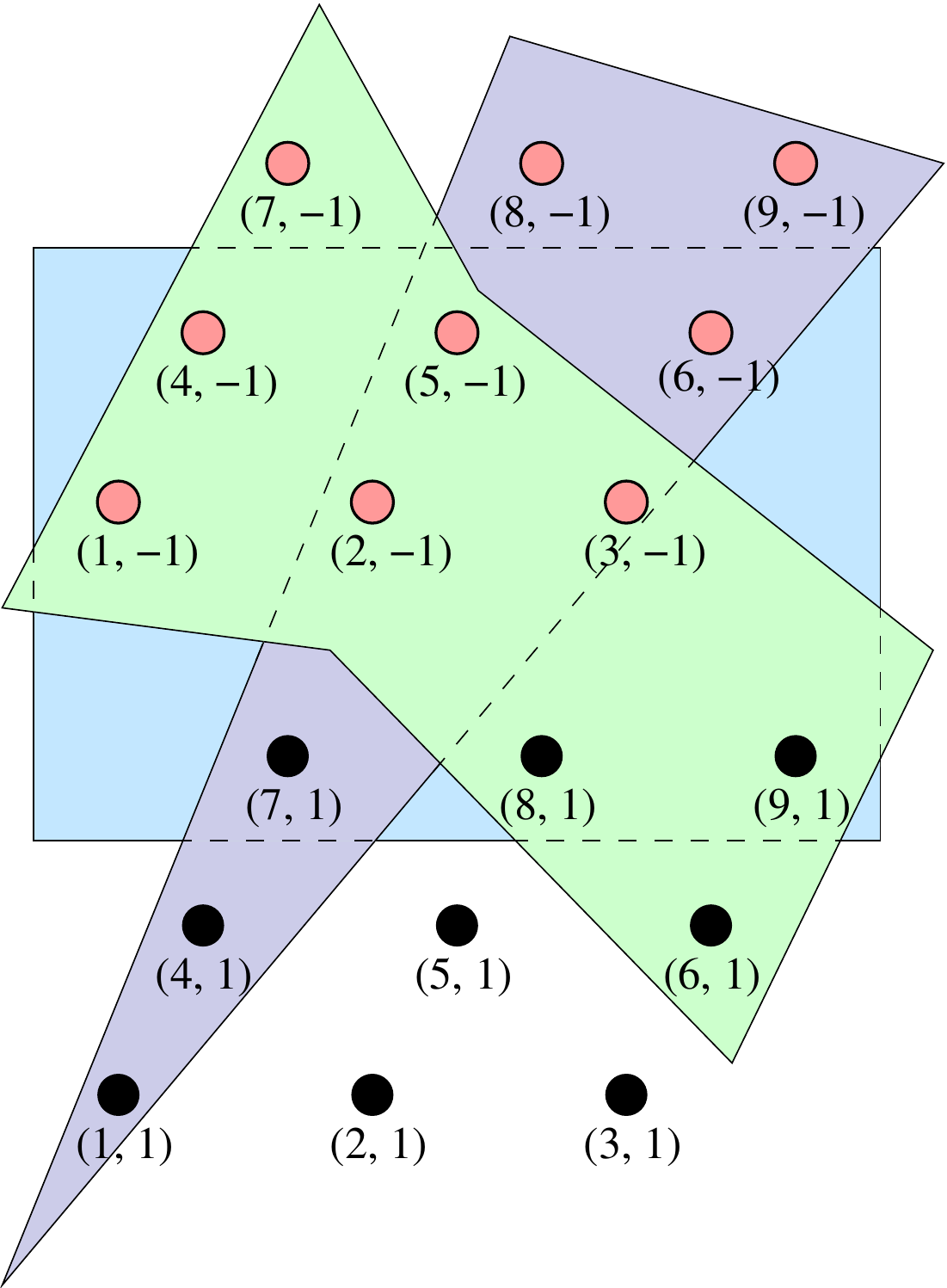}
\caption{A hypergraph representation of three PSNE in a 9-player game with binary actions. The PSNE shown here are the followings: $(1, -1, -1, 1, -1, -1, 1, -1, -1)$ (triangle), $(-1, -1, -1, -1, -1, -1, 1, 1, 1)$ (rectangle), and $(-1, -1, -1, -1, -1, 1, -1, 1, 1)$ (6-gon).}
\label{fig:hyper}
\end{figure}

Let us start by building a hypergraph that can represent the PSNE of a binary-action 
 game. The nodes of this hypergraph are the player-action tuples of the game. Thus, for an $n$-player, binary-action game, we have $2n$ nodes in the hypergraph. That is, for each player $i$ of the game, there are two nodes in the hypergraph: one in which $i$ plays $-1$ (tuple $(i, -1)$, colored red in Figure~\ref{fig:hyper}) and the other in which $i$ plays 1 (tuple $(i, 1)$, colored black).
For every PSNE $\x$ we construct a hyperedge $\{(i, x_i)\ |\ 1 \le i \le n\}$.
Let us call this hypergraph the {\em game hypergraph}. By construction, a set of players $S$ play the same joint-action ${\mathbf a}_S \in \{-1, 1\}^{|S|}$ in two distinct PSNE $\x$ and $\y$ of the LIG if and only if both of the corresponding hyperedges $e_{\x}$ and $e_{\y}$ (resp.) of the game hypergraph contains $T = \{(i , a_i)\ |\ i \in S\}$. 

We can use the above property to translate the most-influential-nodes selection problem, given all PSNE, to an equivalent combinatorial problem on the corresponding game hypergraph $H$. Let $e_{\x^*}$ be the hyperedge in $H$ corresponding to the desired PSNE $\x^*$. Let us call $e_{\x^*}$ the {\em goal hyperedge}. Then the most-influential-nodes selection problem is the problem of selecting a minimum-cardinality set of nodes $T \subseteq e_{\x^*}$ such that $T$ is contained in no other hyperedge of $H$ (recall that we are dealing with a set-preference function that captures the preference for sets of minimal cardinality). 
Let us call the latter problem the {\em unique hyperedge problem}. 
Using the notation above, the equivalence relationship between the influential nodes selection problem (given the set of all PSNE) and the unique hyperedge problem can be stated as follows. The set $S \subseteq \{1, ..., n\}$ is a (feasible) solution to the most-influential-nodes selection problem if and only if $T = \{(i, {x_i}^*)\ |\ i \in S\}$ is a (feasible) solution to the unique hyperedge problem.

We now show that the unique hyperedge problem is equivalent to the {minimum hitting set problem}. Immediate consequences of this result are that the unique hyperedge problem is not approximable within a factor of $c \log h$ for some constant $c > 0$, and that it admits a $(1 + \log h)$-factor approximation~\citep{razsafra97,setcoverapprox74}, where $h$ is the total number of PSNE. 

\begin{theorem}
\label{thm:hittingset2}
The unique hyperedge problem having $2n$ players and $h$ hyperedges is equivalent to the minimum hitting set problem having $n$ nodes and $h$ hyperedges.
\end{theorem}
%
\begin{proof}
Let us consider an instance $I$ of the unique hyperedge problem, given by a game hypergraph $G = (V, E)$, where $V$ is the set of $2n$ nodes and $E$ is the set of $h$ hyperedges, along with a specification of the goal hyperedge $e_{\x^*}$. 
Given $I$, we now construct an instance $J$ of the minimum hitting set problem, specified by the hypergraph $G^\prime = (e_{\x^*}, \{e_{\x^*}\} \cup \{\bar{e} \cap e_{\x^*} \ |\ e \in E$ and $e \neq e_{\x^*}\})$, where $\bar{e}$ indicates the complement set of the hyperedge $e$. Thus, the nodes of 
$G^\prime$ are exactly the $n$ nodes of $e_{\x^*}$ and the hyperedges of it are constructed from the complement hyperedges of
$G$ except $e_{\x^*}$, which is present in both $G$ and $G^\prime$. We show that a set $S$ of nodes is a feasible solution to $I$ if and only if it is a feasible solution to $J$.

If $S$ is a feasible solution to $I$ then $S \subseteq e_{\x^*}$ (because in the unique hyperedge problem, we are only allowed to select nodes from the goal hyperedge) and $S \nsubseteq e$ for any hyperedge $e \neq e_{\x^*}$ of $G$ (otherwise, the uniqueness property is violated). This implies that for any hyperedge $e \neq e_{\x^*}$ of $G$, there exists a node $v \in S$ such that $v \notin e$, which further implies that $v \in \bar{e} \cap e_{\x^*}$. Thus,  every hyperedge of $G^\prime$, including $e_{\x^*}$, of course, has at least one of its nodes selected in $S$, and therefore, $S$ is a feasible solution to $J$. On the other hand, if $S$ is a feasible solution to $J$ then for any hyperedge of $G^\prime$, at least one of its nodes has been selected in $S$. That is, for any hyperedge $e \neq e_{\x^*}$ of $G$, we have $e^\prime \equiv \bar{e} \cap e_{\x^*}$ as the corresponding complementary hyperedge in $G^\prime$, and there exists a node $v \in S$ such that $v \in e^\prime$, which implies that $v \notin e$. Thus, $S \nsubseteq e$ for any hyperedge $e \neq e_{\x^*}$ of $G$. Furthermore, we have selected all the nodes of $S$ from $e_{\x^*}$ of $G$. Thus, $e_{\x^*}$ is the unique hyperedge of $G$ containing the nodes of $S$.

To prove the reverse direction, we start with an instance $J$ of the minimum hitting set problem, specified by a hypergraph $G^\prime = (V, E)$, where $V$ is a set of $n$ nodes and $E$ is a set of $h$ hyperedges. Without the loss of generality, we assume that $E$ contains the hyperedge $e^*$ consisting of all the nodes of $V$.
We now construct an instance $I$ of the unique hyperedge problem that has a hypergraph $G$ with $2n$ nodes and $h$ hyperedges. The node set of $G$ literally consists of two copies of the nodes of $V$, denoted by $V \times \{1, -1\}$. We now construct the hyperedges of $G$. For each hyperedge $e \neq e^*$ of the minimum hitting set instance, we include a hyperedge $e^\prime \equiv \bar{e}\times\{1\} \cup e \times \{-1\}$ in $G$, and for the hyperedge $e^*$ of $J$, we include the hyperedge $e^* \times \{1\}$ in $G$. Thus, the game hypergraph can be defined as $G = ( V \times \{1, -1\} , \{e^* \times \{1\}\} \cup \{\bar{e} \times \{1\} \cup e \times \{-1\}\ |\ e \in E$ and $e \neq e^*\})$. Finally, we designate  $e^* \times \{1\}$ as the goal hyperedge of $I$.
We will show that $S \subseteq V$ is a feasible solution to $J$ if and only if $S \times \{1\}$ is a feasible solution to the unique hyperedge problem instance $I$. The set $S$ is a feasible solution to $J$ if and only if for every hyperedge $e \neq e^*$ of $G^\prime$, there exists a node $v \in S$ such that $v \in e$ (note that $S \subseteq e^*$). This is equivalent to saying that for every hyperedge $e \times\{1\} \neq e^* \times \{1\}$ of $G$, there exists a node $v \in S \times \{1\}$ such that $v \notin e \times\{1\}$. Using the fact that $S \times \{1\} \subseteq e^* \times \{1\}$, $S \times \{1\}$ is a feasible solution to $I$.
\end{proof}

The adaptation of the well-known hitting set approximation algorithm \citep{hittingset80} for our problem can be outlined as follows: At each step, select the least-degree node $v$ of the goal hyperedge, remove the hyperedges that do not contain $v$, remove $v$ from the game hypergraph, and include $v$ in the solution set, until the goal hyperedge becomes the last remaining hyperedge in the hypergraph. In the context of the original LIG, at every round, this algorithm is essentially picking the node whose assignment would reduce the set of PSNE consistent with the current partial assignment the most.  Hence, the algorithm only requires a subroutine to \emph{count} the PSNE extensions for some given partial assignment to the players' actions, not an \emph{a priori} full list or enumeration of all the PSNE. Of course, it may require a complete list of PSNE in the worst case.

\section{Experimental Results}
We performed empirical studies on several types of LIGs, namely, (1)
random LIGs (Erd\"{o}s-R\'{e}nyi and uniformly random), (2)
preferential-attachment LIGs, and (3) LIGs created to model potential interactions in two different real-world scenarios: interactions among U.S. Supreme Court justices and those among U.S. senators.
While the first two types of LIGs are synthetic/artificial, the latter
two are the result of inferring the LIGs from real-world data using machine learning techniques~\citep{honorio_and_ortiz13}. 

The reason for experimenting with synthetic LIGs using Erd\"{o}s-R\'{e}nyi
\citep{erdos_random} and uniformly random graphs is that those types
of network-generation techniques are the most basic available. They
serve as precursors to the more sophisticated preferential-attachment
graphs, which capture the heavy-tailed degree distribution often observed in real-world networks \citep{preferentialattachment}. 

Here is our overall plan for this section. For the synthetic games
(random LIGs and preferential-attachment LIGs), we generate game
instances by varying the appropriate parameters, such as the size of
the game, and evaluate both the number of PSNE of these games and the
computation time of our algorithm. We also compute the most
influential nodes in these games using our approximation algorithm and
compare it to the optimal (i.e., minimum-cardinality) set of
most-influential nodes. For the real-world games on Supreme Court
rulings and congressional voting, we discuss how we learn these games
from data \citep{honorio_and_ortiz13}, and how we compute the set of
PSNE and identify the most influential nodes. For the congressional
voting case, we also \emph{adapt} the simple, greedy selection approximation
algorithm to identify most-influential individuals using the cascade
model in the diffusion setting, as described by~\citet{kleinberg07}, to use as a \emph{heuristic} to
solve \emph{specific instances} of the most-influential-nodes problem
formulation in our context. We compare and contrast the output of our
approximation algorithm and that of the diffusion-based
heuristic resulting from the adaptation. 

\subsection{Random Influence Games}
We began our experiments by generating instances of random graphs using the Erd\"{o}s-R\'{e}nyi model \citep{erdos_random}.
We varied the number of nodes, from 10 to 30, and the probability of including an edge. 
Assuming binary actions, $1$ and $-1$, we chose the threshold $b_i$
and the influence factors $w_{ji}$ of the incoming arcs of each node
$i$ uniformly at random from a unit hyperball. That is, for each node
$i$, $b_i^2 + \sum_{j \in N(i)} w_{ji}^2 = 1$, where $N(i)$ is the set
of nodes having arcs toward $i$. Then, we chose the sign of each
threshold, as well as each weight, to be either $+$ or $-$ with $0.5$ probability.
We applied the heuristic given earlier to find the set of all PSNE in these random graphs. 
Our experimental results show that in all of these random LIGs, the number of
PSNE is almost always very small---usually one or two, and sometimes none. 

We also studied LIGs on uniformly random directed graphs.
While constructing the random graphs, we have independently chosen each arc with probability $0.50$, and assigned it a weight of $-1$ with probability $p$ (named \textit{flip probability}) and $1$ with probability $1-p$. 
Several interesting findings emerged from our study of this parameterized family of LIGs on uniformly random graphs. 
Appendix~\ref{app:exp} summarizes the results in tabular form. 
For various flip probabilities, we independently generated $100$ uniformly random graphs of $25$ nodes each. For each of these random graphs, we first computed all PSNE using our heuristic. We then applied the greedy approximation algorithm to obtain a set of the most influential nodes in each graph and compared the approximation results to the optimal set. 

Unless $p$ is either $0$ or $1$, one cannot guarantee the existence of
a PSNE. In our experiments, we found that, in fact, for $p = 0.50$,
the probability of not having a PSNE is highest (around $5\%$), and as
we go toward the two extremes of $p$, the probability of not having a
PSNE decreases. We report only the games with at least one equilibrium in this experimental study, 
because it is these games we care about for computing the most
influential nodes. Another interesting finding about the number of
PSNE is that it is very small when $p = 0$, that is when all the arcs
have weight $1$, and it is large when $p = 1$, although quite small
(on average, a fraction $5.81 \times 10^{-6} \approx 2^{-17.29}$)
relative to the total number of $2^{25}$ possible joint actions. Also,
the average number of nodes of the search tree that the backtracking
method visits per equilibrium computation is relatively small on the
two extremes of $p$, compared with $p$ around $0.5$. Note that the
backtracking method does a very good job with respect to the number of
search-tree nodes visited in searching the $2^{25}$ space. In fact,
our experiments show that the addition of the {\bf NashProp}-based heuristic on top of the node selection heuristic considerably speeds up the search. Finally, we found that although the approximation algorithm has a logarithmic factor worst-case bound, the results of the approximation algorithm are most often very close to the optimal solution.

\begin{figure}[h]
\centering
\includegraphics[width=4in]{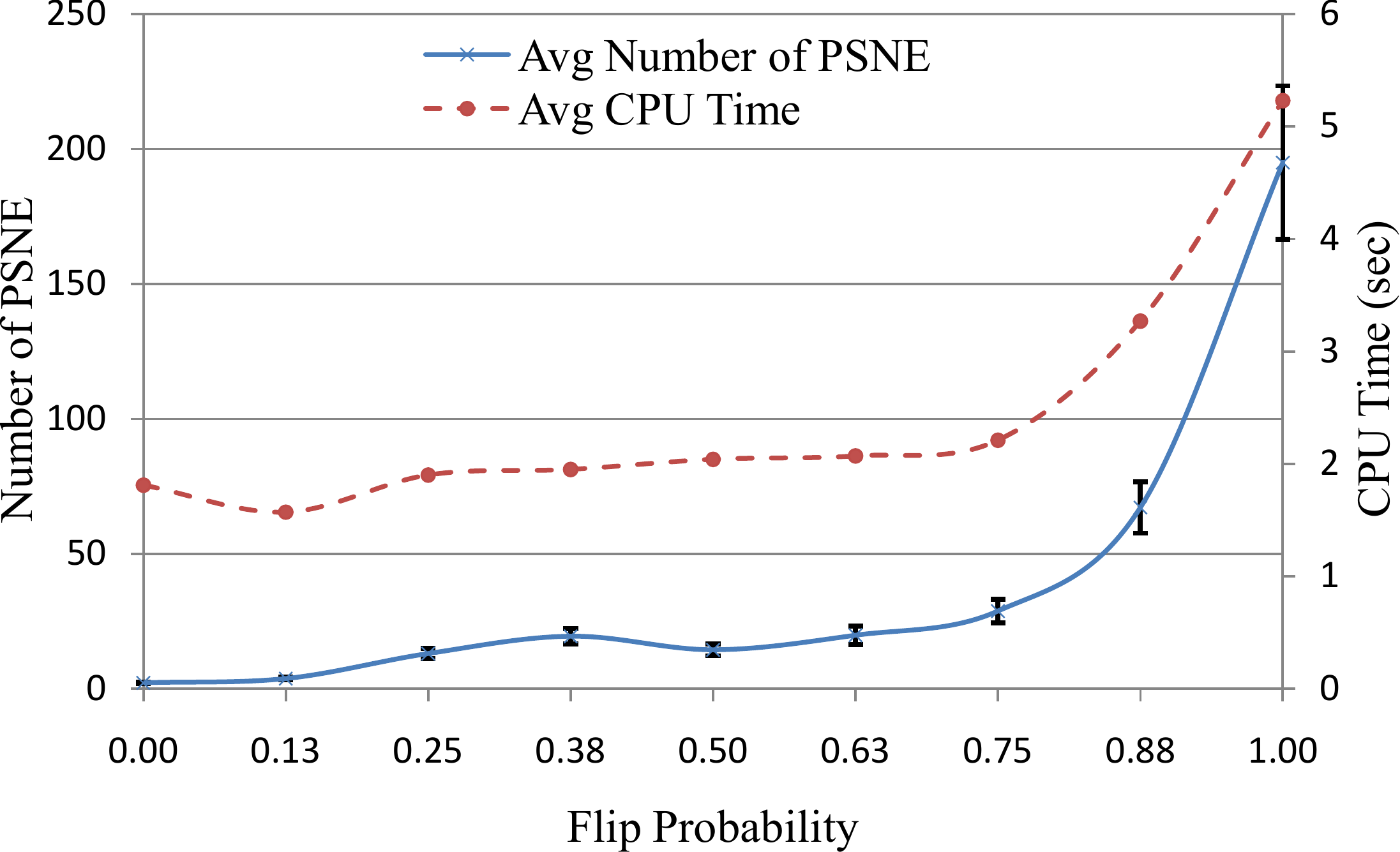}
\caption{PSNE computation on random LIGs. The vertical bars denote 95\% confidence intervals.}
\label{fig:plot_random}
\end{figure}

Figure~\ref{fig:plot_random} shows that the number of PSNE usually increases if we have more negative-weighted arcs than positive ones; although the number of PSNE is still very small relative to the maximum potential number, as remarked earlier. 
We also found once again that although the approximation algorithm for
the influential-nodes selection problem has a logarithmic factor
worst-case bound, the result of the approximation algorithm is most
often very close to the optimal solution. For example, for the random
games having all negative influence factors, in 87\% of the trials the
approximate solution size $\le$ optimal size $+ 1$, and in 99\% of the
trials the approximate solution size $\le$ optimal size $+ 2$ (see Appendix~\ref{app:exp} for more details in a tabular form).
%
\subsection{Preferential-Attachment LIGs}
We also experimented with LIGs based on preferential-attachment graphs
primarily because of its power to explain the structure of many
real-world social networks in a generative
fashion~\citep{preferentialattachment}. In order to construct these
graphs, we started with three nodes in a triangle and then
progressively added each node to the graph, connecting it with three
existing nodes with probabilities proportionate to the degrees. We
made each connection bidirectional and imposed the same weighting
scheme as above: with flip probability $p$, the weight of an arc is
$-1$ and with probability $1-p$ it is $1$. We set the threshold of each node to $0$. We observe that for $0 < p < 1$, these games have very few PSNE, while for $p = 0$ and $p = 1$ the number of PSNE is considerably larger than that. 
Furthermore, these games show very good separation properties, making
the computation amenable to the divide-and-conquer approach. We show
the average number of PSNE and the average computation time for graphs
of sizes 20 to 50 nodes in Figure~\ref{fig:plot_pref} for $p = 1$
(each average is over 20 trials). Note that in contrast to uniformly random LIGs, preferential-attachment graphs show an exponential increase in the number of PSNE as the number of nodes increase, although the number of PSNE is still a very small fraction of the maximum potential number.

\begin{figure}[h]
\centering
\includegraphics[width=4in]{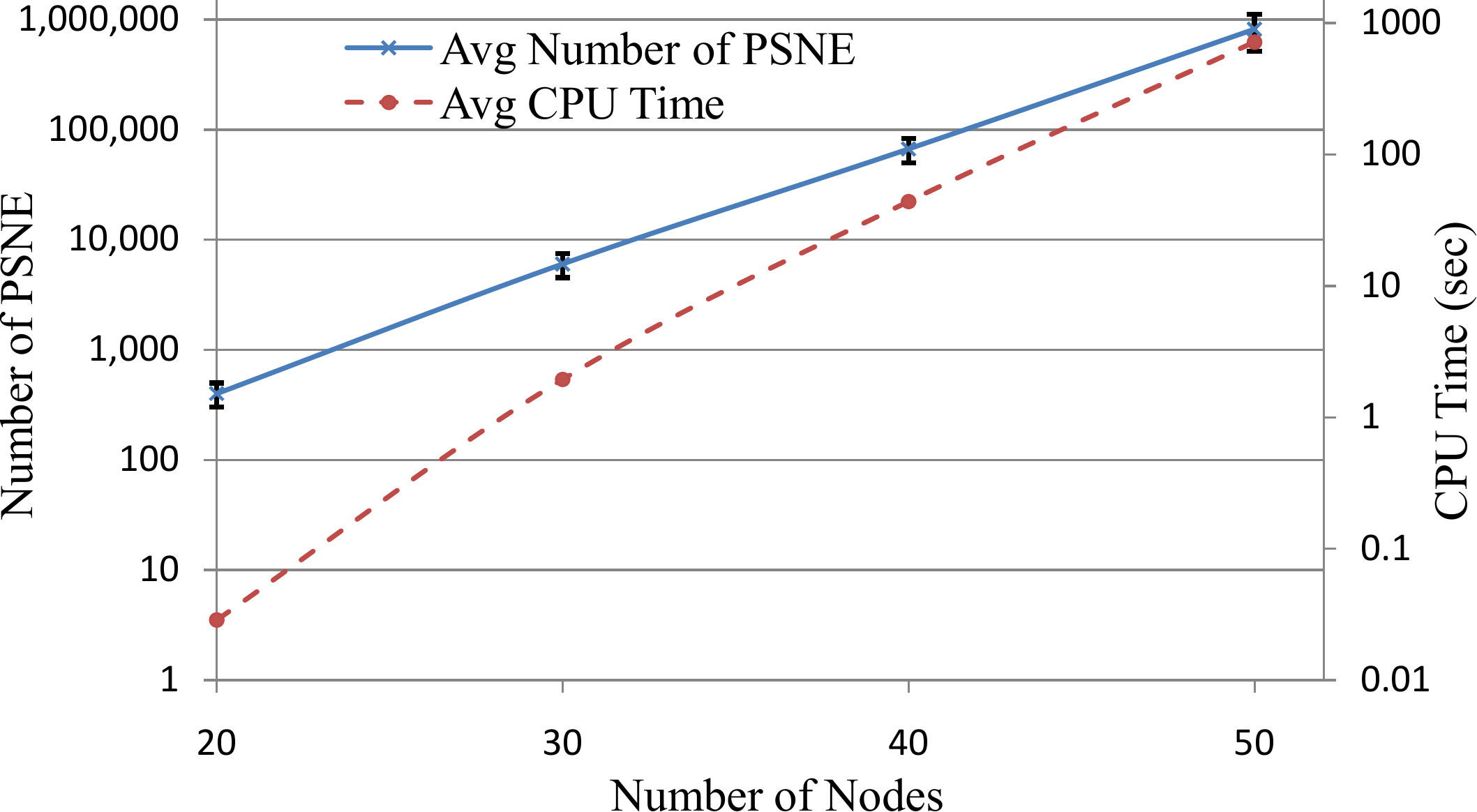}
\caption{PSNE computation on preferential-attachment LIGs ($y$-axis is in $\log$ scale). The vertical bars denote 95\% confidence intervals.}
\label{fig:plot_pref}
\end{figure}

\subsection{Illustration: Supreme Court Rulings}
We used our model to analyze the influence among the justices of the
U.S. Supreme Court. The Supreme Court of the United
States (SCOTUS),~\footnote{\url{http://www.supremecourtus.gov}} is the highest
federal court of the judicial branch of the government. (The other
branches of government are the executive branch, lead by the President, and the
legislative branch, represented by Congress.)  The SCOTUS is the main
interpreter of the Constitution and has final say on the
constitutionality of any federal law created by the legislative branch
or any action taken by the executive branch. It consists of nine
justices---a chief justice and eight associate justices. We chose to
study the SCOTUS in the context of influence games because this is an application domain where the strategic aspects of influence seem of prime importance.

There are two distinctive features that make our approach particularly
suitable to the SCOTUS domain. First, we can model the individual
outcomes (in this case, the decisions of the justices on each case) as
outcomes of a one-shot non-cooperative game (an LIG in our
case). Second, the physical interpretation of the diffusion process is
not as clear in this setting as it is in applications like viral marketing.

\subsubsection{Data}
We obtained data from the Supreme Court Database.~\footnote{\url{http://scdb.wustl.edu/}}
Although the database captures fine-grained details of the cases, we
only focused on the variable varVote. 
Again, the votes of the justices are not simple yes/no instances. Instead, each vote can have eight distinct values. However, for practical purposes, we can attach a simple yes/no interpretation to the values of the votes, as shown in Table~\ref{tab:InterpretationOfVotes}.

\begin{table}[!h]
	\centering
		\begin{tabular}{ c p{7cm} c }
			varVote	& Original Meaning	& Our Interpretation \\
			\hline
			\hline
			1	&	Voted with majority	& Yes \\
			2 &	Dissent							& No \\
			3	&	Regular concurrence	&	Yes \\
			4	& Special concurrence	& Yes \\
			5 &	Judgment of the Court	& Yes \\
			6 &	Dissent from a denial or dismissal of certiorari, or dissent from summary affirmation of
			an appeal	(Interpreted as absent from voting in final outcome) & Majority\\
			7 &	Jurisdictional dissent	(Interpreted as absent from voting in final outcome) & Majority\\
			8 &	Justice participated in an equally divided vote	& ---		\\	
		\end{tabular}
	\caption{Interpretation of Votes}
	\label{tab:InterpretationOfVotes}
\end{table}

In Table~\ref{tab:InterpretationOfVotes}, ``majority'' in the third column signifies that we interpreted the corresponding justice's vote as yes or no, whichever occurs most among the other justices. Also, among the natural courts we studied, we did not encounter voting instances where varVote has a value of 8.

We now present our study of the natural court (with timeline 1994--2004) comprising of Justices 
WH Rehnquist,
JP Stevens,
SD O'Connor,
A Scalia, 
AM Kennedy,
DH Souter,
C Thomas,
RB Ginsburg, and
SG Breyer.

\subsubsection{Learning LIG}
The data for the above natural court consists of 971 voting instances
(each voting instance consists of the votes of all nine
justices). Many instances (i.e., voting patterns) appear repeatedly in
the data set, of course. For example, the most repeated instance
consists of all the justices voting yes, occurring 438 times. The second
most repeated instance, occurring 85 times, consists of five of the
justices, namely, Justices Scalia, Thomas, Rehnquist, O'Connor, and
Kennedy voting yes, and the remaining justices voting no. We used
$L_2$-regularized logistic regression (simultaneous classification) to
learn an LIG from this data. The influence factors and the biases of
the learned LIG appear in tabular form
in Appendix~\ref{app:exp}. Figure~\ref{fig:LIG} shows a pictorial representation of the same LIG.

\begin{figure}[!h]
	\centering
		\includegraphics[width=4in]{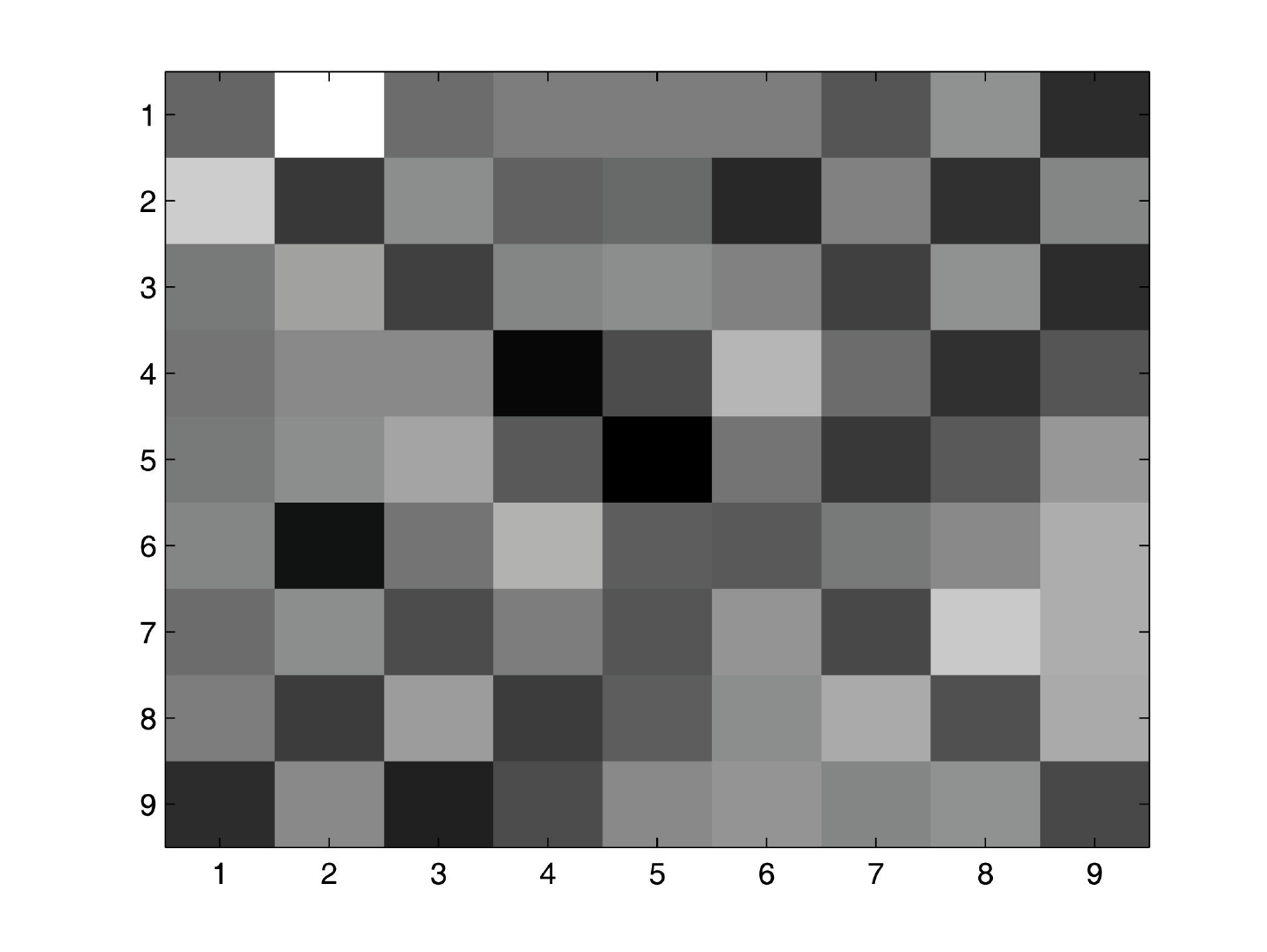}
	\caption{Pictorial representation of LIG learned from data---the non-diagonal elements represent influence factors and the diagonal elements biases. The numbering of the players (from $1$ to $9$) corresponds to the justices in this order: 
Justices A Scalia, 
C Thomas,
	WH Rehnquist,
SD O'Connor,
AM Kennedy,
SG Breyer,
DH Souter,
RB Ginsburg, and
JP Stevens. The darker the color of a cell, the more negative the
corresponding number. For example, the most negative number
($-0.2634$) occurs in cell $(5,5)$ (i.e., the bias of Justice
Kennedy). The most positive number ($0.4282$) occurs in cell $(1,2)$
(i.e., the influence factor from Justice Scalia to Justice Thomas) and
the number closest to zero is $0.001$ in cell $(2,4)$ (i.e., the
influence factor from Justice Thomas to Justice Kennedy).	}
	\label{fig:LIG}
\end{figure}

The learned LIG represents 589 of the 971 voting instances as PSNE. As
expected, it represents the frequently repeated voting instances (such
as the ones mentioned above). Figure~\ref{fig:justice_graph} shows a
graphical representation of the LIG. We clustered the nodes based on
the traditional perception that Justices Scalia, Thomas, Rehnquist,
and O'Connor are ``conservative;''  
Justices Breyer, Souter, Ginsburg, and Stevens are ``liberal;'' and Justice Kennedy is a ``moderate.'' As illustrated in Figure~\ref{fig:justice_graph}, negative influence factors occur only between players of two different clusters. 

\begin{figure}[!h]
	\centering
		\includegraphics[height=5in]{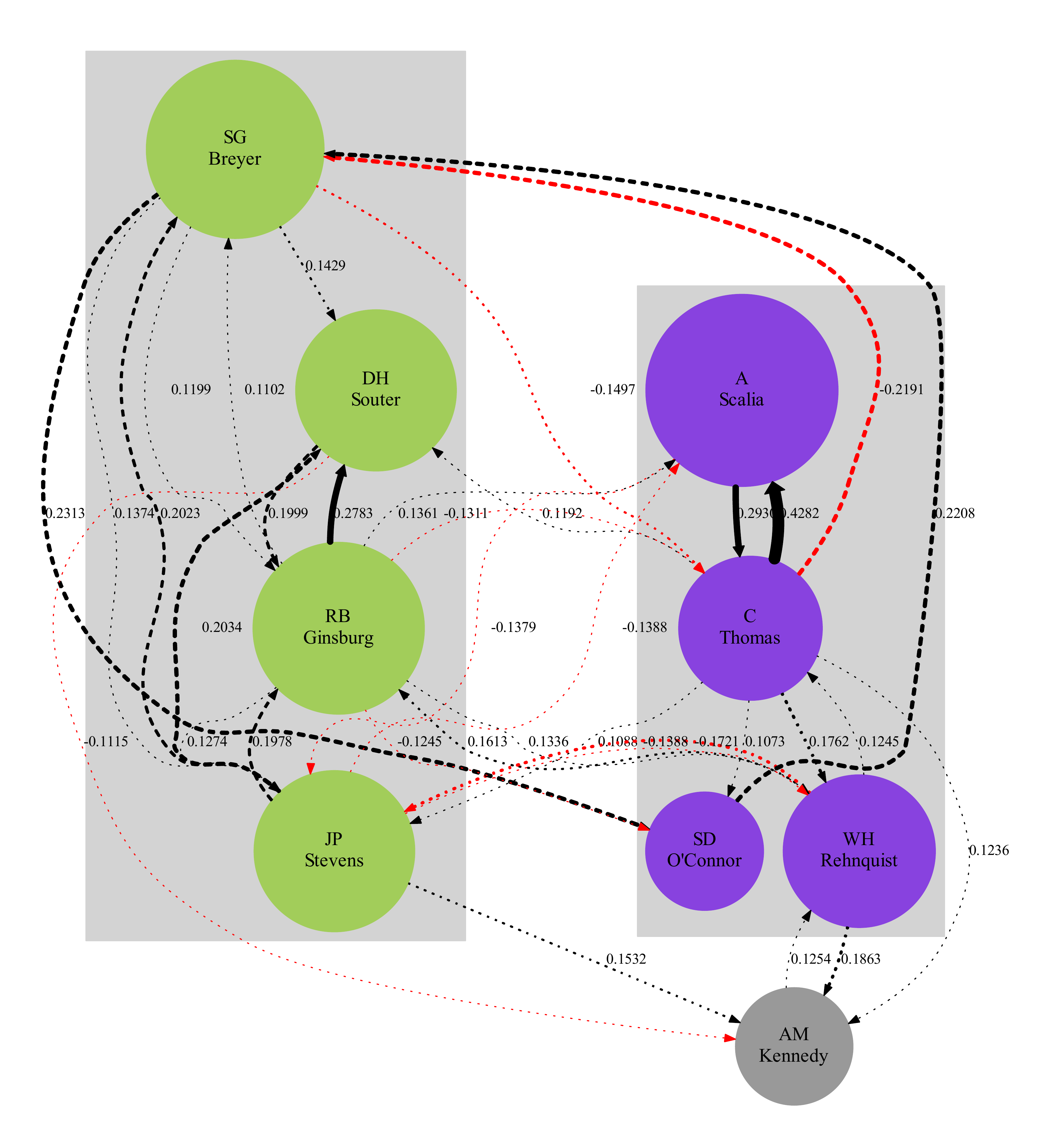}
	\caption{Graphical representation of LIG learned from
          data. Larger node sizes indicate higher thresholds (more
          stubborn). Black and red arcs indicate positive or negative
          influence factors, respectively. 
While the learned LIG is a complete graph, we are only showing approximately half of the arcs (i.e., we are not showing the ``weakest'' arcs in this graph).}
	\label{fig:justice_graph}
\end{figure}

\subsubsection{Most Influential Nodes}
Analysis of the PSNE of this LIG shows that there is a set of two nodes that is ``most influential'' with respect to achieving the objective of every justice voting yes. This most-influential set consists of one node from the set \{Scalia, Thomas\} and another one from the set \{Breyer, Souter, Ginsburg, Stevens\}. Furthermore, any one node from the set \{Breyer, Souter, Ginsburg, Stevens\} is alone most-influential with respect to achieving the objective of a 5-4 vote mentioned above (i.e., the second most repeated instance in the data).

\subsection{Illustration: Congressional Voting}
We further illustrate our computational scheme in another real-world
scenario---the U.S. Congress---where the strategic aspects of the
agents' behavior are also of prime importance. We particularly focus
on the U.S. Senate, which consists of 100 senators; two senators for each of the 50
U.S. states. Together, they form the most important unit of the
legislative branch of the U.S. Federal Government.

We first learned the LIGs among the senators of the 101st and the 110th U.S. Congress~\citep{honorio_and_ortiz13}.
The 101st Congress LIG consists of 100 nodes, each representing a senator, and 936 weighted arcs among these nodes. 
On the other hand, the 110th Congress LIG has the same number of
nodes, but it is a little sparser than the 101st one, having 762
arcs. In these LIGs, each node can play one of the two actions: $1$
(yes vote) and $-1$ (no vote). Figure~\ref{fig:lig_senate_110} shows a
bird's eye view of the 110th Congress LIG, while
Figure~\ref{fig:lig_senate_110_zoom} shows a magnified part of it.

\begin{figure}[!h]
	\centering
		\includegraphics[width=5in]{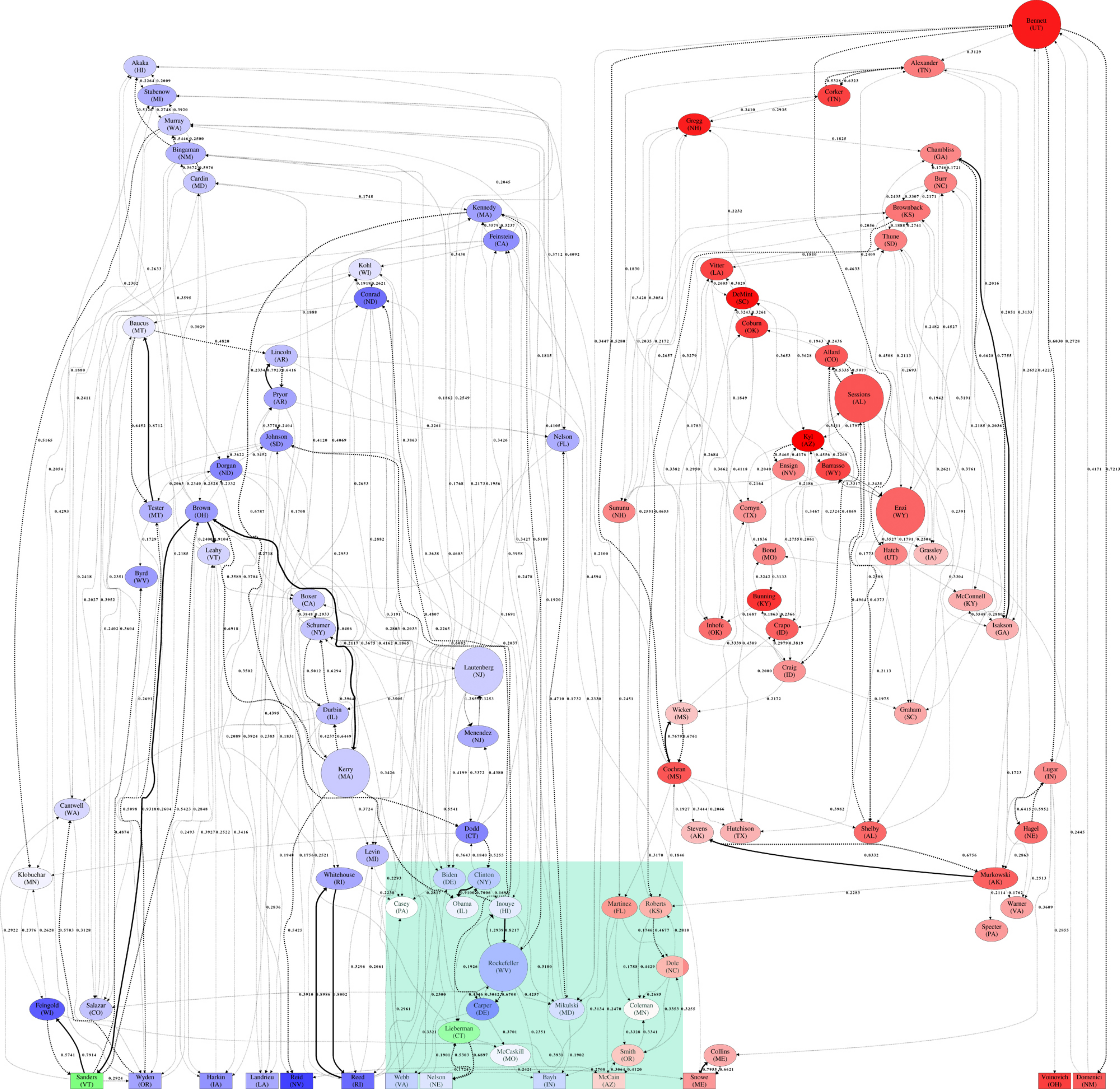}
	\caption{LIG for the 110th U.S. Congress: darker color of nodes represent higher threshold (more stubborn); thicker arcs denote influence factors of higher magnitude (only half of the original arcs with the highest magnitude of influence factors are shown here);
	circles denote most-influential senators;
rectangles denote cut nodes used in the divide-and-conquer
algorithm. Figure~\ref{fig:lig_senate_110_zoom} shows the shaded part for better visualization.}
	\label{fig:lig_senate_110}
\end{figure}

\begin{figure}[!h]
	\centering
		\includegraphics[width=4.5in]{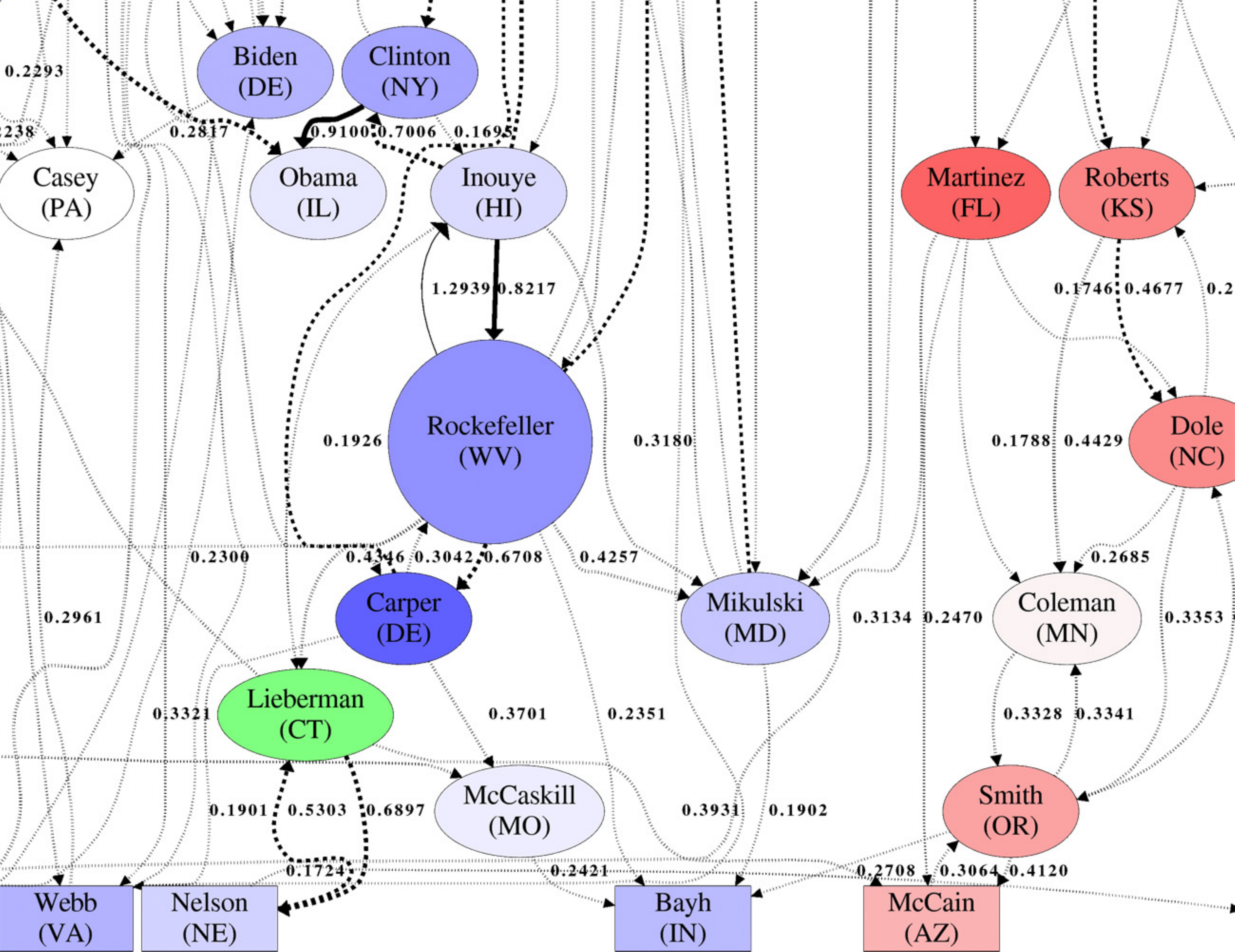}
	\caption{A part of the LIG for the 110th U.S. Congress: blue nodes represent Democrat senators, red Republican, and white independent; darker color of nodes represent higher threshold (more stubborn); thicker arcs denote influence factors of higher magnitude;
	circled node (Senator Rockefeller) denotes one of the most influential senators;
rectangles at the bottom denote cut nodes used in the divide-and-conquer algorithm.}
	\label{fig:lig_senate_110_zoom}
\end{figure}

First, we applied the divide-and-conquer algorithm that exploits the
nice separation properties of these LIGs, to find the set of all PSNE
(we precompute the whole set of PSNE for convenience; as discussed earlier, counting alone is sufficient).
We obtained a total of 143,601 PSNE for the 101st Congress graph and 310,608 PSNE for the 110th. Note that the number of PSNE in these games is extremely small (e.g., a fraction 
$2.45\times10^{-25} \approx 2^{-81.76}$ for the 110th Congress) relative to the maximum possible $2^{100}$ joint actions. Regarding computation time, solving the 110th Congress using the divide-and-conquer approach takes about seven hours, whereas solving the same without this approach, simply relying on the backtracking search, takes about 15 hours on a modern quad-core desktop computer.

Next, we computed the most influential senators using the
approximation algorithm outlined 
at the end of Section~\ref{SecInfluentialNodes}.
We obtained a solution of size five for the 101st Congress graph, which we verified as an optimal solution. This solution consists of Senators Rockefeller (Democrat, WV), Sarbanes (Democrat, MD), Thurmond (Republican, SC), Symms (Republican, ID), and Dole (Republican, KS). Interestingly, none of the maximum-degree nodes were selected.
Similarly, the six most influential senators of the more recent 110th Congress (January 2007--January 2009) are
Kerry (Democrat, MA), Bennett (Republican, UT), Sessions (Republican, AL), Enzi (Republican, WY), Rockefeller (Democrat, WV), and Lautenberg (Democrat, NJ).

We also applied our technique to the more recent 112th Congress, using
voting data from May 9, 2011 to August 23, 2012. A set of the most
influential senators with respect to the outcome of everyone voting
``yes'' consists of Senators Reid (Democrat, NV), Inouye (Democrat,
HI), Johnson (Republican, WI), Sanders (Independent, VT), Hagan
(Democrat, NC), Collins (Republican, ME), Crapo (Republican, ID),
DeMint (Republican, SC), Reed (Democrat, RI), and Barrasso
(Republican, WY). Note that the set of most-influential senators in
the 112th Congress consists of 10 senators, whereas it consists of
only six senators in the earlier Congresses that we studied. This
implies that, according to our model, for the 112th Congress, we now
need a broader group of ``influencing'' senators to lead everyone to a
consensus. This is consistent with the contemporary perception of
polarization in Congress, which has been highlighted both in the
mainstream media and formal research studies in recent times \citep{polarization2013}.

Besides identifying the most influential senators with respect to
passing a bill (e.g., the outcome of everyone voting ``yes''), we can
also apply our model to study the other extreme of \emph{not} passing
a bill (e.g, the outcome of everyone voting ``no''). According to our
model, we can achieve the latter outcome in the 112th Congress if the following 10 senators choose to vote ``no'': Senators Nelson (Democrat, FL), Cardin (Democrat, MD), Klobuchar (Democrat, MN), Reed (Democrat, RI), Murkowski (Republican, AK), Moran (Republican, KS), Vitter (Republican, LA), Enzi (Republican, WY), Crapo (Republican, ID), and DeMint (Republican, SC).

Later on, we will show that besides the desired outcomes of everyone
voting ``yes'' or everyone voting ``no,'' our model and approach
extend to studying even more general outcomes, such as breaking filibusters or preventing clotures.

\subsubsection{Adapting a
Popular Diffusion-Based
  Algorithm as Heuristic for 
Most-Influential-Nodes
Problem
  Instance 
in 
LIGs}
Our one-shot noncooperative game-theoretic approach is fundamentally
different from the diffusion approach, both
syntactically (i.e., mathematical foundations) and semantically (e.g.,
interpretation of objectives, nature of problem formulations, and applicable domains and contexts).
Of course, the most obvious key difference is that our \emph{solution
concept} does not consider dynamics. We concentrate on ``end-state''
behavior and characterize \emph{stable}
outcomes using the
 notion of PSNE, which is ``static'' by nature.~\footnote{Note that PSNE may,
 or may not, be reachable by (possibly networked) best-response dynamics.}
 As we reviewed in Section~\ref{sec:diffusion}, many of the diffusion-based approaches
 lead to the outcomes that are not stable.~\footnote{The reader should
   bear in mind that the intention
   behind our statements
   contrasting the approaches/models is simple: to highlight some
   of the differences. We do not mean to imply that one approach or
   model is ``better'' than the other: they are just different and
   largely incomparable, each with its own
   ``pros and cons'' depending on the problem, context, or domain of
   interest. In other words, each has its own place.}
 
Yet, out of purely scientific curiosity and suggestions/feedback about
our work, we present preliminary results based on a \emph{heuristic} for the problem of most
influential nodes in our context (i.e., LIGs) that
we 
designed by \emph{adapting} the arguably most popular and simple greedy-selection
algorithm for identifying ``most influential'' nodes developed
specifically for the cascade
model of diffusion~\citep{kleinberg07}. The simplicity of the
original algorithm in the diffusion context facilitates the adaptation
to our context.
Note however that this adaptation is exclusive to a \emph{very
  specific instance of our general problem formulation}: identifying ``most influential nodes'' \emph{with
 respect to the goal being maximizing the number of $+1$'s in the
 PSNE} (see Definition~\ref{def:probform} in
Section~\ref{sec:probform} for a definition and discussion of our problem
formulation and the role that the ``objective function'' $g(.)$ plays
in the formulation.)


We realize that, inspired mostly by the \emph{cascade model} described in 
 \citet{kleinberg07}, the literature on diffusion models for and
 approaches to problems
 related to ``influence maximization'' and ``minimization'' has
 increased considerably since the early groundbreaking work of
 \citet{kempe03} and particularly over the last few years.~\footnote{See, for example, \citet{budak2011,chen2011,he2012} and the
 references therein, for variations of the basic model, problem and techniques. A survey by~\citet{guille13} also contains
 latest developments.}
 We reemphasize that all that work uses a diffusion-based approach and thus fundamentally differs
 from the work presented here, both in
 terms of the problem definition and the solution approach. As a
 result, comparing the result of such disparate approaches is not
 scientifically meaningful, in our view. Furthermore, it would be nontrivial and out-of-scope
 for this paper to \emph{adapt} the techniques used or proposed to
 solve their specific variations of the ``influence maximization'' or
 ``minimization'' problem in a diffusion-based setting, to employ as
 \emph{heuristics} to solve our problem; just as we would not expect
 adaptations of our techniques as heuristics to solve \emph{their} problems.
 
 Therefore, \emph{in this paper, we only 
 adapt the popular, greedy algorithm used to find the ``most
 influential individuals'' in the cascade model, as described in~\cite{kleinberg07},
 to use as a heuristic to solve what would be the equivalent/analogous instance
 of that problem in the more general
 problem formulation in our setting} presented in
 Section~\ref{sec:probform}. That greedy algorithm, originally meant for the
 particular cascade model of diffusion, is arguably the most
 simple and fundamental of the work in that area. 
The resulting heuristic corresponds
to a very ``controlled'' version of (possibly networked) \emph{best-response dynamics} on
the influence game of interest starting from specific initial conditions. 
 
\paragraph{Diffusion-based Heuristic}

We first note that we use the same
influence factors and thresholds of the previously learned LIGs to
perform both of these analyses. In particular, for the
diffusion-based heuristic, at each iteration, we select a node $u$ that achieves the
``maximum spread'' of action $1$, by which we mean the node $u$ that leads to the largest
(marginal difference in the) number of other players in the network with action $1$ after
performing best-response dynamics in the LIG; force $u$ to adopt action $1$; and let
all but the previously selected nodes modify their actions as best
responses to $u$'s adoption of action $1$. We repeat this process
until every node adopts action $1$.~\footnote{At the end, we also
  perform a post-processing step, where we try to remove one of the
  selected nodes to test if the remaining nodes are still most-influential.} 

\emph{Note that because of negative influence factors,
  cycling may occur and this procedure may never come to a stop.}
However, in our case, even in the presence of negative influence
factors, we did not encounter such cycling. 

Furthermore, it is well
known that the general greedy-selection recipe just described produces a provable approximation
algorithm in the cascade model with submodular spread
function~\citep{kleinberg07}. But given the presence of negative influence factors
in our LIG
and, more importantly, the fact that we do not even know what
the corresponding ``submodular spread function'' is or means in our setting,
\emph{this claim of approximation
  guarantee essentially vanishes or is irrelevant in our setting.}


\paragraph{Results using Diffusion-based Heuristic}

We can visualize all possible choices of the most influential nodes that an algorithm can make as a directed acyclic graph, as shown in Figures~\ref{fig:LIG_most_infl_110} and~\ref{fig:diffusion_most_infl_110}.

\begin{figure}[!h]
	\centering
		\includegraphics[width=5in]{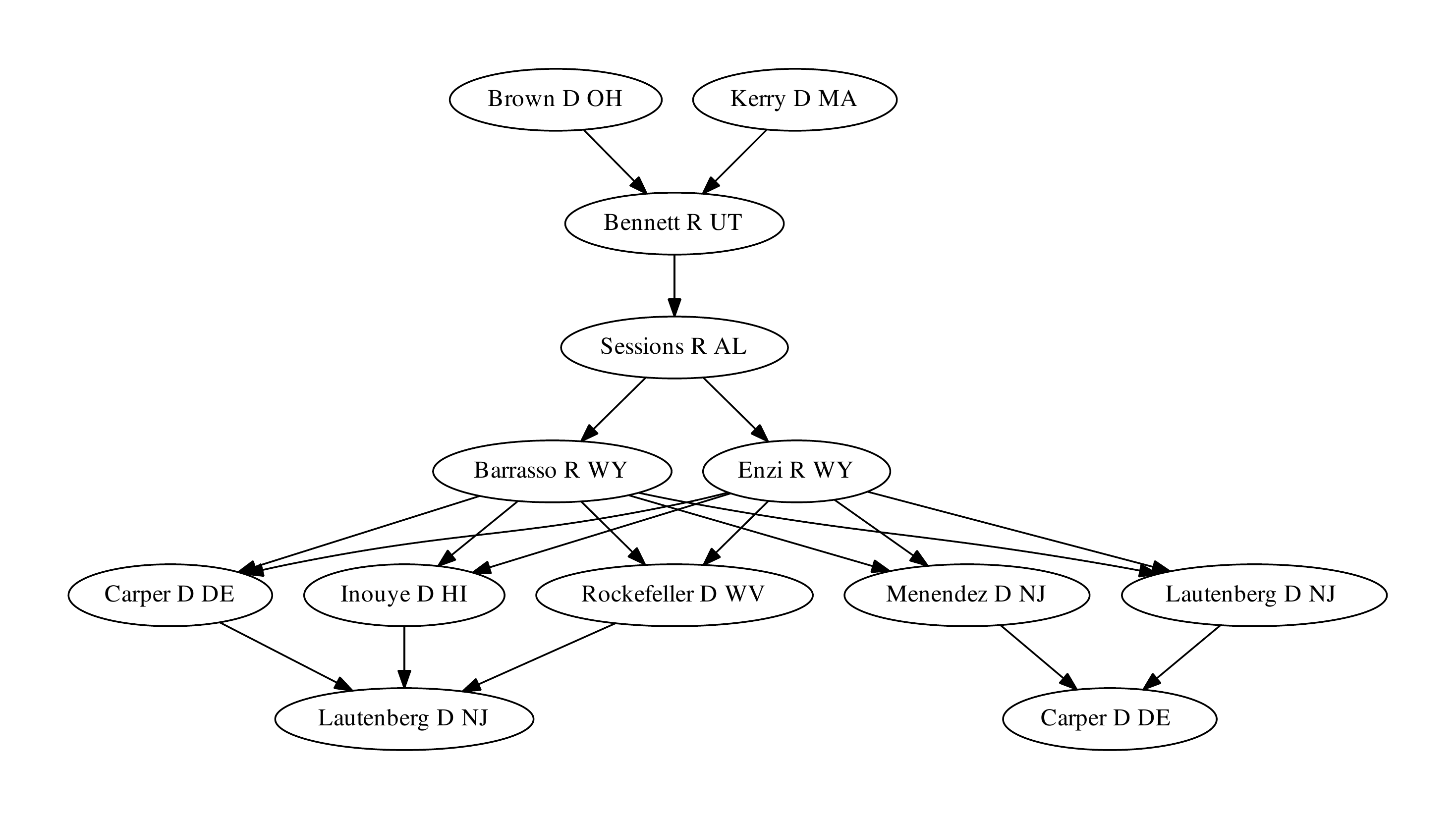}
	\caption{Most-influential nodes in our setting. This directed acyclic graph (dag) illustrates all possible options for node selection that our approximation algorithm considers. A source node represents a node selected
in the first iteration and a sink node represents a node selected in the last step.
Any directed path from a source to a sink represents a sequence of
nodes selected in successive iterations by our algorithm.  All nodes in the same level and having the same parent, are tied in an iteration of the algorithm. Also note that the same node can appear in different paths of the dag at different levels.}
	\label{fig:LIG_most_infl_110}
\end{figure}

\begin{figure}[!h]
	\centering
		\includegraphics[width=4in]{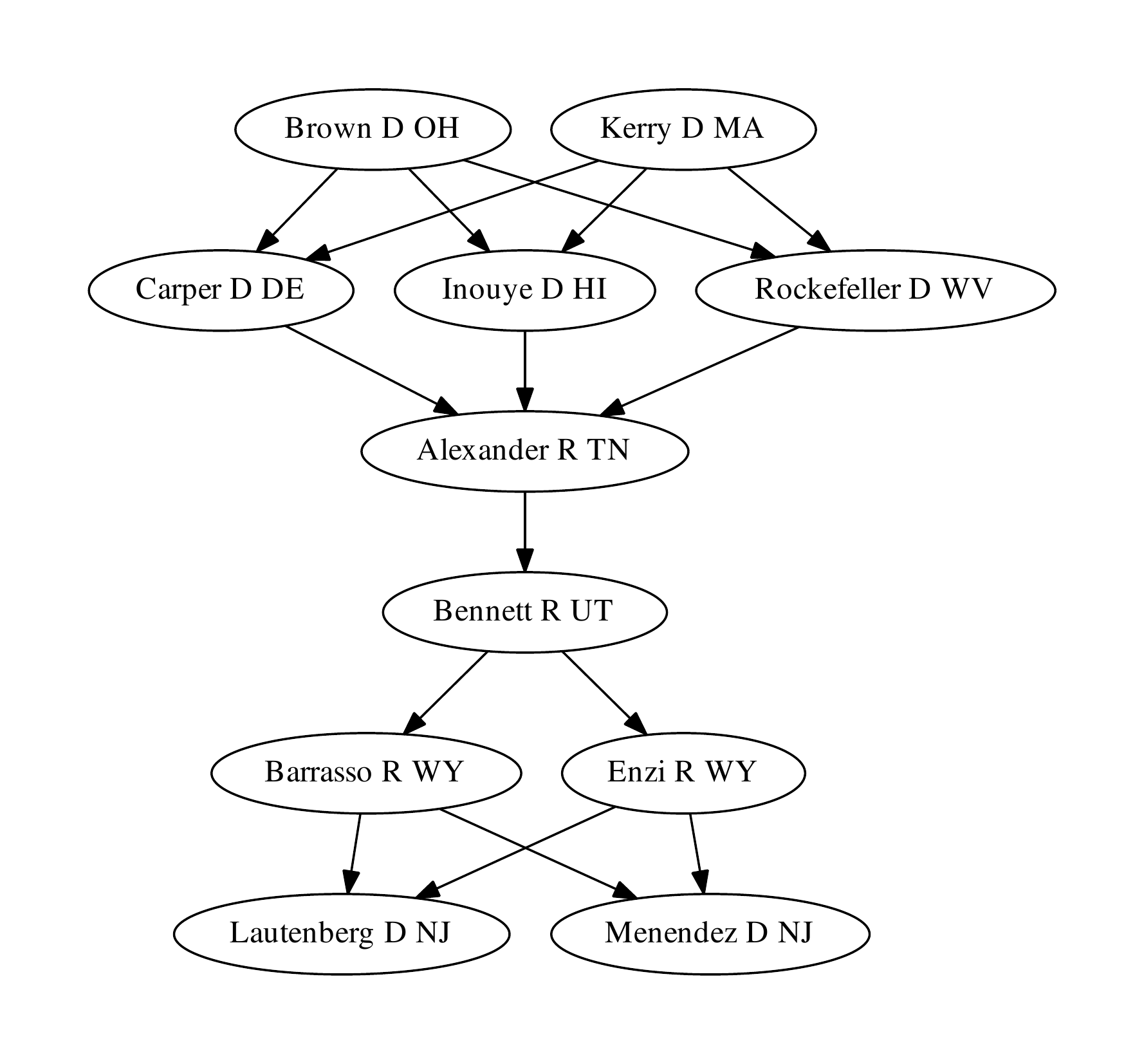}
	\caption{Most-influential nodes as output by the
          diffusion-based heuristic. Each
          directed path from a source to a sink represents a sequence
          of nodes that are most-influential, as determined by
          applying the variant of the standard greedy algorithm used for traditional diffusion models to identify the
          most influential nodes (described in the main body).}
	\label{fig:diffusion_most_infl_110}
\end{figure}

Although Figure~\ref{fig:LIG_most_infl_110} looks more complicated
than Figure~\ref{fig:diffusion_most_infl_110} (due to the appearance
of the same node in different source-sink paths of the dag at
different levels), comparing them we find that, not only a set of six
nodes are most-influential in both cases, but also most of the nodes
are common between these two distinct approximation algorithms. More remarkably,
some of these common nodes are selected at the same iteration in both
algorithms. The obvious question arises, is a set of most-influential nodes in the LIG setting, as output by our approximation
algorithm, also a possible output of the
diffusion-based heuristic? We have exhaustively tested all possible sets
of the most influential nodes (Figure~\ref{fig:LIG_most_infl_110}) and
settled the answer in the negative for each set. Interestingly, if we
add the ``Alexander R TN'' node to \textit{any} of the most
influential sets in the LIG setting, the resulting set can be the
output of the diffusion-based heuristic.
The apparent similarity in
results between the output of our approximation algorithm and the
diffusion-based heuristic gives rise to an intriguing open question as
to the characterization of the exact connection between these two
seemingly different algorithms for identifying this particular
instance (i.e., all players play $1$) of the most-influential problem in our setting.. This open question is beyond the scope of this paper.

\subsection{Filibuster}
Beyond predicting stable behavior and identifying the most influential
nodes in a network, we can also use our model to study other
interesting aspects of a networked population. One example is the
filibuster phenomenon in the U.S. Congress, where a senator uses his
or her right to hold floor for an indefinite time in an effort to
delay the passing of a bill. The procedure of ``cloture,'' which
refers to gathering a majority of at least 60 votes among the current
100 senators, can break a filibuster. However, not every possible
cloture scenario of 60 or more ``yes'' votes may be a \emph{stable}
outcome due to influence among the senators. We call the set of
such outcomes that are indeed stable in the sense of PSNE a \emph{stable cloture set}. 

An interesting general question is whether there exists a small
coalition of senators that can break filibusters. We can think of
preventing a filibuster from the democratic or the republican
perspective (i.e., favoring the respective party). Similarly, we can
also ask what in some sense is the opposite question: is there a small
coalition of senators that can \emph{prevent clotures} by voting
``no?'' First, let us us formally define the problem for breaking
filibusters; the formulation of the problem for preventing clotures is
similar.

\subsubsection{Problem Formulation}
Given the set $\mathcal{S}$ of all stable outcomes (i.e., PSNE) and a subset of $\mathcal{C}$ of these stable outcomes, find a minimal set $T$ of players such that 
\begin{align*}
T & \in \argmax_{V \subset \{1, ..., n\}} 	\{ \left|
  P_{\mathcal{S}}(V)\right| \, \mid \,   P_{\mathcal{S}}(V) \subseteq \mathcal{C}\},
\end{align*}
where $P_{\mathcal{S}}(V)$ is the set of PSNE-extensions of the nodes in $V$ playing action 1, i.e.,
\begin{align*}
	P_{\mathcal{S}}(V) & = \{  \mathbf{x} \, \mid \, \mathbf{x} \in \mathcal{S},\  x_i = 1 \ \forall i \in V\}.
\end{align*}

In words, $\mathcal{C}$ is the stable cloture set, consisting of
stable outcomes that can prevent a filibuster (i.e., every PSNE in
$\mathcal{C}$ contains at least 60 ``yes'' votes and thus, can induce
a cloture). When we consider the notion of preventing a filibuster in
favor of a specific party, we define $\mathcal{C}$ consisting of exactly those PSNE that contain 60 or more ``yes'' votes (thereby representing cloture scenarios) and in addition, are supported (through ``yes'' votes) by the majority of the senators affiliated with that party. Other definitions are possible, as long as the stable cloture set $\mathcal{C}$ is well-defined. 

Now, we would like to select a \emph{minimal} set of senators such
that $\mathcal{C}$ contains the set $P_{\mathcal{S}}(V)$ of the
PSNE-extensions of these senators' voting ``yes'' (i.e., their voting
``yes'' can only lead to a stable cloture scenario, thereby preventing
a filibuster). In addition, we would also like to achieve a maximum
\emph{stable-cloture cover}; that is, we wish to achieve the maximum
possible set $P_{\mathcal{S}}(V)$ so that we are able to capture as
many of the stable cloture scenarios as possible. In this formulation,
we set up as the objective to select a minimal, not minimum, set of
senators; this keeps the formulation simple by avoiding
bicriteria optimization (minimum set of senators vs. maximum
stable-cloture cover). Further note that adding an extra senator to
the set of selected senators can only reduce the stable-cloture cover
because of additional constraints.

The problem formulation above guarantees a nonempty solution $T$ if there exists some PSNE in $\mathcal{C}$ that is not ``dominated'' by any PSNE in $\mathcal{S} \setminus \mathcal{C}$. Here, a PSNE $\mathbf{x}$ dominates another PSNE $\mathbf{y}$ if for every $i$, $y_i = 1 \implies x_i = 1$.

\subsubsection{A Heuristic}
We can modify the approximation algorithm for identifying the most
influential nodes, presented at the end of Section~\ref{SecInfluentialNodes}, to design a heuristic for the problem formulated
above in the following way. At each iteration, we select a node such that adding it to the set of already selected nodes minimizes the number of PSNE-extensions of the selected nodes playing $1$ that are in $\mathcal{S} \setminus \mathcal{C}$. If there is a tie among several nodes in this step, then we can store these nodes in order to explore all solutions that this heuristic can produce. We stop when the above number of PSNE-extensions within $\mathcal{S} \setminus \mathcal{C}$ goes to 0. We then perform a minimality test by excluding nodes from the selected set of nodes and testing whether the resulting set can be a solution. Note that although we can select the ``best'' solution (in terms of the coverage of $\mathcal{C}$) among the ones found due to ties, this heuristic does not guarantee an approximation of the maximum coverage of $\mathcal{C}$.

\subsubsection{Experimental Results on the 110th Congress}
For the 110th Congress, $\mathcal{C}$ consists of 15,288 and 10,029 stable cloture scenarios (i.e., PSNE) with respect to the democratic and republican parties, respectively. Overall, the total number of stable cloture scenarios is 15,595, and most of these are common in both democratic and republican cases. 

Regarding breaking filibusters with respect to the democratic party,
the best solutions found by the above heuristic are Senators \{Brown
(D, OH), Roberts (R, KS), and Graham (R, SC)\} and \{Kerry (D, MA),
Roberts (R, KS), and Graham (R, SC)\}, both of which cover 1,500 of
the 15,288 stable cloture scenarios. The optimal solutions found by a
brute-force procedure are Senators \{Brown (D, OH),  Craig (R, ID),
and Dole (R, NC)\} and \{Kerry (D, MA), Craig (R, ID), and Dole (R,
NC)\}, both covering  1,728 stable cloture scenarios. With respect to
the republican party, the heuristic gives the following two solutions as the best, each covering 40 of 10,029 stable cloture scenarios: Seantors \{Brown (D, OH), Bennett (R, UT), and  Gregg (R, NH)\} and Senators \{Kerry (D, MA), Bennett (R, UT), and  Gregg (R, NH)\}. The optimal solution for this case is Senators \{Bennett (R, UT),  Conrad (D, ND),  and  Sessions (R, AL)\}, which covers 138 stable cloture scenarios.

We now consider the case of preventing cloture scenarios with respect
to the democratic party; that is, the majority of the senators belonging to the democratic party want to pass a bill, but cannot gather 60 votes due to some senators voting ``no.'' To find a small set of senators whose voting choice of ``no'' prevents cloture scenarios and potentially leads to filibusters, we adapt the above heuristic. With respect to the democratic party, there are 295,320 non-cloture stable outcomes (where a majority of the democratic senators voted ``yes,'' yet there are fewer than 60 ``yes'' votes in total). The output of our heuristic matches the optimal set in this case, which covers 9,681 non-cloture scenarios. If Senators McCain (R, AZ), McConnell (R, KY), Coburn (R, OK), and  Hutchison (R, TX) vote ``no'' then the set of PSNE is exactly the set of these 9,681 non-cloture scenarios. In that case, according to our model, a cloture can never take place in the event of a filibuster.


\subsubsection{Application of Another Diffusion-Based Heuristic}
We can once again try to adapt the popular greedy-selection algorithm
for identifying most-influential nodes in the cascade model in the
diffusion setting to work as another \emph{heuristic} for the filibuster problem
in our setting. In doing so, we encounter two notable obstacles that
highlight another difference between our approach and that based on diffusion. 

First, the
notion of stable-cloture cover is not well-defined in the diffusion
setting. The forward recursion mechanism central to diffusion models begins with a set of initial adopters (those senators selected to vote ``yes'' in our case) and propagates the effects of behavioral changes throughout the network until it reaches a steady state (i.e., no change occurs). However, this mechanism focuses on how the dynamics of behavioral changes evolves, not on the count of steady states that are consistent with a given set of players being among the adopters (not necessary early adopters), which is required for stable-cloture covers.
In contrast, the stable-cloture cover is well-defined in our approach. 


Second, and most important, even if we allow reversals of actions due to negative influence factors, forward recursion may produce an \emph{unstable} outcome (i.e., not a PSNE). Although Granovetter's original model precludes this by requiring the initial adopters to have a threshold of 0 \citep{granovetter78}, subsequent development allows forward recursion to start with a set of initial adopters whose thresholds are not necessarily 0 \citep{kleinberg07}. Next, we illustrate this point using our experimental results.

Per the discussion in the last two paragraphs, in our experimental setting regarding the
diffusion-based heuristic, which, once again, result from \emph{our adaptation} of the popular
greedy-selection approximation algorithm used for the cascade model in diffusion
settings, we omit the notion of maximum stable-cloture cover and
thereby forgo the measure of goodness of a solution. We only
concentrate on finding a set of initial adopters that can drive the
forward recursion process to \emph{some} stable cloture scenario
(i.e., a PSNE in $\mathcal{C}$). 

In the following paragraph, we outline our \emph{diffusion-based
  heuristic specifically adapted for the
  filibuster problem.}

For $k = 1, 2, \ldots$, do the following. For all possible sets of $k$
senators, start forward recursion with these $k$ senators forced to
play $1$ all the time and other senators initially playing $-1$ (but are
permitted to switch between $1$ and $-1$ later on). When a steady
state is reached, verify if there are at least $60$ senators who are
playing $1$ in this state. If this is the case, then further verify if
the $k$ senators who are forced to play $1$ are indeed playing their
best response with respect to others' actions, which is the condition
for the cloture scenario being stable. Stop iterating over $k$ once
you find stable cloture scenarios.

In our problem instances, which contain both positive and negative influence factors, it is very much possible that forward recursion oscillates indefinitely. However, that did not happen in our experiments. We tried all possible sets of $k \le 3$ initial adopters, but failed to reach any cloture scenario (stable or unstable). We then tried all possible quadruplets of initial adopters. With respect to democratic party, 1,189 different quadruplets led the forward recursion process to a cloture scenario, but  nearly half of these quadruplets (536 to be exact) led to unstable outcomes. Essentially, those unstable outcomes were due to some of the initial adopters not playing their best response in voting ``yes''---all other nodes were indeed playing their best response (otherwise, the process would not terminate).


Therefore, beyond just emphasizing the stability of an outcome, the
approximation algorithm based on our
approach also captures certain phenomena that the heuristic based on
the traditional approach
seems unable to do. As stated earlier, of course, we would need more research to better
understand this discrepancy, and the degree to which it could be
reduced, as well as the effectiveness and potential for improvement of
the diffusion-based heuristic. Such research is beyond the scope of
this paper and remains open for future work.


\section{Conclusion}
In this paper, we studied influence and stable behavior from a new game-theoretic perspective. To that end, we introduced a rich class of games, named influence games, to capture the core strategic component of complex interactions in a network. We characterized the computational complexity of computing and counting PSNE in LIGs. We proposed practical, effective heuristics to compute PSNE in such games and demonstrated their effectiveness empirically. Besides predicting stable behavior, we gave a framework for computing the most influential nodes and its variants (e.g., identifying a small coalition of senators that can prevent filibuster). We also gave a provable approximation algorithm for the most-influential-nodes problem. 

Although our models are inspired by decades of research by sociologists, at the heart of our whole approach is abstracting the complex dynamics of interactions using the solution concept of PSNE. This allowed us to deal with richer problem instances (e.g., the ones with negative influence factors) as well as to tread into new problem settings beyond identifying the most influential nodes. We conclude this paper by outlining several interesting lines of future work.


First, we leave several computational problems open. We showed that counting the number of PSNE even in a star-type LIG is \#P-complete, but does there exist an FPRAS for the counting problem? The computational complexity of indiscriminant LIGs, which we conjecture to be PLS-complete, is unresolved. Also, computing mixed-strategy Nash equilibria of LIGs, even for special types such as trees, remains an open question. 

Second, we can apply our models to the general setting of ``strategic interventions,'' where we study the effects of changes in node thresholds, connectivity, or influence factors, usually without the possibility of having corresponding behavioral data. The following is an illustrative example of it in the context of the 111th U.S. Congress. After the death of Ted Kennedy, who was a democratic senator from the state of Massachusetts, a republican named Scott Brown was elected in his place. Not only that it was Senator Brown's first appointment in Senate, he was also the first republican from Massachusetts to be elected to Senate for a long time. With no behavioral data available at that point of time, we can perform interventions using our model under various assumptions of thresholds, connectivity, and influence factors regarding Senator Brown, with the general goal of predicting stable outcomes and investigating the effects of the above interventions in various settings, such as filibuster scenarios or the setting of the most influential senators.

Another example of intervention, in the context of the Framingham heart study alluded in Section~\ref{sec:intro}, is the following. Suppose that we would like to implement a policy of targeted interventions in order to reduce smoking by some margin. Using our model, we can modify the thresholds of the selected targets and predict how it could affect the overall level of smoking.

Besides interventions, we can also use our model to analyze past happenings, such as the role of the bipartisan ``gang-of-six'' senators in leading the members of the two major parties to an agreement during the U.S. debt ceiling crisis.~\footnote{\url{http://en.wikipedia.org/wiki/United_States_debt-ceiling_crisis}}
We can use our model to find how influential the gang-of-six senators were as a group. One approach to this problem would be to first find the set of stable outcomes consistent with the gang-of-six senators voting ``yes'' and then to analyze what fraction of these stable outcomes has 60 or more ``yes'' votes (signifying the passing of the corresponding bill without any possibility of a filibuster).

  Our model of influence game can be considered as a step in the direction of modeling competitive contagion in strategic settings (see, for example, \cite{budak2011,competitivecontagion}). Here, one of the main challenges would be to formulate the competitive aspects of multiple ``campaigns'' without having to go through the usual network dynamics. In this general setting, we can ask questions such as who are the most influential individuals with respect to achieving a certain objective that favors one ``campaign'' as opposed to the other? Questions like these and many others shape the long-term goal of this research.

\section*{Acknowledgements}

We are sincerely grateful to the reviewers for their comments,
suggestions, and constructive criticisms. These not only improved this
paper significantly but also contributed to the PhD Dissertation of
M.T. Irfan \citep{irfanphd}, which also contains the work in this paper.

This article is a significantly enhanced version of an earlier
conference paper~\citep{Irfan_Ortiz_AAAI11}. Besides a more detailed
exposition that places this work in the context of existing ones in
sociology, economics, and computer science, we extended the technical
content in several ways. First, we made the connection between linear
influence games and polymatrix games. As a result, the computational
complexity results carry over to 2-action polymatrix games. Second, we
gave the complete proofs of the theoretical results here. In addition
to enhancing the theoretical part, we also extended the experimental
part substantially. First and foremost, we illustrated our approach to
influence in a new setting---the U.S. Supreme Court rulings. Given the
U.S. Supreme Court dataset, we first illustrated how we can learn an
LIG to model the potential strategic interactions among the Supreme
Court justices \citep{honorio_and_ortiz13}. We then applied the
schemes for PSNE computation and the identification of the most
influential individuals in this new setting. Second, we extended our
empirical study of the U.S. congressional voting by contrasting
the output of our approximation algorithm to that of a diffusion-based
heuristic we created by adapting the simple, greedy-selection
approximation algorithm used for the analogous problem in the cascade model. We also applied our approach to a new problem of preventing filibuster by a coalition of senators, which highlights the broad range of scenarios where our techniques can be applied.

This work was supported in part by NSF CAREER Award IIS-1054541.


\begin{appendix}

\section{On the Connection to Rational Calculus Models of Collective Action}
\label{app:conn}
The formal study of individual behavior in a collective setting
originally began under the umbrella of ``collective behavior'' in
sociology and social psychology. The classical treatment of collective
behavior views individuals in a ``crowd'' as irrational beings with a
lowered intellectual and reasoning ability. The proposition is that an
increased level of suggestibility among the individuals facilitates
the rapid spread of the homogeneous ``mind of the
crowd''~\citep{lebon_crowd_1897,park_burgess_1921,blumer_1939}. Herbert
Blumer's work, in particular, popularized the classical theory of
collective behavior well beyond academia and into such domains as
police and the armed forces \citep[p. 9]{mcphail_myth_91}. However,
this theory was subjected to much criticism primarily because it did
not study empirical accounts systematically.~\footnote{For instance, Blumer himself referred to this as a ``miserable job'' by sociologists \citep{blumer_1957}.}

In response to that, Clark McPhail undertook a massive effort,
spanning three decades, to record the behavior of individuals in
collective settings that he calls ``gatherings'' in order to
distinguish it from (homogeneous) ``crowds'' in collective behavior
(see \citet[Ch. 5, 6]{mcphail_myth_91} for a summary of his two-decade
study). His empirical accounts, stored in a range of media formats as
technology improved, reveal one common thing---that a gathering
consists of individuals with diverse objectives, who nevertheless
behave rationally and purposefully. To distinguish this purposive
nature of individuals from irrationality in the classical treatment,
he calls his study ``collective action'' and broadly defines it as
``any activity that two or more individuals take with or in relation
to one another'' \citep[p. 881]{mcphail_ency_2007}. In short,
collective action can be seen as the modern approach, as opposed to
the ``old'' (but not unimportant) approach of collective behavior
\citep[p. 14--15]{miller_cb_ca_2000}.~\footnote{A brief review of
  collective behavior and collective action literature is included in Appendix~\ref{app:rev} for interested readers.}

Many of the rational calculus or economic choice models that were originally proposed  for collective behavior, are now discussed under collective action due to the purposive nature of the individuals. Here, we will conduct a very narrow and focused review of the relevant literature in order to place our model in its proper context. Our review will be concentrated around Mark Granovetter's threshold models \citep{granovetter78}, which is one of the most influential models of collective action to date. Before that, we will briefly review two prominent precursors to Granovetter's models---Schelling's models of segregation and Berk's ``gaming'' approach.

\subsection*{Schelling's Models of Segregation}
A notable precursor to Granovetter's threshold models is Nobel-laureate economist Thomas Schelling's models of segregation \citep{Schelling_dynamic_71,Schelling78}. Schelling's models account for segregations that take place as a result of \textit{discriminatory individual behavior} as opposed to organized processes (e.g., separation of on-campus residence between graduate and undergraduate students due to a university's housing policy) or economic reasons (e.g., segregation between the poor and the rich in many contexts). An example of a segregation due to individual choice, or ``individually motivated segregation'' as Schelling puts it \citep[p. 145]{Schelling_dynamic_71}, is the residential segregation by color in the U.S. Although Schelling's models expressly focus on this case, these can be applied to many other scenarios as well.

In Schelling's \textit{spatial  proximity model}, if an individual's
\textit{level of tolerance} for population of the opposing type is
exceeded in his neighborhood, he moves to another spatial location
where he can be ``happy.'' Schelling studied the dynamics of
segregation in this model using a rule of movement for the ``unhappy''
individuals. The \textit{bounded-neighborhood model} is concerned with
one global neighborhood. An individual enters it if it satisfies its
level of tolerance constraint and leaves it otherwise. Schelling
studies the stability of equilibria and the \textit{tipping}
phenomenon in this model when the distribution of tolerances and the
population ratio of the two types are varied. An important finding is
that in the cases studied, the modal level of tolerance does not
correspond to a tipping point.~\footnote{More on Schelling's models
  can be found in Appendix~\ref{app:rev}.}

\subsection*{Berk's ``Gaming'' Approach}
Another notable precursor to Granovetter's models is Berk's rational calculus approach \citep{berk_1974}. Berk strongly criticizes the assumption of individual irrationality which became prevalent in collective behavior literature. He formulates his approach by first giving a detailed empirical account of an anti-war protest at Northwestern University that originated in a town-hall meeting addressing dormitory rent hike.~\footnote{Berk's description gives accounts of both mundane and exciting happenings during the course of the protest and is recognized as ``among the best in the literature'' \citep[p. 126]{mcphail_myth_91}.} He explains individual decision making through Raiffa's decision theory principles.

To motivate his approach, he first notes that participating individuals in that protest were diverse in their disposition and that they exercised their reasoning power. He then broadly classifies the participants into two types---militants (with the desired action of trashing properties) and moderates (with the desired action of an anti-war activity, but not trashing). Each participant, militant or moderate, estimates the support in favor of his disposition, and with enough support, he will ``act'' (e.g., trash properties if he is militant). 

Clearly, an individual's estimate of support directly affects his ``payoff.'' If an individual estimates that there is not enough support to act in favor of his disposition, he can try to persuade others to support his disposition so that he can receive a higher payoff by being able to act. This can be translated as an attempt to change others' payoff matrices, which is facilitated by the \textit{milling} phase when they communicate and negotiate with each other. The milling phase ends when a consensus or a compromise is reached and becomes common knowledge. In this way, a concerted action takes place according to Berk's model.

\subsection*{Granovetter's Threshold Models}
\cite{granovetter78} presented his threshold models  in the setting of a crowd, where each individual is deciding whether to riot or not. In the simplest setting, each individual has a threshold and his decision is influenced by the decisions of others---if the number (or the proportion) of individuals already rioting is below his threshold, then he remains inactive, otherwise he engages in rioting. The emphasis is on investigating equilibrium outcomes due to the process of forward recursion \citep[p. 1426]{granovetter78}, given a distribution of the thresholds of the population. It may be mentioned here that  forward recursion starts only if there is an individual with a threshold of $0$.

Granovetter's models are inspired by Schelling's models of segregation. In fact, one can draw a parallel between Schelling's level of tolerance and Granovetter's threshold in the following way. In Schelling's models, an individual leaves a neighborhood if his level of tolerance is exceeded, whereas in Granovetter's models, an individual becomes active in rioting if his threshold is exceeded. Furthermore, in both models, dynamics is of utmost importance and serves the purpose of explaining \textit{how} an equilibrium collective outcome emerges from individual behavior. However, apart from these similarities, these two models are semantically different and also focus on completely different outlooks. First, Granovetter ascribes a deeper meaning to the concept of threshold. Threshold of an individual is not just ``a number that he carries with him'' from one situation to another \citep[p. 1436]{granovetter78}. It rather depends on the situation in question and can even vary within the same situation due to changes occurring in it. Second, in Granovetter's models, a very small perturbation in the distribution of population threshold may lead to sharply different equilibrium outcomes. Granovetter highlights this property of his models as an explanation of seemingly paradoxical outcomes that goes against the predispositions of the individuals.

Two features of Granovetter's models make it stand out among the
rational calculus models. First, the models are capable of capturing
scenarios beyond the classical realm of collective
behavior. Granovetter begins by setting up his model to complement the
emergent norm theory (see Appendix~\ref{app:rev}) by providing an explicit model of how ``individual preferences interact and aggregate''  to form a new norm \citep[p. 1421]{granovetter78}. Not only that such an explicit model eliminates the need for implicit assumptions (such as a new norm emerges when the majority of the population align themselves with that norm), it can also capture paradoxical outcomes alluded above that cannot be captured by the implicit assumption on the majority. Beyond the emergent norm theory, Granovetter's models can capture a wide range of phenomena that do not fall within the classical realm of collective behavior, such as diffusion of innovation, voting, public opinion, and residential segregation, to name a few. The second prominent feature of Granovetter's models is its ease of adaptation when dealing with a networked population. The same mechanism of forward recursion is applicable when the underlying influence structure is specified by a ``sociomatrix,'' which accounts for how much an individual influences another \citep[p.1429]{granovetter78}. This is particularly useful for studying collective action in the setting of a social network. 

\subsection*{Criticism of Rational Calculus Models}

An implicit assumption regarding Granovetter's sociomatrix is that the elements of the matrix are non-negative. Otherwise, the process of forward recursion may never terminate, even on the simplest of examples. However, many real world scenarios do exhibit co-existence of both positive and negative influences. For the most part, democrat senators in the U.S. Congress influence their republicans colleagues negatively, while they influence colleagues of their own party positively. In residential segregation involving more than two types of individuals, an individual is negatively influenced, in different magnitudes, by individuals belonging to other types. Clearly, such a situation cannot be modeled using a non-negative sociomatrix. Furthermore, if we take a second look at Berk's account, militant individuals positively reinforce each other in their decision to engage in trashing properties, whereas their decision is negatively affected by the moderates (that is, the presence of too many moderates makes it risky for militants to engage in violent action).

Critiques of rational calculus models point out the lack of behavioral
adjustment in a ``negative feedback'' fashion
\citep[p. 883]{mcphail_ency_2007}. Here, negative feedback is defined
in the context of the perceptual control theory that lays the
foundation of McPhail's \textit{sociocybernetics theory} of collective
action (see Appendix~\ref{app:rev}). In a negative feedback system, an individual can adjust his behavior depending on the discrepancy between the input signal and the desired signal (the sign of this discrepancy has no correlation to negative feedback). In contrast, in a positive feedback system, such control of behavior is not possible. A typical example of a positive feedback system is a chemical chain reaction. An analogue to this is  the ``domino effect'' cited often in rational calculus models \citep[p. 1424]{granovetter78}. It is true that rational calculus models neither accounts for ``errors'' as desired by the proponents of the sociocybernetics theory, nor is it well-defined in the context of rational calculus. But it is not the case that ``reversal'' of behavior, which can be thought of as a crude form of behavioral adjustment, is precluded in rational calculus models. Such a form of behavioral adjustment can certainly be incorporated by allowing negative elements in the sociomatrix, but the challenge lies in the forward recursion process which may oscillate indefinitely because of those negative elements.

\section{Brief Review of Collective Behavior and Collective Action in
  Sociology}
\label{app:rev}
In sociology, the umbrella of \textit{collective behavior} is very
broad and encompasses an incredibly rich set of models explaining
various aspects of a wide range of social phenomena such as
revolutions, movements, riots, strikes, disaster, panic, and diffusion
of innovations (e.g., fashion, adopting contraceptives, electronic
gadgets, or even religion), just to name a few.~\footnote{In fact, the
  richness of just one subfield of collective behavior, termed
  \textit{micro-level theories of collective behavior}, led Montgomery
  to comment in his book \citep[p. 67]{montgomery_religion_99}, ``The
  variety of theories focusing on the micro level is confusing, but is
  an indication of the complexity and variations in the process by
  which movements emerge or perhaps fail to emerge...''} Sociologists
Marx and McAdam, in their concise introductory book on collective
behavior \citep{marx_mcadam_cb_94}, contend that unlike many other
fields of sociology, the field of collective behavior is not easy to
define, partly because of the varied opinion of scholars: from very
narrow perspectives, to such wide, all-encompassing perspectives
(e.g., Robert Park and Herbert Blumer's) that virtually eliminates the
need to have collective behavior as an individual field in sociology.~\footnote{Quoting from their book, ``The field of collective behavior is like the elephant in Kipling's fable of the blind persons and the elephant. Each person correctly identifies a separate part, but all fail to see the whole animal.''}
Yet, Marx and McAdam point out the traditional disposition to
categorize collective behavior as a ``residual field'' in sociology; that is, the study of collective behavior consists of those elements of behavior (e.g., fads, fashion, crazes), organization (e.g., social movement), group (e.g., crowd), individual (e.g., psychological states such as panic), etc., that do not readily fit into well-established and commonly observed social structures. Similarly, collective behavior is defined in Goode's textbook \citep[p. 17]{goode_cb_92} as the ``relatively spontaneous, unstructured, extrainstitutional behavior of a fairly large number of individuals.''

\subsection*{Classical Treatment of Collective Behavior: Mass Hysteria}
The classical treatment of collective behavior views individuals in a
crowd as non-rational,  transformed into hysteria by the collective
environment. The central tenet of the early work of Gustave Le Bon's
is that individuals in a crowd share a ``mind of the crowd,'' and that
the ``psychological law of the mental unity of crowds" guides their
psychological state and behavior in the collective setting
\citep[p. 5]{lebon_crowd_1897}. In Le Bon's account, an individual in
a crowd may retain some of the ordinary characteristics he shows in
isolation; but the emphasis is on the extraordinary characteristics
that emerge only in a crowd because of ``a sentiment of invincible
power,'' contagion, and most importantly, the susceptivity of
individuals to take suggestions as if they were hypnotic
subjects. Examples of such extraordinary characteristics of a crowd
are ``impulsiveness,'' ``incapacity to reason,'' and ``the absence of
judgment," to name a few \citep[p. 16]{lebon_crowd_1897}. In sum,
individuals in a crowd are depleted of their intellectual capacity and
become uniform in their psychological state. This leads the crowd to
an identical direction of collective behavior that may be heroic or
criminal, depending on the type of ``hypnotic suggestion'' alluded
above (although Le Bon gives examples of heroic crowds
\citep[p. 14]{lebon_crowd_1897}, for the most part, he tends to give
``crowds'' a negative connotation).  

Le Bon's work influenced around half-a-century of subsequent
developments. Park and Burgess upheld his proposition on the
transformation of individuals in a crowd. They put forward the concept
of \textit{circular reaction}  \citep{park_burgess_1921}, later
refined by Herbert Blumer \citep{blumer_1939}. Circular reaction
refers to a reciprocal process of social interaction that explains
\textit{how} crowd members become uniform in their behavior, something
that Le Bon could not really explain. In this process, an individual's
behavior stimulates another individual to behave alike; and when the
latter individual does so, it reinforces the stimulation that the
former individual acted upon. In circular reaction, individuals do not
act rationally or intellectually; that is, they do not reason about
the action of others, but instead, they only align themselves with the behavior of others. This is different from \textit{interpretative interaction}, another mechanism that Blumer defined to explain routine group behavior (e.g., a group of individuals shopping in a mall), as opposed to collective behavior (e.g., social movement).
In interpretative interaction, individuals react (perhaps differently)
to their \textit{interpretation} of others' action, not the action
itself. Therefore, in interpretative interaction, one can treat
individuals as rational beings.~\footnote{For clarity of presentation,
  we differentiate between routine group behavior and collective
  behavior. To the contrary, Blumer, as well as Park in his earlier
  work, viewed collective behavior as encompassing a wide range of
  social phenomena, including  routine group behavior. Within the
  continuum of collective behavior in Blumer's view, the presence or
  the absence of rationality of individuals earmarks two distinct
  mechanisms named interpretative interaction and circular reaction, respectively.}

In Blumer's account, a crowd goes through several well-defined stages
before a collective behavior finally emerges. The three underlying
mechanisms that facilitate transitions among these stages are circular
reaction, collective excitement, and social contagion; one can roughly
think of collective excitement as a more intense form of circular
reaction, and of social contagion as an even more intense form of
circular reaction \citep[p. 11]{mcphail_myth_91}.~\footnote{The
  mechanism of ``social contagion'' as defined by Blumer or the
  ``social contagion theory''  \citep[p. 11]{locher_2001} in general
  should not to be confused with the term ``social contagion'' that
  computer scientists use \citep{kleinberg_social_contagion}.
  Although both have their roots in epidemics, in the former case,
  individuals are transformed into being more suggestible, which facilitates ``rapid, unwitting, and non-rational dissemination'' of behavior \citep{blumer_1939}; whereas in the latter case, individuals act rationally.}

\subsection*{Emergent Norm Theory}
Although Blumer's account of collective behavior received wide-spread
acceptance, even beyond academia \citep[p. 9]{mcphail_myth_91}, others
deemed many of the underlying assumptions in it, as well as in the
general mass-hysteria theory, unrealistic. 
Arguably,  individuals with
\textit{different} objectives in mind participate in a collective
behavior, and one can observe \textit{changes} in their
individual behavior throughout the process of a
collective behavior too.~\footnote{We refer the reader to
  \citep[p. 26--27]{miller_cb_ca_2000} for a beautiful example in the
  context of the 1967 anti-war demonstration in Washington, DC, which
  was participated by nearly 250,000 people. While many of the
  participants might have been there to genuinely voice their opinion
  against the war, some might have been looking for ``excitement,
  drug, or sex." Yet again, individuals playing different roles, such
  as protest leaders, street vendors, and the police, behaved
  differently.} Therefore, the assumption of complete uniformity
behavior in the classical mass-hysteria treatment is very much a
stretch. Furthermore,  the assumption of hysteric crowd in the
classical approach has also been called into question. One notable
critique of the mass-hysteria theory comes from Ralph Turner and Lewis
Killian \citep{turner_killian_1957}. They view individuals in a crowd
as behaving under normative constraints and showing ``differential
expression.'' However, a new norm emerges when the established norms of the society cannot adequately guide a crowd facing an
extraordinary situation. They call this the \textit{emergent norm} and contend that it is the emergent norm that gives the ``illusion of unanimity.''

\subsection*{Collective Action}
The goal-oriented nature of collective behavior was further
highlighted by sociologists studying social movements during the 1970s
and 80s. In order to distinguish their approach from the traditional
approach to collective behavior, dominated by the assumption of
irrational and aimless nature of crowds, they used the term
\textit{collective action} to mean ``people acting together in pursuit
of common interests'' \citep{tilly_1978}. Strikingly, based on a
series of systematic observations, Clark McPhail's contends that the
goal-oriented nature of crowds is not limited to social movements and
revolutions alone, but is a feature of various other types of
crowds. In his book \textit{the Myth of the Madding Crowd}, he uses
two decades of empirical observations pertaining to a multitude of
crowd settings to formulate a theory of collective behavior now
recognized as a significant paradigm shift \citep[Ch. 5,
6]{mcphail_myth_91}. To distance himself from the term ``crowds,''
which has already gained several meanings depending on whose theory is
being considered, he gives his formulation in the setting of
``gatherings.'' But first, he places a justifiably strong emphasis on
the definition of collective behavior. His ``working definition of collective behavior'' is the study of ``two or more persons engaged in one or more behaviors (e.g., locomotion, ...) judged common or concerted on one or more dimensions (e.g., direction, velocity, ...)''~\footnote{Note that it is the ``behavior," not a specific action, that needs to be common or concerted. For example, when a group of people are chatting together, their behavior is concerted, even though they are not speaking identical words.} \citep[p. 159]{mcphail_myth_91}. 

The broad nature of McPhail's definition of collective behavior,
although based on extensive empirical evidence, did not receive
immediate acceptance. Even modern textbooks on collective behavior try
to conserve the classical appeal of collective behavior.~\footnote{For
  example, in reference to McPhail's definition, Goode writes in his
  textbook, ``In the view of most observers, myself included, many
  gatherings are \emph{not} sites of collective behavior (most casual
  and conventional crowds, for example), and much collective behavior
  does not take place in gatherings of any size (the behavior of most
  masses and publics, for example)''
  \citep[p. 17]{goode_cb_92}. Ironically, this is the very viewpoint
  that McPhail seeks to portray as a myth.} Perhaps to further
distance himself from the traditional viewpoint, McPhail later began
to use the term \textit{collective action} instead of collective
behavior (for example, in a recent encyclopedia article, McPhail
refers to the above mentioned definition as that of collective action
\citep{mcphail_ency_2007}). According to David Miller, the modern view
on the distinction between collective action and collective behavior
is beyond simply terminological. Collective action is given the status
of a ``new'' theory in sociology, while collective behavior is marked
as ``old,'' but not
unimportant. \citep[p. 14--15]{miller_cb_ca_2000}.~\footnote{Miller also points
out that sociology textbooks are likely to talk about collective
behavior only, whereas recent journal articles talk about collective
action only.}

McPhail's approach to collective action is known as the \textit{social
  behavioral interactionist (SBI)} approach. As much as it agrees with
the emergent-norm theory, in terms of the diversity of individual
objectives in a collective setting, it does not agree with the concept
of an emergent norm suppressing this diversity. The SBI approach
studies gatherings in three phases of its life cycle: the assembling
process, collective action within the assembled gathering, and the
dispersal process \citep[p. 153]{mcphail_myth_91}. Although each of
these three phases is rich and interdependent, the goal is to manage
the complexity of collective action as a whole by focusing on the
recognizable parts of it. Interestingly, the underlying mechanism to
explain collective action is drawn from the \textit{perceptual control
  theory} \citep[Ch. 6]{mcphail_myth_91}. McPhail adapts this theory
to formulate his \textit{sociocybernetics theory} of collective
action. In brief, an individual receives sensory inputs, compares the
input signal to its desired signal,~\footnote{In contrast to
  engineering control systems theory, the desired signal is not
  external, but set by individuals themselves (which is also
  highlighted by the term \textit{cybernetics}).} and adjusts its
behavior in response to the discrepancy. The behavior of individuals affects the ``environment,'' which in turn affects the input signal, thereby completing a loop. An important aspect of this theory is that various external factors (or ``disturbances'') may drive an individual to make different behavioral adjustments at different points in time even if the discrepancy between the input signal and the desired signal remains the same.

\subsubsection*{Additional Notes on Schelling's Models}
Schelling's models assume that individuals behave in a discriminatory
way. For example, individuals are aware, consciously or unconsciously,
about the types of other individuals in their neighborhood and behave
(i.e., stay in the neighborhood or leave) according to their
preference. This is different from organized processes (e.g.,
separation of on-campus residence between graduate and undergraduate
students due to a university's housing policy) or economic reasons
(e.g., segregation between the poor and the rich in many contexts)
\citep{Schelling_dynamic_71,Schelling78}. An example of a segregation
due to individual choice, or ``individually motivated segregation'' as
Schelling puts it \citep[p. 145]{Schelling_dynamic_71}, is the
residential segregation by color in the United States. In fact,
Schelling's models and their analyses expressly focus on this
case. Yet, we can apply Schelling's theory to many other scenarios as
well. This is because it explains, at an abstract level, how
\textit{collective} outcomes are shaped from \textit{individual}
choice. We note here that connecting individual actions to collective outcomes is a mainstream theme of research in collective action.

Schelling introduces two basic models to study the dynamics of
segregation among individuals of two different types
\citep{Schelling_dynamic_71}. In the first model, the \textit{spatial
  proximity model}, individuals are initially positioned in a spatial
configuration (such as a line or a stylized two-dimensional area) and
individuals of the same type share a common ``level of tolerance.''
This level of tolerance quantifies the upper limit on the percentage
of an individual's opposite type in his local neighborhood that he can
put up with. Here, we define an individual's local neighborhood with
respect to the individual's position in the specified spatial
configuration.  In this model, we use a rule of movement for the ``unhappy''
individuals to study the dynamics of segregation. For example, an
individual whose level of tolerance has been exceeded, moves to the
closest location where the tolerance constraint can be satisfied. Assuming
an equal number of individuals of each type and a fixed local
neighborhood size, Schelling first studies how clusters evolve from
the initial configuration of a random placement of the individuals on
a straight line. He then generalizes the experimentation by varying
different model parameters such as neighborhood size, level of
tolerance, and the ratio of individuals of the two types. Notable
findings are that decreasing the local neighborhood size leads to a
decrease in the average cluster size and that for an unequal number of individuals of the two types, decreasing the relative size of the minority leads to an increase in the average minority cluster size. 

Schelling extends this experimentation to a different setting of a
two-dimensional checkerboard. The individuals are randomly distributed
on the squares of the checkerboard, leaving some of the squares
unoccupied. An individual's local neighborhood is defined by the
squares around it and an unhappy individual moves to the ``closest''
unoccupied  square (leaving its original square unoccupied) that can
satisfy its tolerance constraint. In addition to studying clustering
properties by varying different model parameters, two new classes of
individual preferences have been studied: congregationist and
integrationist. In a congregationist preference, an individual only
wants to have at least a certain percentage of neighbors of its own
type and does not care about the presence of individuals of the
opposite type in its neighborhood. Experiments show that even when
each individual is happy being a minority in its neighborhood (e.g.,
having three neighbors of its own type out of eight), the dynamics of
segregation leads to a configuration as if the individuals wished to
be majority in their neighborhoods. In an integrationist preference,
individuals have both an upper and a lower limit on the level of
tolerance. The dynamics is much more complex in this case and leads to clusters of unoccupied squares. 

Schelling's second model, the \textit{bounded-neighborhood model},
concerns one global neighborhood. An individual enters it if it
satisfies its level of tolerance constraint and leaves it
otherwise. The level of tolerance is no longer fixed for each type and
the distribution of tolerances among individuals of each type is
given. The emphasis is on the stability of equilibria as the
distribution of tolerances and the population ratio of the two types
vary. For example, under a certain linear distribution of tolerances
and a population ratio of $2:1$, there exist only two stable
equilibria, each consisting of individuals of one type only; whereas a
mixture of individuals of both types can arise as a stable equilibrium
under a different setting. This model has been adapted to study \textit{tipping} phenomenon, with one notable constraint; that is, the capacity of the neighborhood is fixed. An example of a tipping phenomenon is 
when a neighborhood consisting of only one type of individuals is later inhabited by some individuals of the opposite type and as a result, the entire population of the original type evacuates the neighborhood. An important finding is that in the cases studied, the modal level of tolerance does not correspond to a tipping point.

\section{Experimental Results in Tabular Form}
\label{app:exp}

Table~\ref{tbl:random_graph_eq} shows experimental data of PSNE computation on uniform random directed graphs. Of particular interest is the result that the number of PSNE usually increases when the flip probability $p$ is increased, i.e., when the number of arcs with negative influence factors is increased.

Table~\ref{tbl:random_graph_infl} illustrates the experimental result that the logarithmic-factor approximation algorithm for identifying the most influential individuals performs very well in practice.

Finally, Table~\ref{tbl:LIG_supreme} shows the influence factors and thresholds of the LIG among the U.S. Supreme Court justices, which are learned using the U.S. Supreme Court dataset.

\begin{table*}[t]
\caption{PSNE computation on uniform random directed graphs. Offsets of 95\% confidence intervals are shown in parenthesis.}
\label{tbl:random_graph_eq}
\begin{center}
\begin{tabular}{cccccc}
\multicolumn{1}{c}{$p$}  &\multicolumn{2}{c}{\# of equilibria} 
												&\multicolumn{2}{c}{\# of node visits/equilibrium}
												&\multicolumn{1}{c}{Avg CPU time (sec) for}\\
\multicolumn{1}{c}{}  &\multicolumn{2}{c}{Avg (95\% CI)} 
												&\multicolumn{2}{c}{Avg (95\% CI)}
												&\multicolumn{1}{c}{computing all equilibria}

\\ \hline \\
0.00	&	2.18 & (0.16)	&			35379.73 & (3349.77)	&			1.81	\\
0.125	&	3.72 & (0.50)	&			22756.15 & (2673.56)	&			1.57	\\
0.25	&	13.00 & (1.92) &		9796.30 & (1748.76)	&				1.9	\\
0.375	&	19.42 & (2.88)	&		7380.97 & (1870.11)	&				1.95	\\
0.50	&	14.40 & (2.17)	&		9826.61 & (1696.52)	&				2.04	\\
0.625	&	19.78 & (3.40)	&		8167.60 & (1450.48)	&				2.07	\\
0.75	&	28.76 & (4.34)	&		6335.18 & (1963.11)	&			2.21	\\
0.875	&	67.14 & (9.52)	&		4064.06 & (1539.33)	&				3.27	\\
1.00	&	194.96 & (28.47)	&	1879.45 & (235.90)	&				5.23	
\end{tabular}
\end{center}
\end{table*}

\begin{landscape}
\begin{table*}[t]
\caption{Computation of the most influential nodes: Comparison between the solution given by the ApproximateUniqueHyperedge algorithm and the optimal one in random directed graphs. Offsets of the 95\% confidence intervals are shown in parenthesis.}
\label{tbl:random_graph_infl}
\begin{center}
\begin{tabular}{cccccccc}
\multicolumn{1}{c}{$p$}  &\multicolumn{2}{c}{Approx soln size}
												&\multicolumn{2}{c}{Opt soln size}
												&\multicolumn{1}{c}{Approx = Opt}
												&\multicolumn{1}{c}{Approx $\le$ Opt+1}
												&\multicolumn{1}{c}{Approx $\le$ Opt+2}\\
\multicolumn{1}{c}{}  &\multicolumn{2}{c}{Avg (95\% CI)}
												&\multicolumn{2}{c}{Avg (95\% CI)}
												&\multicolumn{1}{c}{\% of times}
												&\multicolumn{1}{c}{\% of times}
												&\multicolumn{1}{c}{\% of times}
												
\\ \hline \\

0.00			&	1.06 & (0.05)	&	1.06 & (0.05)	&	100	&	100	&	100	\\
0.125	&	1.50 & (0.11)	&	1.48 & (0.10)	&	98	&	100	&	100	\\
0.25	&	2.57 & (0.18)	&	2.14 & (0.11)	&	62	&	96	&	99	\\
0.375	&	2.75 & (0.18)	&	2.35 & (0.13)	&	62	&	98	&	100	\\
0.50		&	2.47 & (0.16)	&	2.09 & (0.10)	&	66	&	96	&	100	\\
0.625	&	2.82 & (0.18)	&	2.31 & (0.13)	&	52	&	97	&	100	\\
0.75	&	3.02 & (0.19)	&	2.48 & (0.14)	&	51	&	95	&	100	\\
0.875	&	3.83 & (0.20)	&	3.04 & (0.15)	&	36	&	86	&	99	\\
1.00			&	4.72 & (0.23)	&	3.95 & (0.18)	&	37	&	87	&	99	
\end{tabular}
\end{center}
\end{table*}
\end{landscape}

\begin{landscape}
\begin{table}[!h]
	\centering
		\begin{tabular}{ c c c c c c c c c c }
			& Scalia & Thomas & Rehnquist & O'Connor & Kennedy & Breyer & Souter & Ginsburg & Stevens \\
			\hline
			\hline
		Scalia & 0.0120  &  0.4282  &  0.0317 &   0.0717  &  0.0721  &  0.0772  & -0.0321  &  0.1362  & -0.1388 	\\
    Thomas & 0.2930  & -0.1020  &  0.1245  & -0.0010  &  0.0183  & -0.1497  &  0.0839  & -0.1311  &  0.0965	\\
    Rehnquist & 0.0671  &  0.1762  & -0.0834  &  0.0973  &  0.1254  &  0.0921  & -0.0861  &  0.1336  & -0.1388 	\\
    O'Connor & 0.0580  &  0.1073  &  0.1045  & -0.2522  & -0.0537  &  0.2313  &  0.0325  & -0.1245  & -0.0359 	\\
    Kennedy & 0.0666  &  0.1236  &  0.1863  & -0.0255  & -0.2634  &  0.0548  & -0.1115  & -0.0149  &  0.1532 	\\
    Breyer & 0.1009  & -0.2191  &  0.0570  &  0.2208  & -0.0061  & -0.0209  &  0.0627  &  0.1102  &  0.2023 	\\
    Souter & 0.0368  &  0.1192  & -0.0476  &  0.0762  & -0.0338  &  0.1429  & -0.0619  &  0.2783  &  0.2034 	\\
    Ginsburg & 0.0779  & -0.1000  &  0.1613  & -0.0962  & -0.0089  &  0.1199  &  0.1999  & -0.0381  &  0.1978 	\\
		Stevens & -0.1379  &  0.1088  & -0.1721  & -0.0568  &  0.1053  &  0.1374  &  0.0932  &  0.1274  & -0.0611 	\\

		\end{tabular}
	\caption{LIG learned from data. Each non-diagonal element in each row represents the influence factor the row player receives from the column players (for example, Justice Scalia's influence factor on Justice Thomas is 0.2930 and that in the opposite direction is 0.4282). The diagonal elements represent thresholds of the corresponding row players.}
	\label{tbl:LIG_supreme}
\end{table}
\end{landscape}

\end{appendix}

\bibliographystyle{model2-names}
\bibliography{influence}

\end{document}